\documentclass{amsart}

\usepackage{color, xcolor, colortbl}
\usepackage{graphicx,epstopdf}
\graphicspath{{./fig/}}
\usepackage{amsmath}
\usepackage{amssymb}
\usepackage{amsthm}
\usepackage{bm}
\usepackage{geometry}
\usepackage{appendix}
\usepackage{multirow}
\usepackage{braket}
\usepackage[english]{babel}
\usepackage{hyperref}
\usepackage[capitalize]{cleveref}
\usepackage{caption}
\usepackage{enumerate}
\usepackage{enumitem}
\crefrangeformat{enumi}{#3#1#4--#5#2#6}

\usepackage{algorithm}
\usepackage{algorithmic}
\usepackage{xspace}
\usepackage[qm]{qcircuit}
\usepackage{subcaption}
\usepackage{multirow}

\usepackage{newfloat}
\usepackage{etoolbox}
\usepackage{dsfont}

\usepackage{tikz-cd}
\usepackage{stmaryrd}
\usepackage[normalem]{ulem}

\newtheorem{OldTheorem}{Theorem}

\newtheorem{theorem}{Theorem}

\newtheorem{lemma}[theorem]{Lemma}
\newtheorem{prop}[theorem]{Proposition}

\newcommand{\bvec}[1]{\mathbf{#1}}
\DeclareMathAlphabet{\mymathbb}{U}{BOONDOX-ds}{m}{n}

\newcommand{\va}{\bvec{a}}
\newcommand{\vb}{\bvec{b}}
\newcommand{\vc}{\bvec{c}}

\newcommand{\ve}{\bvec{e}}

\renewcommand{\Re}{\mathrm{Re}}
\renewcommand{\Im}{\mathrm{Im}}
\newcommand{\I}{\mathrm{i}}

\newcommand{\mc}[1]{\mathcal{#1}}

\newcommand{\wt}[1]{\widetilde{#1}}

\newcommand{\abs}[1]{\left\lvert#1\right\rvert}
\newcommand{\norm}[1]{\left\lVert#1\right\rVert}

\newcommand{\ud}{\,\mathrm{d}}
\newcommand{\Or}{\mathcal{O}}

\newcommand{\NN}{\mathbb{N}}
\newcommand{\RR}{\mathbb{R}}
\newcommand{\CC}{\mathbb{C}}

\def\Z{\ensuremath{\mathbb Z}}
\def\C{\ensuremath{\mathbb C}}
\def\N{\ensuremath{\mathbb N}}
\def\R{\ensuremath{\mathbb R}}
\def\T{\ensuremath{\mathbb T}}
\def\D{\ensuremath{\mathbb D}}

\renewcommand{\tilde}{\widetilde}
\renewcommand{\epsilon}{\varepsilon}
\renewcommand{\phi}{\varphi}
\renewcommand{\bar}{\overline}

\newcommand{\Id}{\operatorname{Id}}
\DeclareMathOperator*{\esssup}{\mathrm{ess\,sup}}
\renewcommand{\equiv}{:=}

\renewcommand{\Im}{\operatorname{Im}}
\renewcommand{\Re}{\operatorname{Re}}

\newcommand{\Hil}{\mathcal{H}}
\newcommand{\dense}{\mathcal{D}}
\newcommand{\ddense}{\mathcal{E}}
\newcommand{\dddense}{\mathcal{F}}

\newcommand{\REV}[1]{#1}

\title{Infinite quantum signal processing for arbitrary Szeg\H o functions}

\author[M. Alexis]{Michel Alexis}
\author[L. Lin]{Lin Lin}
\author[G. Mnatsakanyan]{Gevorg Mnatsakanyan}
\author[C. Thiele]{Christoph Thiele}
\author[J. Wang]{Jiasu Wang}
\thanks{
M. Alexis, G. Mnatsakanyan and C. Thiele acknowledge support by the Deutsche Forschungsgemeinschaft (DFG, German Research Foundation) under Germany's Excellence Strategy -- EXC-2047/1 -- 390685813 as well as CRC 1060.}

\address{Mathematical Institute, 
        University of Bonn,
        Endenicher Allee 60, 53115 Bonn,
        Germany}
\email{alexis@math.uni-bonn.de}
         \email{gevorg@math.uni-bonn.de}
  \email{thiele@math.uni-bonn.de}

\address{Department of Mathematics, University of California, Berkeley, CA 94720, USA}
\email{linlin@math.berkeley.edu}
\email{jiasu@berkeley.edu}

\address{Institute of Mathematics of the National Academy of Sciences of Armenia, Marshal Baghramyan Ave. 24/5, 0019 Yerevan, Armenia}
\email{mnatsakanyan\_g@yahoo.com}

\begin{document}

%\date{August 2023}

\subjclass{68Q12,81P68,34L25,42C99}

%\author{Michel Alexis} 
%        \affiliation{Mathematical Institute, University of Bonn, Endenicher Allee 60, 53115 Bonn, Germany.}
%        \email{alexis@math.uni-bonn.de}
%\author{Lin Lin}
%        \affiliation{Department of Mathematics, University of California, Berkeley,  CA 94720, USA.}
%        \affiliation{Applied Mathematics and Computational Research Division, Lawrence Berkeley National Laboratory, Berkeley, CA 94720, USA}
%        \affiliation{Challenge Institute for Quantum Computation, University of California, Berkeley,  CA 94720, USA}
%        \orcid{0000-0001-6860-9566}
%        \email{linlin@math.berkeley.edu}
%\author{Gevorg Mnatsakanyan} 
%        \affiliation{Mathematical Institute, University of Bonn, Endenicher Allee 60, 53115 Bonn, Germany.}
%        \email{gevorg@math.uni-bonn.de}
%\author{Christoph Thiele} 
%        \affiliation{Mathematical Institute, University of Bonn, Endenicher Allee 60, 53115 Bonn, Germany.}
%        \email{thiele@math.uni-bonn.de}
%\author{Jiasu Wang} 
%        \affiliation{Department of Mathematics, University of California, Berkeley, CA 94720, USA.}
%        \orcid{0000-0002-1321-2649}
%        \email{jiasu@berkeley.edu}

\begin{abstract}
We provide a complete solution to the problem of infinite quantum signal processing for the class of Szeg\H o functions, which are functions that satisfy a logarithmic integrability condition and include almost any function that allows for a quantum signal processing representation. 
We do so by introducing a new algorithm called the Riemann-Hilbert-Weiss algorithm, which can compute any individual phase factor independent of all other phase factors. Our algorithm is also the first provably stable numerical algorithm for computing phase factors of any arbitrary Szeg\H o function. The proof of stability involves solving a Riemann-Hilbert factorization problem in nonlinear Fourier analysis using elements of spectral theory.
\end{abstract}

\maketitle

%\newpage
\tableofcontents

\section{Introduction}

\subsection{Problem setup}
Let $\mathbf{P}$ denote the space of infinite sequences
$\Psi=(\psi_k)_{k\in \N}$ with $\psi_k\in [-\pi/2,\pi/2]$. For $x\in[0,1]$, we define 
\begin{equation}
    W(x) = \begin{pmatrix}
        x & \I \sqrt{1-x^2}\\
        \I \sqrt{1-x^2} & x
    \end{pmatrix},
    \text{ and } Z = \begin{pmatrix}
        1& 0\\
        0 & -1
    \end{pmatrix}.
\end{equation}
Given any $\Psi \in \mathbf{P}$ and $x\in[0,1]$, one can define a sequence of unitary matrices using the following recursive relation:
\begin{equation}
\begin{split}
    &U_0(x,\Psi) = e^{\I \psi_0 Z}
    \\
    &U_d(x,\Psi) = e^{\I \psi_{d} Z}W(x) U_{d-1}(x,\Psi)W(x) e^{\I \psi_d Z}.
\end{split}
\label{eqn:sym_qsp}
\end{equation}
Let $u_d(x,\Psi)$ denote the upper left entry of $U_d(x,\Psi)$. If $\psi_k=0$ for all $k$, then $\Im [u_d(x,\Psi)]=0$, while  $\Re [u_d(x,\Psi)]=T_{2d}(x)$ is the Chebyshev polynomial of the first kind with degree $2d$. For a general $\Psi\in \mathbf{P}$, $u_d(x,\Psi)$ can be verified to be a complex polynomial of $x$ with degree up to $2d$. This procedure of encoding a polynomial as an entry of an ${\rm SU}(2)$ matrix $U_d(x,\Psi)$ is called \emph{quantum signal processing} (QSP)~\cite{LowChuang2017,GilyenSuLowEtAl2019}, and has found numerous applications in quantum computation (see, e.g., \cite{MartynRossiTanEtAl2021,DalzellMcArdleBertaEtAl2023}). 
In \cref{eqn:sym_qsp}, each matrix $U_d(x,\Psi)$ is complex symmetric. %, and \cref{eqn:sym_qsp} is called symmetric quantum signal processing~\cite{DongMengLinEtAl2020,WangDongLin2021}. 
For a polynomial $f(x)$ with degree $2d$, if 
\begin{equation}
\norm{f}_{\infty}:=\esssup\limits_{x\in[0,1]} \abs{f(x)}\le 1,
\label{eqn:f_infty_bound}
\end{equation}
then there exists a sequence $\Psi\in\mathbf{P}$ such that $\Im[u_d(x,\Psi)] = f(x)$~\cite[Theorem 1]{WangDongLin2021} (the sequence is not unique in general). QSP only represents polynomials of definite parity. Without loss of generality, we restrict to even functions $f(x)$ throughout the paper, and a similar treatment can be extended to the odd case.

 The problem of \emph{infinite quantum signal processing} (iQSP) asks whether the QSP representation can be extended to non-polynomial functions $f$ through a product of countably many unitary matrices. The first positive answer to this question is given in \cite{DongLinNiEtAl22}.  Consider $f(x)=\sum_{k\in \NN} c_k T_{2k}(x)$ expressed as an infinite Chebyshev polynomial series. If the $\ell^1$ norm of the Chebyshev coefficient $\norm{\vc}_{1}:=\sum_{k} \abs{c_k}\le 0.903$, then $f(x)$ is continuous, and there exists a sequence $\Psi\in\ell^1(\NN)\subset \mathbf{P}$ such that 
\begin{equation}
\lim_{d\to \infty} \Im[u_d(x,\Psi)] - f(x)=0, \quad \forall x\in [0,1].
\label{eqn:pointwise_conv}
\end{equation}
Under the same assumptions, the fixed point iteration (FPI) algorithm in \cite{DongLinNiEtAl22} is the first provably numerically stable (and perhaps also the simplest) algorithm for computing $\Psi$. A numerically stable algorithm means that the number of bits required in the computation scales as $\mathrm{polylog}(d/\epsilon)$, where $d$ is the polynomial degree and $\epsilon$ is the target precision.

The main questions of this paper are as follows:

{
\em
Let $f$ be an arbitrary even function satisfying the norm constraint in \cref{eqn:f_infty_bound}.
\begin{enumerate}

\item Is there a $\Psi\in \ell^2(\NN)$ such that $\Im[u_d(x,\Psi)]$ converges to $f(x)$, and is this $\Psi$ unique in some sense?

\item Is there a provably numerically stable algorithm to compute $\Psi$, which uses $\mathrm{polylog}(d/\epsilon)$ bits of precision and has a cost of $\mathrm{poly}(d \log(1/\epsilon))$?
\end{enumerate}
}

Recently, Ref.~\cite{AlexisMnatsakanyanThiele2023} observed that after a change of variables, the structure of iQSP can be interpreted using the nonlinear Fourier transform (NLFT) described in \cite{tsai}. This insight opens the door to many new results. A real-valued measurable even functions $f:[0,1]\to [-1,1]$ is called a Szeg\H o function if it satisfies the following Szeg\H o-type condition 
\begin{equation}\label{eq:Szego}
\int_{0}^1 \log|1-f(x)^2| \frac{dx}{\sqrt{1-x^2}} > - \infty \, .
\end{equation}
We use  $\mathbf{S}$ to denote the set of all Szeg\H o functions, and define the norm
\begin{equation}\label{hsnorm}
\left \| f \right \|_{\mathbf{S}} \equiv \left ( \frac{2}{\pi} \int\limits_{0} ^1 \left | f(x) \right |^2 \frac{dx}{\sqrt{1-x^2}} \right )^{\frac{1}{2}} \, ,
\end{equation}
which is finite for all $f \in \mathbf{S}$.
Note that $\norm{\cdot}_{\mathbf{S}}$ induces an inner product, and $\mathbf{S}$ is a subset of a Hilbert space. In particular, for $f(x)=\sum_{k\in \NN} c_k T_{2k}(x)$, we have $\left \| f \right \|_{\mathbf{S}}^2=\abs{c_0}^2+\frac12 \sum_{k>0} \abs{c_k}^2$. So $\left \| f \right \|_{\mathbf{S}}<\infty$ is equivalent to the square summable condition $\norm{\vc}_{2}:=\sqrt{\sum_{k} \abs{c_k}^2}<\infty$.

Ref.~\cite[Theorem 1]{AlexisMnatsakanyanThiele2023} provides a partial answer to Question (1) above, namely, if $f\in \mathbf{S}$ satisfies $\norm{f}_{\infty}<\frac{1}{\sqrt{2}}$,  then there exists a \emph{unique} sequence $\Psi\in\ell^2(\NN)\subset \mathbf{P}$ such that the following equality holds
\begin{equation}
\label{plancherel}
\sum_{k\in \Z}\log(1+\tan^2\psi_{|k|})
=
-\frac{2}{\pi}\int_{0}^1 \log|1-f(x)^2| \frac{dx}{\sqrt{1-x^2}} 
\end{equation}
and
\begin{equation}
\lim_{d\to \infty} \norm{\Im[u_d(x,\Psi)]-f(x)}_{\mathbf{S}}=0.
\label{eqn:L2_conv}
\end{equation}
\cref{plancherel} 
 is a nonlinear version of the Plancherel identity, which connects the $L^2$-norm of the Fourier space representation of a function with its $L^2$ norm in the real space. The phase factors can be interpreted as a nonlinear version of the Fourier coefficients of $f(x)$. In fact, among the sequences $\Psi$ for which \cref{eqn:L2_conv} holds, 
 the sequence $\Psi$ constructed in \cite[Theorem 1]{AlexisMnatsakanyanThiele2023}
 is the unique sequence where the quantity
  \[
 \sum_{k\in \Z}\log(1+\tan^2\psi_{|k|}) 
 \]
takes the minimum 
  \[
 -\frac{2}{\pi}\int_{0}^1 \log|1-f(x)^2| \frac{dx}{\sqrt{1-x^2}}.
 \]See \cite[Lemma 3.1]{tsai}, as well as \cref{sec:maximal} below for further details.

Note that neither $\norm{\vc}_{1}\le 0.903$ nor $\norm{f}_{\infty}<\frac{1}{\sqrt{2}}$ is  stronger than the other (just consider $f(x)=\frac12 \cos(100 x)$, and $f(x)=0.8x$, respectively). Furthermore, the techniques developed in \cite{DongLinNiEtAl22,AlexisMnatsakanyanThiele2023} encounter significant  difficulties towards representing all functions $f\in\mathbf{S}$.

\subsection{Main results}
In this work, we provide positive answers to both questions above, which constitute a complete solution of the iQSP problem.
For each $\eta \in (0,1)$, we define  
\begin{equation}
\mathbf{S}_{\eta}=\set{f\in \mathbf{S}| \norm{f}_{\infty} \le 1-\eta}.
\end{equation}

\begin{theorem}\label{thm:main}
For each $f \in \mathbf{S}$, there exists a unique 
sequence $\Psi \in \mathbf{P}$ such that both the $L^2$ convergence criterion in \cref{eqn:L2_conv} and the nonlinear Plancherel identity in \cref{plancherel} hold.

Furthermore, given $0 < \eta < \frac{1}{2}$, for two functions $f, f' \in \mathbf{S}_{\eta}$ with corresponding sequence $\Psi, \Psi '$ as above, we have the Lipschitz bound
\begin{equation}\label{eq:Lip_bd_QSP}
\left \| \Psi - \Psi ' \right \|_{\infty} \leq 1.6 \eta^{- 3} \left \| f - f' \right \|_{\mathbf{S}} \, .
\end{equation}
\end{theorem}

% \LL{ Is it natural to bound the left hand side by the $\infty$ norm?} 

A fundamental result in Fourier analysis is that the mapping from the function to its Fourier / Chebyshev coefficients is  a linear functional, and hence each Fourier / Chebyshev coefficient can be evaluated independently from the others using a single inner product. In the case when $f(x)$ is even, we can compute the Chebyshev coefficients explicitly as  $c_k = \frac{2(2-\delta_{k,0})}{\pi} \int\limits_{0} ^1 f(x) T_{2k}(x) \frac{dx}{\sqrt{1-x^2}}$.

Can the phase factors $\psi_k$ be evaluated independently as well? The proof of \cref{thm:main} provides several useful tools for characterizing the phase sequence $\Psi$. From \cref{lem:f_to_phase_factor} below, we can compute an individual phase factor via the formula 
\begin{equation}
\psi_k = \arctan \frac{(B_k z^{-k}) (0)}{i A_k (\infty)} \, ,
\end{equation}
where $A_k$ is a Laurent series on $\C$, and $B_k z^{-k}$ is a Taylor series on $\C$, both depending on $f$, and the pair $(A_k,B_k)$ is the unique solution to a linear system
(see \cref{eq:linear_system_k}). By solving this linear system, we obtain an algorithm that is able to compute each individual phase factor $\psi_k$ independently.  
 This is in sharp contrast to \emph{all} algorithms in the literature, where phase factors need to be computed in an interdependent fashion. 

We introduce this new algorithm, dubbed the Riemann-Hilbert-Weiss algorithm, in  \cref{sec:main_algorithm}.
 Applying this algorithm to compute all $d$ phase factors, we obtain the first provably numerically stable algorithm to evaluate the phase factor sequence for any $f\in \mathbf{S}_\eta$  (without losing generality, we assume that $\eta<\frac{1}{2}$). An even polynomial in one variable 
 is a linear combination of monomials with even power.
For real functions $g, h: \R\rightarrow\R$, we write $g=\Or (h)$ if there exists $c>0$, such that $|g(\tau)|\le c|h(\tau)|$ for all $\tau\in\R$. Again a similar result can be obtained for odd polynomials and we omit the discussion here.

\begin{theorem}\label{thm:main_alg}
   Let $0<\epsilon<1$, $0<\eta <\frac{1}{2}$ and let $k$ and $d$ be integers satisfying $d\ge 1$ and $0\le k\le d$. There exists a deterministic algorithm to compute the $k$-th phase factor $\psi_k$ for any even input polynomial $f \in \mathbf{S}_{\eta}$ with degree $2d$, 
   to precision $\epsilon$ with a computational cost of \REV{$\Or\left(d^3 + \frac{d}{\eta}\log^2(d/(\eta\epsilon))\right)$} and using $\Or(\log(d/(\eta \epsilon)))$ bits. The computational cost of $\Or(d^3)$ arises from solving a linear system of size $\Or(d)$. To determine all phase factors, the algorithm solves $\Or(d)$ such linear systems, resulting in a cumulative cost of \REV{$\Or\left(d^4 + \frac{d}{\eta}\log^2(d/(\eta\epsilon))\right)$}, and the bit requirement remains $\Or(\log(d/(\eta \epsilon)))$.\end{theorem}

It is worth noting that the main purpose of \cref{thm:main_alg}
is to provide a numerically stable, polynomial scaling that works for arbitrarily small $\eta$. Due to the highly structured form of the linear system, the quartic power in $d$ may be improved using fast linear solvers for Hankel and Toeplitz matrices~\cite{GohbergKailathOlshevsky1995,ChandrasekaranSayed1996,ChandrasekaranGuSunEtAl2008}. However, the numerical stability of these fast algorithms should be carefully investigated.

% We use $\widetilde \Or$ to suppress logarithmic factors in the asymptotic expression, i.e., $g=\widetilde \Or (h)$ if $g=\Or\left(h\poly(\log h)\right)$.We write $g=\Omega(h)$ if $h=\Or (g)$, and $g=\widetilde \Omega(h)$ if $h=\widetilde \Or (g)$.

\subsection{Related works}

For a complex polynomial $p$ of degree $d$, where $d$ can be either even or odd, Ref.~\cite{LowChuang2017} provides a nonconstructive proof of the existence of a finite phase factor sequence $\Psi\in\mathbb{R}^{d+1}$ corresponding to $p$. Representing a real polynomial $p$ of definite parity in QSP is considerably more challenging, as it requires finding a complementary polynomial $q$ that is not known \emph{a priori}. The problem of finding a complementary polynomial was solved constructively in Refs.~\cite{GilyenSuLowEtAl2019,Haah2019}. These constructive methods require finding all roots of the Laurent polynomial $1-p((z+z^{-1})/2)^2$. Neither the complementary polynomial $q$ nor the phase sequence is uniquely defined after such a root-finding process. 
The Prony method~\cite{Ying2022} provided a construction of a complementary polynomial using a contour integral approach without root-finding. Ref. \cite{AlexisMnatsakanyanThiele2023} showed a different way of directly constructing a complementary polynomial, using a method that is referred to as the Weiss algorithm in this paper. Given the complementary polynomial, Ref. \cite{ChaoDingGilyenEtAl2020} proposed a ``halving'' algorithm which can find phase factors by solving linear systems of equations. However, there was no upper bound of the condition number of such linear systems.

Iterative methods~\cite{DongMengWhaleyEtAl2021,WangDongLin2021,DongLinNiEtAl22,DongLinNiEtAl23} tackle the problem in a very different way, which directly finds the inverse to the map $\Psi\mapsto p$ by choosing a symmetric phase factor sequence  as in \cref{eqn:sym_qsp}. Ref.~\cite{WangDongLin2021} identified a particular branch of the solution, called the maximal solution, which leads to the first infinite QSP representation~\cite{DongLinNiEtAl22}. The phase sequence obtained in Ref. \cite{AlexisMnatsakanyanThiele2023} using nonlinear Fourier analysis coincides with the choice of the maximal solution, see \cref{sec:maximal}. This work shows that the maximal solution is well-defined for almost all polynomials which admits a QSP representation, and is the unique solution that satisfies a Plancherel identity. 

From an algorithmic point of view, finding the phase factor sequence was considered a major computational bottleneck in early applications of QSP for quantum computation even when the polynomial degree is less than $50$~\cite{ChildsMaslovNamEtAl2018}. The root-finding based method is a significant progress towards systematically computing phase factors, but the algorithm requires $\Or(d\log(d/\epsilon))$ bits of precision~\cite{Haah2019} and is not numerically stable. Substantial algorithmic improvements have been made in recent years \cite{ChaoDingGilyenEtAl2020,DongMengWhaleyEtAl2021,Ying2022,MotlaghWiebe2023,DongLinNiEtAl23}, which can accurately compute the phase sequence for polynomials of degree larger than $10^4$, using only the standard double precision arithmetics. As $d\to \infty$, the only provably numerically stable algorithm  so far is the fixed point iteration (FPI) algorithm in \cite{DongLinNiEtAl22}, which requires the $\ell^1$ norm of the Chebyshev coefficients to be bounded by a constant. This work shows that the phase sequence can be computed with a numerically stable algorithm for almost all functions which admit a QSP representation. The Weiss algorithm in \cite{AlexisMnatsakanyanThiele2023} as well as this work shares similarities with the contour integral based method for finding a complementary polynomial~\cite{BerntsonSuenderhauf2024}.%}\ma{extremal complimentary polynomials?}

\subsection{Discussion and open questions}

Recently there have been a number of generalizations of QSP, including generalized QSP~\cite{MotlaghWiebe2023} which replaces the $Z$ rotation $e^{i\psi Z}$ by a more general parameterized {\rm SU}(2) rotation, QSP for {\rm SU}(1,1) matrices~\cite{RossiBastidasMunroEtAl2023}, and multi-variable QSP for commuting matrices~\cite{RossiChuang2022,RossiCeroniChuang2023,NemethKoeverKulcsarEtAl2023}. We note that while the proposed QSP algorithm in ${\rm SU}(1,1)$ and the ${\rm SU}(1,1)$ model of the nonlinear Fourier transform seem to have the roles of the phase factors (nonlinear Fourier coefficients) and spatial variable $\beta$ (or $z$) swapped, we anticipate that the nonlinear Fourier analysis perspective \cite{TaoThiele2012} may be fruitful in these settings as well.

\cref{thm:main} proves the existence and uniqueness of the phase sequence for any $f\in \mathbf{S}$. Furthermore, the sensitivity result of \cref{eq:Lip_bd_QSP}, which requires  
$\norm{f}_\infty \le  1-\eta$ for some $\eta>0$, agrees with the numerical observation that the Jacobian of the mapping $\Psi\mapsto f$ becomes singular as $\eta \to 0$~\cite{DongMengWhaleyEtAl2021,DongLinNiEtAl23}. 

Our algorithm in \cref{thm:main_alg} allows us to compute all phase factors independently and in parallel.  This completely circumvents the error accumulation issue associated with the ``layer stripping'' method introduced in \cite{GilyenSuLowEtAl2019}, which can be another source of numerical instability and requires at least in theory, high precision arithmetics   \cite[Eq. (38)]{Haah2019}. The computation of each phase factor requires solving a Riemann-Hilbert factorization problem via a linear system of equations, leading to a $\Or(d^3)$ cost per phase factor. Consequently, the leading term of the computational cost in \cref{thm:main_alg} is $\Or(d^4)$. On the other hand, the FPI algorithm \cite{DongLinNiEtAl22} only requires a much lower cost of $\Or(d^2\log(1/\epsilon))$ upon convergence. The lower bound for evaluating phase factors for all functions $f\in\mathbf{S}$ is unclear either at this point.

%Thus, developing a numerically stable algorithm that operates with a cost of $\Or(d^2\log(1/\epsilon))$ for all functions $f \in \mathbf{S}$ remains an open question. One possible approach to achieve this goal is to reexamine the stability of the layer stripping method and investigate whether it can become stable under specific parameter regimes. 

Ref.~\cite{DongMengWhaleyEtAl2021} observed a close connection between the decay of the Chebyshev coefficients of $f$ and the decay of the phase factors. This decay relationship was rigorously established under $\ell^1$ conditions in \cite[Theorem 4]{DongLinNiEtAl22}. \REV{We anticipate that this relationship may be generalized to all functions $f\in\mathbf{S}$.}

% \LL{ Organizational stuff here, or right after the numerics} 
The paper is organized as follows. The preliminaries are given in \cref{sec:prelim}. A concise overview of Hardy functions is presented in \cref{sec:hardy_space}. \cref{sec:nlft} introduces the nonlinear Fourier analysis and its relation to QSP. Then, in \cref{sec:construct_a_given_b} and \cref{sec:existing_rhf}, we review some results of \cite{AlexisMnatsakanyanThiele2023} which are critical to this work. \cref{sec:main_algorithm} introduces the Riemann-Hilbert-Weiss algorithm, including  numerical results demonstrating the performance of the algorithm in \cref{sec:numerical_experiments}. The proof of \cref{thm:main} is provided in \cref{sec:main_proof}. The complexity of our algorithm is analyzed in \cref{sec:complexity_analysis}, which proves \cref{thm:main_alg}.

\subsection*{Acknowledgment}

M. A., G. M. and C. T.  acknowledge support by the Deutsche Forschungsgemeinschaft (DFG, German Research Foundation) under Germany's Excellence Strategy -- EXC-2047/1 -- 390685813 as well as CRC 1060.  L. L. and J. W. acknowledge support by the Challenge Institute for Quantum Computation (CIQC) funded by National Science Foundation (NSF) through grant number OMA-2016245, and the Simons Investigator award through Grant Number 825053. We thank Andras Gily\'en, Hongkang Ni, and Zane Rossi for helpful discussions.

\vspace{1em}
\noindent\textit{Note:} Towards the completion of this work, we learned that \REV{Hongkang Ni and Lexing Ying were developing a new numerically stable method based on our findings, which can reduce the cost of computing all phase factors to $\wt{\Or}(d^2\log(1/\epsilon))$.} 

\section{Preliminaries}\label{sec:prelim}

\subsection{Hardy functions} 
%and Szeg\H o projection}
\label{sec:hardy_space}
Hardy spaces, denoted by $H^p$, are function spaces that arise in complex analysis and harmonic analysis with different definitions depending on the domain considered. In this paper, we only consider the Hardy spaces on the unit disk $\mathbb{D}$ which are defined as follows.

On the open unit disk $\mathbb{D}:=\{z\in \C:\abs{z}\le 1\}$, a function $g(z)$ is in the Hardy space $H^p(\mathbb{D})$ for $1\le p <\infty$ if $g(z)$ is holomorphic on $\mathbb{D}$ and
\begin{equation}
\sup_{0 \leq r < 1} \int_{0}^{2\pi} |g(re^{i\theta})|^p \ud \theta < \infty \, .
\end{equation}
Similarly, $g \in H^{\infty} (\mathbb{D})$ if 
\[
\sup_{0 \leq r < 1}  \sup\limits_{\theta} |g(re^{i\theta})| < \infty
\]
%\[
%\sup_{0 \leq r < 1}  \operatorname{ess~sup}\limits_{\theta} %|g(re^{i\theta})| < \infty
%\]
Functions in $H^p(\mathbb{D})$ have radial limits almost everywhere on the unit circle $\mathbb{T}:=\{z\in \C:\abs{z}=1\}$, and 
%similarly to the upper half-plane case,
these boundary values determine the function uniquely \REV{\cite[Chapter 2, Thm. 3.1 and Cor. 3.2]{garnett}}. Thus when we say a function $g \in L^p (\T)$ belongs to $H^p (\D)$, we mean $g$ coincides with the boundary values of a unique $H^p (\D)$ function, which we also denote by $g$. 
%Thus by 
By the mean value property for harmonic functions, for every function $g \in H^p (\D)$ we have 
\[
g(0) = \int\limits_{\T} g :=\frac{1}{2\pi}\int_{0}^{2\pi}g(e^{i\theta}) d\theta \, .
\]

For a subset $\Omega$ of the Riemann sphere $\CC\cup \{\infty\}$, the reflected set is
\begin{equation}
\Omega^*=\set{\overline{z^{-1}}: z\in \Omega}.
\end{equation}
If $g$ is a function on $\Omega$, then 
\begin{equation}
g^*(z):=\overline{g(\overline{z^{-1}})}
\end{equation}
is a function on $\Omega^*$. Hence $g$ is analytic on $\mathbb{D}^*$ if and only if $g^*$ is analytic on $\mathbb{D}$.
The anti-Hardy space $H^p(\mathbb{D}^*)$ consists of functions $g$ satisfying $g^*\in H^p(\mathbb{D})$. And again, we have the identity
\[
g(\infty) = \overline{g^*(0)} = \overline{\int\limits_{\T} g^*} = \int\limits_{\T} g  \, .
\]

Let $P_{\D}$ and $P_{\D^*}$ denote the $L^2 (\T)$ orthogonal projections onto $H^2 (\D)$ and $H^2 (\D^*)$, respectively. If $g \in L^2 (\T)$, then we can write $g (e^{\I \theta}) = \sum\limits_{n =-\infty}^{\infty} c_n e^{\I n \theta}$, in which case we have the explicit formulas
\[
P_{\D} g = \sum\limits_{n \geq 0} c_n e^{\I n \theta} \, , \qquad P_{\D ^*} g = \sum\limits_{n \leq 0} c_n e^{\I n \theta} \, .
\]

If $g$ is a periodic smooth function on $\mathbb{T}$, define the Hilbert transform
\begin{equation}
\mathrm{H}(g)(x) := \frac{1}{\pi} \operatorname{p.v.} \int_{0}^{2\pi} g(e^{i\theta}) \frac12 \cot\left(\frac{x - \theta}{2}\right) \ud \theta.
\end{equation}
Direct calculation shows that
\begin{equation}\label{eq:Hilbert_tranform_rule}
\begin{split}
&\mathrm{H}(z^n)=-\I z^n, \quad n\in \mathbb{N}_+,\\
&\mathrm{H}(z^{-n})=\I z^{-n}, \quad n\in \mathbb{N}_+.
\end{split}
\end{equation}
More generally, \cite[Chapter 3, Theorem 2.1 and Theorem 2.3]{garnett} shows that $\mathrm{H}$ extends to a bounded operator on $L^p (\T)$ for $1<p<\infty$, and is a bounded map from $L^1 (\T)$ to $L^{1, \infty} (\T)$.
Given $g\in H^p(\mathbb{D})$, let us denote $g(e^{i\theta})=\sum_{n\ge 0} c_n e^{in \theta}$. If $c_0=0$, then
\begin{equation}\label{eqn:hilbert_HpD}
\mathrm{H}(g)(z)=-\I g(z).
\end{equation}
Note that the output of $\mathrm{H}$ is a real-valued function if the input function is real-valued, and for $g\in H^p(\mathbb{D})$
\begin{equation}
\mathrm{H}(\Re g)(z)=\Im g(z), \quad \mathrm{H}(\Im g)(z)=-\Re g(z).
\end{equation}
Similarly, if $g\in H^p(\mathbb{D}^*)$ and $c_0=0$, then
\begin{equation}\label{eqn:hilbert_HpDstar}
\mathrm{H}(g)(z)=\I g(z), \quad \mathrm{H}(\Re g)(z)=-\Im g(z), \quad \mathrm{H}(\Im g)(z)=\Re g(z).
\end{equation}
These relations can also be proved using the more powerful Sokhotski--Plemelj relation.

If $g\in H^p(\mathbb{D})$ has modulus $1$ a.e.\ on $\mathbb{T}$, then $g$ is called an inner function. Two functions $a, b \in H^2 (\D)$ are said to have no common inner factor if for every inner function $g$, both $a/g$ and $b/g$ are $H^2 (\D)$ functions if and only if $g$ is constant. A function $g \in L^{\infty} (\T)$ is called outer if $\log |g| \in L^1 (\T)$ and $g=e^{G}$ where $G = \log |g| + i \mathrm{H}(\log |g|)$. Note that if $g$ is outer, then $G$ has an analytic extension to $\D$ \REV{\cite[first page of Chapter 3]{garnett}, whose} real part \REV{is at most $\sup\limits_{z \in \T} \log |g(z)|$}, and hence $g \in H^{\infty} (\D)$.

\subsection{Nonlinear Fourier analysis}\label{sec:nlft}

Given a compactly supported sequence $F = (F_n)_{n \in \Z}$ of complex numbers, for each $n \in \Z$ define a pair of Laurent polynomials $(a_n (z), b_n(z))$ via the recurrence relation
\begin{equation} \label{eq:defn_NLFT_gen}
\begin{pmatrix}
   a_n (z) & b_n (z) \\
   -b_n ^* (z) & a_n ^* (z)
\end{pmatrix} = \begin{pmatrix}
   a_{n-1} (z) & b_{n-1} (z) \\
   -b_{n-1} ^* (z) & a_{n-1} ^* (z)
\end{pmatrix} \frac{1}{\sqrt{1+|F_n|^2}} \begin{pmatrix}
   1 & F_n z^n \\
   -\overline{F_n} z^{-n} & 1
\end{pmatrix}  
\end{equation}
with the initial condition
\begin{equation}
\begin{pmatrix}\label{eq:init_condition_NLFT}
    a_{-\infty} (z) & b_{-\infty} (z) \\
    -b ^* _{-\infty} (z) & a_{-\infty} ^* (z)
\end{pmatrix} = \begin{pmatrix}
    1 & 0 \\ 0 & 1
\end{pmatrix} \, .
\end{equation}
Here, $a^*(z) := \overline{a(\overline{z^{-1}})}$ for any function $a$. Note that because the sequence $F$ has compact support, $(a_{-\infty}, b_{-\infty}) = (a_n, b_n)$ for all $n$ to the left of the support of $F$. The nonlinear Fourier series, or nonlinear Fourier transform (NLFT), of the sequence $F$ is defined as the pair of Laurent polynomials 
\[
(a(z), b(z)) := (a_{\infty} (z), b_{\infty} (z)) \, ,
\]
where again $(a_{\infty}, b_{\infty}) = (a_n, b_n)$ for all $n$ to the right of the support of $F$. 
Because the initial condition is given by the identity matrix, and then in the recurrence relation we only multiply by matrices with determinant $1$, the resulting Laurent polynomials $a, b$ then satisfy the determinant condition
\begin{equation}\label{eq:det_cond_1}
a(z) a^* (z) + b(z) b^* (z) = 1 \, .
\end{equation}

For convenience, for any functions $a, b \in L^2 (\T)$ satisfying \cref{eq:det_cond_1} for almost every $z \in \T$, we identify the pair $(a,b)$ with the matrix
\begin{equation}\label{nlftg}
  G(z):=  \begin{pmatrix}
    a (z) & b(z) \\
    - b^*  (z) & a^*(z)
\end{pmatrix} \, .
\end{equation}

The matrix product of two matrices $(a,b)$ and $(c,d)$ can be expressed concisely as 
\begin{equation}
    (a,b)(c,d) = (ac-bd^*, ad+bc^*) \, ,
\end{equation}
and similarly for the inverse,
\[
(a,b)^{-1} = (a^*, -b) \, .
\]

The nonlinear Fourier transform can be extended to  square-integrable sequences supported on the half-line $\ell^2(\NN)=:\ell^2 ([0, \infty))$ 
 \cite{tsai, AlexisMnatsakanyanThiele2023}. In the latter case, $(a(z), b(z))$ may no longer be a pair of Laurent polynomials, but for every $k \in \Z$, \cite{tsai} characterized the image  of $\ell^2 ([k, \infty))$ under the NLFT denoted by $\mathbf{H}_{\geq k}$. Define 
$\bar{\mathbf{H}_{\geq k}}$ to be the space of pairs $(a,b) \in H^2 (\D^*) \times z^k H^2 (\D)$ on $\T$ such that $a (\infty) > 0$ and 
\[
a a^* + b b^* = 1
\]
almost everywhere on $\T$, and define $\mathbf{H}_{\geq k}$ to be the set of pairs $(a,b) \in \bar{\mathbf{H}_{\geq k}}$ for which $a^*$ and $b$ share no common inner factor. Then \cite{tsai} showed that the NLFT is a homeomorphism from $\ell^2 ( [k, \infty))$ onto 
$\mathbf{H}_{\geq k}$ with an appropriate topology, and similarly the NLFT is a homeomorphism from $\ell^2 ((-\infty, k])$ onto 
\[
\mathbf{H}_{\leq k} := \{ (a,b) ~:~ (a, b^*) \in \mathbf{H}_{\geq -k}  \} \, , 
\]
and one may also define $\overline{\mathbf{H}_{\leq k}}$ similarly.
See also \cite[Section 6]{AlexisMnatsakanyanThiele2023} for a summary of these results. 

One can also extend the nonlinear Fourier transform to sequences in $\ell^2 (\Z)$ as follows: given a sequence $F$ in $\ell^2(\Z)$, split it as $F_-+F_+$, where $F_-$ is supported in $(-\infty, -1] $  and $F_+$ is supported in $[0,\infty)$. Then we define the nonlinear Fourier transform of $F$ to be the pair
\begin{equation}\label{eq:factorization_explain}
(a,b):= (a_- , b_{-}) (a_+, b_+)
\end{equation}
where $(a_- , b_{-})$ denotes the NLFT of $F_{-}$ and $(a_+, b_+)$ denotes the NLFT of $F_+$.  
The problem of finding factors $(a_{-}, b_{-})$ and $(a_+, b_+)$ as in \cref{eq:factorization_explain} is known as a Riemann-Hilbert factorization problem \cite[Lecture 3, p.31]{TaoThiele2012}.

% are not symmetric. \LL{ What does "symmetric" mean in this context?} 

% \JS{This is actually consistent with that the phase factors are not unique without symmetry constraint. The nonlinear Fourier transform can also be related to Quantum signal processing, provided that $He^{\I \phi_n Z} H = \frac{1}{\sqrt{1+|F_n|^2}} \begin{pmatrix}
%    1 & F_n z^n \\
%    -\overline{F_n} z^{-n} & 1
% \end{pmatrix}  $ and $HWH = \begin{pmatrix}e^{\I \theta} & 0 \\ 0 & e^{-\I\theta} \end{pmatrix}$. Why we need the symmetry constraint is still because of the uniqueness.}

The nonlinear Fourier transform and symmetric quantum signal processing are related by the following Lemma.
Define the Hadamard gate
\[ {\rm Had}:=\frac{1}{\sqrt{2}} \begin{pmatrix}
            1 & 1\\
            1 & -1
        \end{pmatrix}\, .\]

% More precisely, $G_d(z)$ possesses the form 
% \begin{equation}\label{eq:NLFT_ab}
%     G_d(z) = : \begin{pmatrix}
%         a(z) & b(z) \\
%         -b^*(z) & a^*(z)
%     \end{pmatrix}.
% \end{equation}

\begin{lemma}[{\cite[Lemma 2]{AlexisMnatsakanyanThiele2023}}]\label{lem:change_of_vars}
 Let $\Psi\in \mathbf{P}$
  and consider $U_d(x,\Psi)$ and $u_d(x,\Psi)$ as defined near \cref{eqn:sym_qsp}.
Define for $n \in \mathbb{Z}$
\begin{equation}\label{eqn:F_psi_mapping}
    F_n = \I \tan(\psi_{\abs{n}})\, ,
\end{equation} 
and let $G_d(z)$ denote the nonlinear Fourier series 
in the form in \cref{nlftg} 
of the truncated sequence 
\[\left(F_n 1_{\{-d\leq n\leq d\}}\right)_{n\in \Z}\, .\]

    Then the following relations hold for every $ d\geq 0$ and every $\theta\in [0,\frac \pi 2]$,
    \begin{equation}\label{eqn:QSP_NLFT_connection} 
         U_d( \cos(\theta), \Psi)  = {\rm Had} \begin{pmatrix}
            e^{\I d\theta} & 0 \\
            0 & e^{-\I d\theta}
        \end{pmatrix}G_d (e^{2\I\theta}) \begin{pmatrix}
            e^{\I d\theta} & 0 \\
            0 & e^{-\I d\theta}
        \end{pmatrix}  {\rm Had}
    \end{equation}
    and in particular, if the upper right entry of $G_d$ is $b$, 
    \begin{equation}\label{urelateb}
        \I \Im [u_d(\cos(\theta),\Psi)] = b(e^{2\I\theta}).
    \end{equation}
\end{lemma}

Following \cref{eqn:QSP_NLFT_connection} and \cref{urelateb},
we will often relate  $x\in [0,1]$ with the unique 
$\theta \in [0, \frac{\pi}{2}]$ and $z$ in the upper half of $\T$
so that 
 \begin{equation}\label{eq:change_of_variables}
         \cos(\theta) = x, \quad z= e^{2\I \theta}.
    \end{equation}
For example, we write for \cref{urelateb}
   \begin{equation}\label{urelatebxz}
        \I \Im [u_d(x,\Psi)] = b(z).
    \end{equation}
The Lemma enables us to convert the original problem of finding for given $f$ some phase factors $\mathbf{P}$ with
\begin{equation}
    \Im[u_d(x,\Psi)] = f(x)
\end{equation}
into that of determining the infinite sequence $(F_n)_{n\in\Z}$ for some data $(a,b)$ with $\I f(x)=b(z)$. 
Note that if $f$ is an even polynomial of degree $2d$ in $x$, then $b$ is a Laurent polynomial of degree $d$ in $z$ satisfying the symmetry $b(z)=b(z^{-1})$. Moreover, $b$ is pure imaginary on $\T$
as $f$ is real and we have $b^*=-b$.

We have the flexibility to select for given $b$ any $a$ for which $(a,b)$ is the nonlinear Fourier transform of a sequence. A good choice of $a$ is addressed in the following subsection.

\subsection{Weiss algorithm for constructing \texorpdfstring{$a$}{a} given \texorpdfstring{$b$}{b}}\label{sec:construct_a_given_b}
Given $b$ with infinity norm bounded by $1$, there are multiple choices of $a$ for which $(a,b)$ is a NLFT. But \cite[Theorem 10]{AlexisMnatsakanyanThiele2023} provides a way of constructing a specific and convenient choice of $a(z)$ given $b(z)$. Initially, \cite[Theorem 10]{AlexisMnatsakanyanThiele2023} was confined to functions $b$ with 
\[
\|b\|_{\infty}:=\esssup\limits_{z \in \T} |b(z)|  < \frac{1}{\sqrt{2}} \, .
\]
However, the proof of the theorem can be extended to accommodate any function $b$ satisfying $\|b\|_{\infty} \leq 1$ and the Szeg\H o condition on the unit circle
\begin{equation}
\int\limits_{\T} \log (1- |b(z)|^2) > - \infty \, .
\end{equation}
Namely in our extension \cref{thm:construct_a} below, let $\mathbf{B}$ be the set of pairs 
 of measurable functions $(a,b)$ on $\mathbb{T}$ for which 
 $a^*$ is outer with $a^* (0) > 0$, 
and 
\begin{equation}\label{eq:det_torus_gen} 
aa^*+bb^*=1
\end{equation}
almost everywhere on $\T$.

\begin{theorem}[Extension of {\cite[Theorem 10]{AlexisMnatsakanyanThiele2023}}]\label{thm:construct_a}
For each complex valued measurable function \( b \) on \( \mathbb{T} \) with $\norm{b}_{\infty} \leq 1$,
if $b$ satisfies the Szeg\H o condition 
\begin{equation}\label{eq:Szego_cond_b}
\int\limits_{\T} \log (1 -|b(z)|^2) > -\infty \, ,
\end{equation}
then there is a unique measurable function \( a \) on \( \mathbb{T} \) such that \( (a, b) \in \mathbf{B} \).
\end{theorem}
\begin{proof}
The idea is the same as that of \cite[Theorem 10]{AlexisMnatsakanyanThiele2023}. Existence follows by  explicitly constructing a function $a$ satisfying \( (a, b) \in \mathbf{B} \). Namely, define
\begin{equation}
R(z):=\log\sqrt{1-\abs{b(z)}^2}
\end{equation}
for almost every $z\in \T$. Moreover, $R(z) \in L^1(\mathbb{T})$. Then $G:=R-\I \mathrm{H}(R)$ is well-defined a.e.\ because the Hilbert transform $\mathrm{H}$ maps the $L^1 (\T)$ function $R$ to a $L^{1, \infty} (\T)$ function. Furthermore, $G$ has an analytic extension to $\mathbb{D}^*$, and so does 
\begin{equation}\label{eq:a_G}
a(z):=e^{G(z)} \, .
\end{equation}
\REV{Note that since conjugate functions vanish at the origin, then $G^*$ has vanishing imaginary part at the origin \cite[first page of Chapter 3]{garnett}. It follows that  $G^* (0)$ is the mean of $R = \log \sqrt{1 - |b(z)|^2}$ on $\T$, which is real-valued and finite by $\|b\|_{\infty} \leq 1$ and the Szeg\H o condition  \cref{eq:Szego_cond_b}. 
It follows that 
\[
a^*(0) = e^{G^* (0)}>0 \, .
\]}

Notice that $\log |a^*|$ is exactly $R(z)$ on $\T$, which belongs to $L^1(\T)$. Besides, the values of $a^*$ on $\T$ are the almost everywhere defined nontangential limits of the analytic function $e^{G^*(z)}$. So this construction of $a(z)$ not only meets the requirement $aa^* + bb^* =1$, but also satisfies that $a^*$ is outer with $a^*(0)>0$. Hence $(a,b) \in \mathbf{B}$.

As for the uniqueness, one can prove uniqueness by contradiction and assume that $\Tilde{a}$ is another function as claimed in the theorem. Notice that $a\Tilde{a}^{-1}$ and its reciprocal are analytic in the disc and have modulus $1$ a.e.\ on $\mathbb{T}$. By the maximum principle, this means $a\Tilde{a}^{-1}$ equals a unimodular constant. Given $a^*(0)$ and $\Tilde{a}^*(0)$ are both positive, the constant is $1$, that is, $a = \Tilde{a}$.
\end{proof}

This constructive proof leads to a numerical method of constructing $\frac{b}{a}$ given $b$, which is an important component of our \cref{alg:RHW_alg} below. We call this method of construction the Weiss algorithm, as it follows the same idea in the Guido and Mary Weiss algorithm \cite{MaryWeiss62}. A more detailed discussion about the Weiss algorithm is provided in \cref{sec:Weiss_alg}. 

The connection between the solution of the Weiss algorithm and the maximal solution as defined in \cite{WangDongLin2021} is discussed in \cref{sec:maximal}.

\subsection{Existing results on Riemann-Hilbert factorization}\label{sec:existing_rhf}
By \cref{lem:change_of_vars} and \cref{thm:construct_a}, the problem of showing there exists a phase factor sequence $\{\psi_k \}$ associated to a given target function $f$ is reduced to the following: given $(a,b) \in \mathbf{B}$, can we show $(a,b)$ is an NLFT? Under the additional assumption
 \begin{equation}\label{eq:f_bded}
\|f\|_{\infty} = \|b\|_{\infty} < \frac{1}{\sqrt{2}}, 
 \end{equation}
  \cite[Theorem 11]{AlexisMnatsakanyanThiele2023} provides a positive answer, whose proof we summarize here. In what follows, we will often write elements of $L^2 (\T) \times L^2 (\T)$ as column vectors, so that we may then write operators acting on such elements using matrix notation.

By the discussion surrounding \cref{eq:factorization_explain}, $(a,b) \in \mathbf{B}$ is a NLFT if and only if there exist $(a_+, b_+) \in \mathbf{H}_{\geq 0}$ and $(a_{-}, b_{-}) \in \mathbf{H}_{ \leq 1}$ for which the factorization \cref{eq:factorization_explain} holds. To prove the existence of $(a_+, b_+)$, one first argues that for any such pair, we have 
\[
\begin{pmatrix}
    A \\ B
\end{pmatrix} := a_+ (\infty) \begin{pmatrix}
    a_+ \\ b_+
\end{pmatrix}
\]
is a fixed point of the map appearing in \cite[(7.14)]{AlexisMnatsakanyanThiele2023} from $H^2 (\D^*) \times H^2 (\D)$ to itself. This map may be rewritten as  
\begin{equation}\label{eq:contraction_mapping}
\begin{pmatrix}
    A \\ B
\end{pmatrix} \mapsto \begin{pmatrix}
    1 \\ 0
\end{pmatrix} + \begin{pmatrix}
    - (P_{\D} \frac{b}{a} B^*)^* \\ 
    P_{\D} \frac{b}{a} A
\end{pmatrix}  = \begin{pmatrix}
    1 \\ 0
\end{pmatrix} + \begin{pmatrix}
    - P_{\D^*} \frac{b^*}{a^*} B \\ 
    P_{\D} \frac{b}{a} A
\end{pmatrix} = \begin{pmatrix}
    1 \\ 0
\end{pmatrix} -M \begin{pmatrix}
    A \\ B
\end{pmatrix}  \, ,
\end{equation}
where we used the identity
\[
(P_{\D} f)^* = P_{\D^* } f^* \, ,
\]
and where \begin{equation}\label{eq:naive_operator}
M:= \begin{pmatrix}
    0 & P_{\D^*} \frac{b^*}{a^*} \\  -P_{\D} \frac{b}{a} & 0 
\end{pmatrix} \, .
\end{equation}
The size assumption \cref{eq:f_bded} in \cite{AlexisMnatsakanyanThiele2023} implies 
\begin{equation}\label{eq:bounded_b/a}
\left \| \frac{b}{a} \right \|_{\infty} < 1 \, ,
\end{equation}
and hence \cref{eq:contraction_mapping} is a Banach contraction mapping. Thus there exists a unique fixed point $\begin{pmatrix}
    A \\ B
\end{pmatrix}$, from which we can deduce the existence of a unique $(a_+, b_+)$ and $(a_{-}, b_{-})$.

\section{Riemann-Hilbert-Weiss algorithm for finding phase factors}\label{sec:main_algorithm}

\subsection{Outline}\label{sec:outline}

The key observation of this work is that the condition $\norm{f}_{\infty}<\frac{1}{\sqrt{2}}$, which ensures the Banach contraction mapping condition in \cref{sec:existing_rhf}, may be relaxed. The existence of a fixed point $\begin{pmatrix}
    A \\ B
\end{pmatrix}$ in \cref{eq:contraction_mapping} is equivalent to the existence of the solution of
\begin{equation}\label{eq:basic_linear_eqn}
\left ( \Id + M  \right ) \begin{pmatrix}
    A \\ B
\end{pmatrix} = \begin{pmatrix}
    1 \\ 0
\end{pmatrix}.
\end{equation}
Then \cref{eq:bounded_b/a} implies $M$ has operator norm strictly less than $1$, and hence by the von Neumann series we write
\begin{equation}\label{eq:inversion_geometric}
 (\Id + M)^{-1}= \sum\limits_{k=0} ^{\infty} (-M)^k \, ,
\end{equation}
where the sum converges absolutely with respect to the operator norm. Thus \cref{eq:basic_linear_eqn} has the explicit solution
\[
\begin{pmatrix}
    A \\ B
\end{pmatrix} = \left ( \Id + M  \right )^{-1} \begin{pmatrix}
    1 \\ 0
\end{pmatrix} = \sum\limits_{k=0} ^{\infty} (-M)^k \begin{pmatrix}
    1 \\ 0
\end{pmatrix}.
\]

It is reasonable to question whether the condition \cref{eq:f_bded} can be relaxed to 
\[
\norm{f}_{\infty} = \|b\|_{\infty} \leq 1,
\]
since this is a sufficient and necessary condition for the existence of QSP representation for polynomial $f$ with definite parity~\cite[Corollary 5]{GilyenSuLowEtAl2019}. We provide a positive answer in \cref{thm:RH_a_bded_below} below, extending \cite[Theorem 11]{AlexisMnatsakanyanThiele2023}.

In this more general setting, \cref{eq:bounded_b/a} may not hold and so the inversion formula \cref{eq:inversion_geometric} may no longer be valid. We instead invert $\Id + M$ by showing $M$ is an antisymmetric operator on $H^2 (\D^*) \times H^2 (\D)$ and hence has pure imaginary spectrum. Of course, when $|a|$ is not bounded below uniformly on $\T$, one must make sense of the operator $M$. Our approach in \cref{subsection:unbded_op} is to use the theory of unbounded operators to replace the operator in \cref{eq:naive_operator} by its unbounded analogue in \cref{eq:def_extension_M}. 

\begin{theorem}[Extension of {\cite[Theorem 11]{AlexisMnatsakanyanThiele2023}}, Riemann-Hilbert factorization]\label{thm:RH_a_bded_below}
    Let $(a,b) \in \mathbf{B}$. Then for each $k \in \Z$, there exists a unique factorization 
\begin{equation}\label{eq: factorization}
    (a,b) = (a_{<k} , b_{ <k} )( a_{\geq k}, b_{\geq k })
\end{equation}
with $(a_{<k}, b_{<k}) \in \mathbf{H}_{ \leq k-1}$ and $(a_{\geq k}, b_{\geq k}) \in \mathbf{H}_{\geq k}$.
\end{theorem}
According to the discussion surrounding \cref{eq:factorization_explain}, solving for $(a_{<k}, b_{<k})$, $(a_{\geq k}, b_{\geq k})$ in the Riemann-Hilbert factorization problem \cref{eq: factorization} is equivalent to showing that each $(a,b) \in \mathbf{B}$ is the NLFT of a unique sequence $F \in \ell^2 (\Z)$.

Using the factorization result in \cref{thm:RH_a_bded_below}, an individual nonlinear Fourier coefficient $F_k$, and hence the phase factor $\psi_k$, may be computed using the formula (see \cite[Eq. (6.13)]{AlexisMnatsakanyanThiele2023})
\begin{equation}
F_k = \frac{(b_{\geq k} z^{-k}) (0)}{a_{\geq k} ^* (0)}.
\end{equation}
We phrase this process using the key quantities appearing in our algorithm, namely a scalar multiple $(A_k, B_k)$ of $(a_{\geq k}, b_{\geq k})$.

\begin{lemma}\label{lem:f_to_phase_factor}
Let $k \in \mathbb{N}$. Given any $f \in \mathbf{S}_\eta$, we can recover the phase factor $\psi_k$ via the maps
    \begin{equation}\label{eq:tentative_Riesz_composition}
f \mapsto \frac{b}{a} \mapsto (A_k ,B_k) \mapsto F_k \mapsto \psi_k \, .
\end{equation}
Here, $(A_k, B_k)$ is the unique element of 
\begin{equation}\label{eq:Hilbert_defn_k}
    \mathcal{H}_k :=H^2 (\D^*) \times z^{k} H^2 (\D)
\end{equation} satisfying  
\begin{equation}\label{eq:linear_system_k}
(\Id +M_k) \begin{pmatrix}
    A_k \\ B_k
\end{pmatrix} = \begin{pmatrix}
    1 \\ 0
\end{pmatrix} \, , \quad M_k = \begin{pmatrix}
    0 & P_{\D^*}  \frac{b ^* }{a ^*} \\ - z^k  P_{\D} z^{-k} \frac{b}{a} & 0
\end{pmatrix} \, .
\end{equation}
Then
\[
F_k := \frac{(B_k z^{-k}) (0)}{A_k ^*(0)} \, , \quad 
\psi_k := \arctan (-\I F_k) \, .
\]
\end{lemma}

The proof of this lemma is included in the proof of \cref{thm:main} which is given in \cref{sec:main_proof}. The maps described in \cref{eq:tentative_Riesz_composition} provide a novel method for computing phase factors independently, which we name the Riemann-Hilbert-Weiss algorithm (\cref{alg:RHW_alg}). The name convention follows the chronological order and it reflects two important steps required by the computation of phase factors, constructing $\frac{b}{a}$ and computing $(A_k, B_k)$. For the first step, we develop a numerical method (\cref{alg:weiss_alg}) for evaluating the Fourier coefficients of $\frac{b}{a}$. We call it the Weiss algorithm, since it follows the same idea in the Guido and Mary Weiss algorithm \cite{MaryWeiss62}. As for the second step, we only need to numerically solve the linear system outlined in \cref{eq:linear_system_k}, where the existence and uniqueness of the solution is guaranteed by Riemann-Hilbert factorization results. Here $\ve_0$ refers to the vector where the first entry is 1 and all other entries are zero.

\begin{algorithm}[htbp]
\caption{Riemann-Hilbert-Weiss algorithm for finding phase factors}
\label{alg:RHW_alg}
\begin{algorithmic}
\STATE{\textbf{Input:} An even real-valued polynomial $f\in {\bf S}_{\eta}$ of degree $2d$ ($ d\geq 1$), and $\epsilon > 0$.} 
\STATE{\textbf{Output:} A set $\Psi$ of symmetric phase factors .}
\STATE{Let $b(z) = \I f(x)$, where $x$ and $z$ are related by  \cref{eq:change_of_variables}.}
\STATE{Obtain coefficients $\{\hat{c}_i\}_{i=0}^d$ using \cref{alg:weiss_alg}. }
\FOR{$k=0,\cdots,d$}
\STATE{Solve the linear system $\begin{pmatrix}
    I &  -\Xi_k\\
    -\Xi_k & I 
\end{pmatrix}\begin{pmatrix}
    \va_k\\
    \vb_k
\end{pmatrix} = \begin{pmatrix}
    \ve_0\\
    \bvec{0}
\end{pmatrix}$ for $\va_k$ and $\vb_k$, where $\Xi_k$ is the Hankel matrix of size $(d+1-k)\times (d+1-k)$ with $\{\hat{c}_i\}_{i=k}^d$ as its first column and zeros below the secondary diagonal, and $\ve_0$ is the first column of the identity matrix.}
\STATE{Compute $\psi_k = \arctan\left(-\I \frac{b_{k,0}}{a_{k,0}}\right)$, where $a_{k,0}$ and $b_{k,0}$ are the first entries of $\va_k$ and $\vb_k$.}
\ENDFOR
\RETURN $\Psi$
\end{algorithmic}
\end{algorithm}

\subsection{Weiss algorithm}\label{sec:Weiss_alg}

The construction of $a$ as shown in the proof of \cref{thm:construct_a} consists of three steps:  $R(z):=\log\sqrt{1-\abs{b(z)}^2}$;  $G(z):=R(z)-\I\mathrm{H}(R(z))$; and $a(z):= e^{G(z)}$. To construct $\frac{b}{a}$, we only need to replace the last step by evaluating $\frac{b}{a}:= b(z) e^{-G(z)} $. In our  method, we directly evaluate the Fourier coefficients of $\frac{b}{a}$ with indices ranging from $0$ to $d$ using the Fast Fourier Transform (FFT). Let $b(z)$ be a Laurent polynomial of degree $d$. Due to the correspondence between the target function $f(x)$ and $b(z)$, we have $b(z)= b(z^{-1})$ and $\norm{b}_{\infty}=\norm{f}_{\infty} \leq 1-\eta$. We also denote by $N$ the discretization parameter associated with FFT.  

We first evaluate function $\log \left(\sqrt{1-\abs{b(z)}^2}\right)$ at the $N$th roots of unity $z_j := e^{\I \frac{2j\pi}{N}}$.  Then FFT is applied to obtain $\hat{R}:= \left(\hat{r}_{-\lfloor\frac{N}{2}\rfloor}, \cdots,\hat{r}_0, \cdots, \hat{r}_{N-1-\lfloor\frac{N}{2}\rfloor}\right)$, which is an approximation of the Fourier coefficients of $R(z)$. The identities in \cref{eq:Hilbert_tranform_rule} imply that
\begin{equation*}
\begin{split}
&z^n -\I \mathrm{H}(z^n)=0, \quad n\in \mathbb{N}_+,\\
&z^{-n}-\I\mathrm{H}(z^{-n})=2 z^{-n}, \quad n\in \mathbb{N}_+,\\
& c - \I \mathrm{H}(c) = c, \quad c\in \CC.
\end{split}
\end{equation*}
To construct $G(z):=R(z)-\I \mathrm{H}(R(z))$, we only need to discard those positive frequencies and double those negative frequencies, while keeping the zero-frequency component. Thus $\hat{G} (z):=\hat{r}_0 + 2\sum_{\ell = 1}^{\lfloor \frac{N}{2}\rfloor} \hat{r}_{-\ell} z^{-\ell}$ provides a numerical approximation to $G(z)$. We evaluate $b(z) e^{-G(z)}$ at $\{z_j\}$ and apply FFT again to obtain $\hat{C}$, which is an approximation of the Fourier coefficients of $\frac{b}{a}$.

The Weiss algorithm is given in \cref{alg:weiss_alg}.  The choice for the parameter $N$ to achieve precision $\epsilon$ is given by \cref{thm:sufficient_condition_N} and a more detailed discussion is provided in \cref{sec:complexity_analysis}. 
% \LL{ This is quite standard and is commented out} Note that when implementing \cref{alg:weiss_alg}, we need swap the right and the left halves of the results each time we apply FFT, due to the indexing conventions of the FTT. This rearrangement results in a symmetrical and centered spectrum, where negative frequencies are positioned on the left of the zero-frequency component and positive frequencies to its right.

\begin{algorithm}[htbp]
\caption{Weiss algorithm for evaluating the Fourier coefficients of $\frac{b}{a}$}
\label{alg:weiss_alg}
\begin{algorithmic}
\STATE{\textbf{Input:} A pure imaginary Laurent polynomial $b(z)$ of degree $d$ with $\norm{b}_{\infty}\leq 1-\eta<1$ and $b(z) = b(z^{-1})$, and  $\epsilon > 0$.}
\STATE{\textbf{Output:}  $(\hat{c}_0, \hat{c}_{1}, \cdots, \hat{c}_{d})$, which are approximating Fourier coefficients of $\frac{b}{a}$, .} 
\STATE{}

\STATE{Choose $N = N (d, \eta,\epsilon)$ to be the smallest power of $2$ satisfying \cref{eq:N_lower_bd}.
}
%\STATE{Set $\{z_k\}$ to be the $N$th roots of unity, $z_k = e^{\I \frac{2k\pi}{N}}$.}
\STATE{Apply FFT to the evaluation of function $\log \left(\sqrt{1-\abs{b(z)}^2}\right)$ at $\{z_j\}$, and denote the rearranged results as $\hat{R} = \left(\hat{r}_{-\lfloor\frac{N}{2}\rfloor}, \cdots, \hat{r}_0, \cdots, \hat{r}_{N-1-\lfloor\frac{N}{2}\rfloor}\right)$.}
\STATE{Apply FFT to the evaluation of function $b(z) e^{-\hat{r}_0 - 2\sum_{\ell = 1}^{\lfloor\frac{N}{2}\rfloor} \hat{r}_{-\ell} z^{-\ell}}$ at $\{z_j\}$, and the rearranged results are denoted as $\hat{C} = \left(\hat{c}_{-\lfloor\frac{N}{2}\rfloor}, \cdots, \hat{c}_0, \cdots, \hat{c}_{N-1-\lfloor\frac{N}{2}\rfloor}\right)$.} 

\STATE{}
\RETURN $(\hat{c}_0, \hat{c}_{1}, \cdots, \hat{c}_{d})$. 
\end{algorithmic}
\end{algorithm}

\subsection{Riemann-Hilbert factorization}\label{sec:RHW_alg}
We now introduce a method for computing $(A_k, B_k)$ for any $0\le k\le d$. We also introduce a truncation parameter 
\begin{equation}\label{eq:trunc_param_def}
    n:= d-k \, .
\end{equation} 

 %which shows the relation between \cref{alg:RHW_alg} and Theorem \ref{thm:RH_a_bded_below}.

% Let us fix $\eta,\epsilon > 0$ and Given a real-valued polynomial $f \in {\bf S}_{\eta}$ of degree $2d$ ($d \geq 1$). Set $b (z) = i f (x)$ satisfying the assumptions of \cref{alg:RHW_alg}. Also fix $0\leq k\leq d$ and fix truncation parameter 
% \begin{equation}\label{eq:trunc_param_def}
%     n:= d-k \, .
% \end{equation}
As $aa^* + bb^* = 1$ and $\norm{b}_{\infty} \leq 1-\eta$ on $\T$, we have  $|a(z)| \geq \sqrt{1 - (1-\eta)^2}\geq\sqrt{\eta}$ on $\T$.
Together with $a\in H^2(\D^*)$, we have $\frac{1}{a}\in H^2(\D^*)$. Also $b$ is a Laurent polynomial of degree $d$. Hence, $\frac{b}{a}$ has Fourier support on $(-\infty, d]$. Let $c_j$ denote the $j$th Fourier coefficient of $\frac{b}{a}$. These Fourier coefficients turn out to be all pure imaginary.

\begin{lemma}\label{lem:imaginary_Fourier_coeffs}
Let $(a,b) \in \mathbf{B}$, and assume $\|b\|_{\infty} \leq 1 - \eta$ for some $\eta > 0$. If $b$ satisfies $b(z) = b(\overline{z})$ for all $z \in \T$ and $\I b$ is real-valued on $\T$, then $\frac{b}{a}$ has pure imaginary Fourier coefficients.
\end{lemma}
\begin{proof}
 It suffices to show $a^{-1}$ has real Fourier coefficients, while $b$ has pure imaginary Fourier coefficients, since then the product $\frac{b}{a}$ of the bounded functions $b$ and $a^{-1}$ will then have pure imaginary Fourier coefficients. Because $b(z) = \I f(x)$ where $z \in \T$ is related to $x \in [0,1]$ through \cref{eq:change_of_variables}, then $b(\overline{z}) = b(z)$ for all $z \in \T$. Combined with the fact that $b(z)$ takes on pure imaginary values, we find that $b$ has pure imaginary Fourier coefficients. And because $\log \sqrt{1 - |b(z)|^2}$ is real-valued and is also invariant under the change of variable $z \mapsto \overline{z}$, we get $\log \sqrt{1 - |b(z)|^2}$ has real Fourier coefficients. Because $\log \sqrt{1 - |b(z)|^2} \in L^{\infty} (\T) \subset L^2 (\T)$, and $(\Id + \I \mathrm{H})$ maps real-valued $L^2 (\T)$ functions into $H^2 (\D)$, we thus have 
\[
G^* (z) := \log \sqrt{1 - |b(z)|^2}+\I  \mathrm{H} \log \sqrt{1 - |b(z)|^2}
\] belongs to $H^2 (\D)$ and has real Fourier coefficients as well\REV{, since by \eqref{eq:Hilbert_tranform_rule} the operator $\I  \mathrm{H}$ is a real-valued multiplier on the frequency side}. Since the Fourier expansion of $G$ coincides with its Taylor expansion as a holomorphic function in the unit disk $\D$, then $G(z)$ and all of its derivatives, when evaluated at $z=0$, must be real-valued. Thus the $H^{\infty} (\D)$ function 
\[
\frac{1}{a^* (z)} = e^{-G ^*(z)} 
\]
and all its derivatives, when evaluated at $z=0$, must be real-valued. Thus $(a^*)^{-1}$, and hence $a^{-1}$, must have real Fourier coefficients.
\end{proof}

%\begin{lemma}
    %Let the pair $(a,b)$ be the nonlinear Fourier transform of a purely imaginary sequence $(F_k)_{k \in \Z}$, then $b$ is a linear Fourier series with purely imaginary coefficients, while $a$ is a linear Fourier series with real coefficients. Furthermore, the Fourier coefficients of $\frac{b}{a}$ are all purely imaginary.
%\end{lemma}
%\begin{proof}
%Indeed, a multilinear expansion of $(a,b)$ using the NLFT equations \cref{eq:NLFT_ab}, as done in \cite[(5.9) and (5.10)]{AlexisMnatsakanyanThiele2023}, shows that $a$ (or $b$) may be expanded into a sum of linear Fourier series whose coefficients are each a real constant times a product of an even (or odd) number of $F_j$'s. Thus if $\{F_k\}_k$ is a purely imaginary sequence, then $b$ is a \emph{linear} Fourier series with purely imaginary coefficients, while $a$ is a \emph{linear} Fourier series with real coefficients. Under the assumption that $\norm{b}_\infty\leq 1 - \eta$, then we have $\min_{z\in \T} \abs{a}\geq \sqrt{1-(1-\eta)^2}\geq \sqrt{\eta}$ and so because $a$ has only real Fourier coefficients, then the uniformly convergent series 
%\[
%a^{-1} = \frac{1}{1-(1-a)} = \sum\limits_{\ell=0}^{\infty} (1-a)^{\ell}
%\]
%has only real Fourier coefficients as well. Thus all Fourier coefficients of the product $\frac{b}{a}$ must be purly imaginary.
%\end{proof}

Now we truncate the Hilbert space $\mathcal{H}_k$, defined in \cref{eq:Hilbert_defn_k}, to the finite dimensional space,
\begin{equation}
    \mathcal{H}^n _k =\text{span}(1,z^{-1}, z^{-2},\cdots, z^{-n}) \times \text{span}(z^{k},\cdots,z^{n+k}).
\end{equation}
%Here, $n$ is referred to as the truncation degree.
We define
\begin{equation*}
    \Lambda_{\ell} :=\begin{cases}
    (z^{-\ell}, 0) & 0\leq \ell \leq n,\\
    (0, z^{\ell-(n+1)+k}) & n+1\leq \ell \leq 2n+1.
    \end{cases} 
\end{equation*}
so that $\{\Lambda_\ell\}_{\ell=0}^{2n+1}$ is an ordered basis of $\mathcal{H}^n_k$. The matrix representation of the operator on the left of \cref{eq:linear_system_k} with respect to the ordered basis $\{\Lambda_{\ell}\}_{\ell=0}^{2n+1}$ in $\mathcal{H}^n_k$ is
\begin{equation}\label{eq:approx_matrix_hankel}
    \Id + M_k= \begin{pmatrix}
    I &  -\Xi_k\\
    -\Xi_k & I 
    \end{pmatrix},
\end{equation}
where $\Xi_k$ is the $(n+1)\times (n+1)=(d-k+1) \times (d-k+1)$ Hankel matrix  with $(c_k, c_{k+1}, \cdots, c_{d})^{T}$ as its first column, that is, the $(ij)$-entry of $\Xi_k$ satisfies
\begin{equation*}
    (\Xi_k)_{ij} = \begin{cases}
        c_{i+j+k} & i+j+k\leq d, \\
        0 & \text{otherwise}.
    \end{cases}
\end{equation*}
Because $c_j$ is pure imaginary for all $j$ by \cref{lem:imaginary_Fourier_coeffs}, $\Xi_k$ is a pure imaginary matrix.

Thus, computing $(A_k, B_k)$ is equivalent to solving the linear system 
\begin{equation}\label{eqn:numerical_linear_system}
    (\Id + M_k ) \begin{pmatrix}
        \va_k\\
        \vb_k\\
    \end{pmatrix} = \begin{pmatrix}
        \ve_0\\
        \bvec{0}
    \end{pmatrix},
\end{equation}
where $\va_k, \vb_k, \bvec{0}, \ve_0\in\mathbb{R}^{n+1}$ and $\ve_0$ is the first column of the identity matrix.

We will show in \cref{sec:complexity_analysis} that, analogous to the theoretical results of \cref{eq:inverse_bded},  the matrix
$$\Id + M_k:=\begin{pmatrix}
    I & -\Xi_k \\
    - \Xi_k & I 
\end{pmatrix}$$
is non-singular. Hence, $\va_k,\vb_k$ are well-defined and we can recover the phase factor $\psi_k$ through
\begin{equation}
F_k = \frac{(B_k z^{-k}) (0)}{A_k ^*(0)} = \frac{b_{k,0}}{a_{k,0}}, \quad \psi_k = \arctan(-\I F_k).
\end{equation}
Here $a_{k,0}$ and $b_{k,0}$ are the first entries of $\va_k$ and $\vb_k$.

The accuracy of the phase factors computed by the Riemann-Hilbert-Weiss algorithm actually depends on that of Fourier coefficients obtained by the Weiss algorithm, which turns out to be determined by the discretization parameter $N$. In the following theorem, we present a sufficient condition about the choice of $N$ to achieve the desired precision. 

%\begin{theorem}\label{thm:sufficient_condition_N}
%    For any $\eta\in (0,1)$ and any polynomial $f\in \mathbf{S}_{\eta}$ with degree $d$, when applying the Riemann-Hilbert-Weiss algorithm to compute $\psi_k$ $(0\leq k\leq d)$, to achieve the relative error bounded by $\epsilon$, it is sufficient to choose $N$ no less than
%\begin{equation}\label{eq:upper_bound_N}
%    4d \frac{\log \left(\frac{126d(d+1)}{\epsilon \eta^2} \log\left(\frac{1}{\eta}\right)\right)}{\log(\frac{1}{1-\eta})} = \Theta\left( \frac{d}{\eta}  \log\left(\frac{d}{\eta\epsilon}\right) \right).
%\end{equation} 
%\end{theorem}
%\gm{alternate take}

\begin{theorem}\label{thm:sufficient_condition_N}
    Given any $0<\eta<\frac{1}{2}$, $0<\epsilon<1$ and integer $d\geq 1$, assume $f\in \mathbf{S}_{\eta}$ to be a polynomial of degree $2d$. Let
    \begin{equation}\label{eq:N_lower_bd}    N\ge \frac {8d}{\eta}\log\left(\frac {576 d^2} {\eta^4 \epsilon}\right)
    \end{equation}
    %\begin{equation}
    %N \geq  4d \frac{\log \left(\frac{48d(d+2)}{\epsilon \eta^4} \log\left(\frac{1}{\eta}\right)\right)}{\log(\frac{1}{1-\eta})} = \wt{\Theta}\left( \frac{d}{\eta}  \log\left(\frac{d}{\eta\epsilon}\right) \right) \, ,
%\end{equation}
be an even integer. For any $0\leq k\leq d$, let $\psi_k$ be the $k$th phase factor of $f$ and $\hat{\psi}_k$ be the approximation to the $k$th phase factor of $f$ computed by the Riemann-Hilbert-Weiss algorithm with the discretization parameter $N$. Then
\begin{equation}
    |\psi_k-\hat{\psi}_k| \le \epsilon, \quad \forall 0\leq k \leq d \, .
\end{equation}
\end{theorem}

The proof of this theorem is presented in \cref{sec:complexity_analysis}.

Next, we present some numerical experiments to demonstrate the ability of Riemann-Hilbert-Weiss algorithm.

\subsection{Numerical experiment}\label{sec:numerical_experiments}
We initiate our experiment by randomly generating a phase factor sequence $\Psi$ of length $1000$. To ensure that the generated phase factors are actually the maximal solution for certain target function, we first normalize $\Psi$ to have 1-norm bounded by a small absolute constant. Subsequently, we apply selective scaling to the elements of $\Psi$: each $\psi_k$ is scaled by $10^{-7}$ for $1\leq k \leq 333$ and $667\leq k \leq 1000$. The magnitudes of the adjusted phase factors $\Psi$ are depicted in \cref{fig:phi_per_index}. 

The target function $f$ is chosen to be associated with the constructed phase factors $\Psi$. We use two methods:  Newton's method (which has been numerically demonstrated to be robust even for small $\eta$ \cite{DongLinNiEtAl23}), and \cref{alg:RHW_alg} with a large discretization parameter $N = 10^6$ to find the phase factors given target function $f$. We also present a comparative analysis of the computation errors for each phase factor in \cref{fig:error_per_index}, demonstrating that the accuracy achieved using the Riemann-Hilbert-Weiss  method is comparable to that obtained with Newton's method.

To show that \cref{alg:RHW_alg} is also robust for small $\eta$, we consider the target function $f(x) = 0.999\cos(\tau x)$ with $\tau = 1000$, which originates from the application to Hamiltonian simulation and $\eta=0.001$. The Chebyshev series expansion, known as the Jacobi-Anger expansion~\cite{LowChuang2017}, is commonly employed to approximate this target function:
\begin{equation}
        \label{eq;Jacobi-Anger}
        0.999\cos(\tau x)=0.999\left(J_0(\tau)+2\sum_{k\text{ even}}(-1)^{k/2} J_k(\tau) T_k(x)\right),
\end{equation}
where $J_k$'s are the Bessel functions of the first kind. As a result, by truncating the Jacobi-Anger series, a polynomial approximation of the target function can be obtained. To ensure that the truncation error is upper-bounded by $\epsilon_0$, it is sufficient to choose the degree of truncation as $d = \lceil e|\tau|/2 + \log(1/\epsilon_0)\rceil$.

We apply both Riemann-Hilbert-Weiss algorithm and Newton's method to find phase factors given target function $f$, and  the discretization parameter is chosen to be $N=10^7$. We plot the magnitude of phase factors obtained by Riemann-Hilbert-Weiss algorithm in \cref{fig:phi_per_index_HS}, as well as the difference of phase factors obtained by both methods in \cref{fig:error_per_index_HS}. We also assess the accuracy of the results by evaluating $\Im [u_d(x,\Psi)]$ at the Chebyshev nodes of $T_{500}(x)$ and comparing these values to the exact value of the target function, $0.999\cos(\tau x)$. The resulting errors are depicted in \cref{fig:error_per_x_HS}. The findings clearly indicate that the accuracy of the Riemann-Hilbert-Weiss algorithm is on par with Newton's method in a nearly fully coherent regime.

\begin{figure}
\centering
\begin{subfigure}{0.45\textwidth}
  \includegraphics[width=60mm]{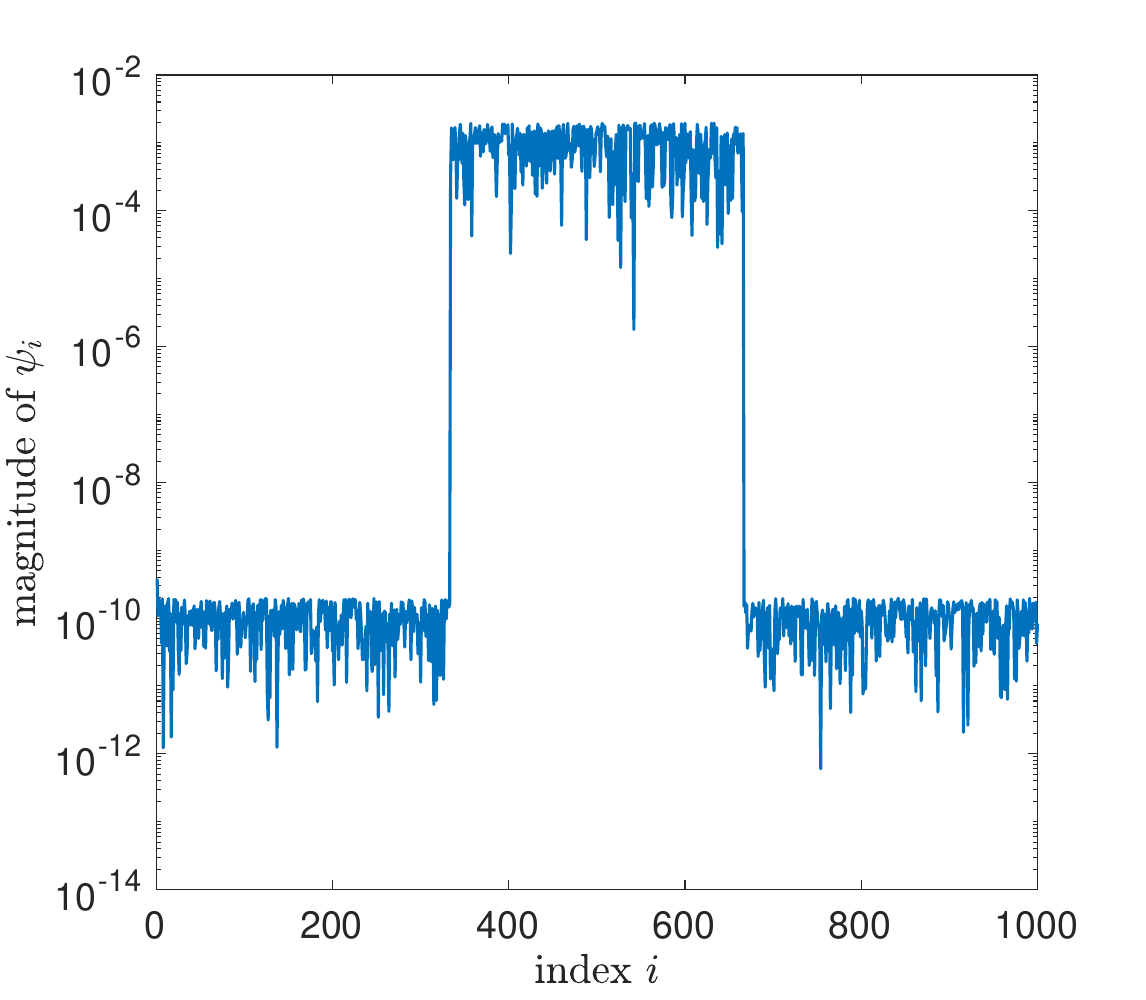}
  \caption{The magnitude of phase factors $\Psi$ for each index $i$.}
  \label{fig:phi_per_index}
\end{subfigure}
\hfill
\begin{subfigure}{0.45\textwidth}
  \includegraphics[width=60mm]{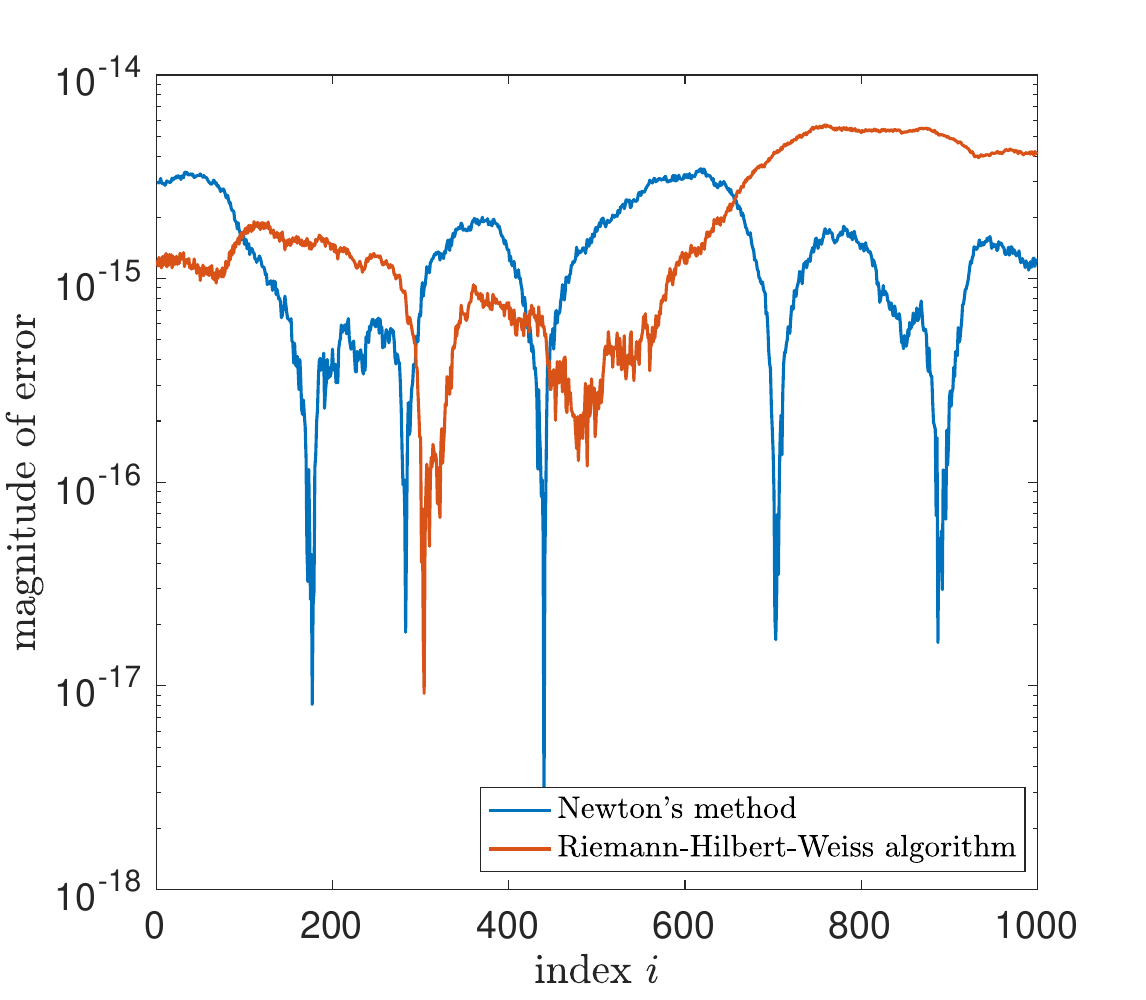}
   \caption{The error in phase factors $\Psi$ obtained by two methods.}
  \label{fig:error_per_index}
\end{subfigure}
\caption{The performance of \cref{alg:RHW_alg} and Newton's method to find phase factors for the randomly generated phase factors $\Psi$.}
\label{fig:numerics}
\end{figure}

\begin{figure}
\centering
\begin{subfigure}{0.45\textwidth}
  \includegraphics[width=60mm]{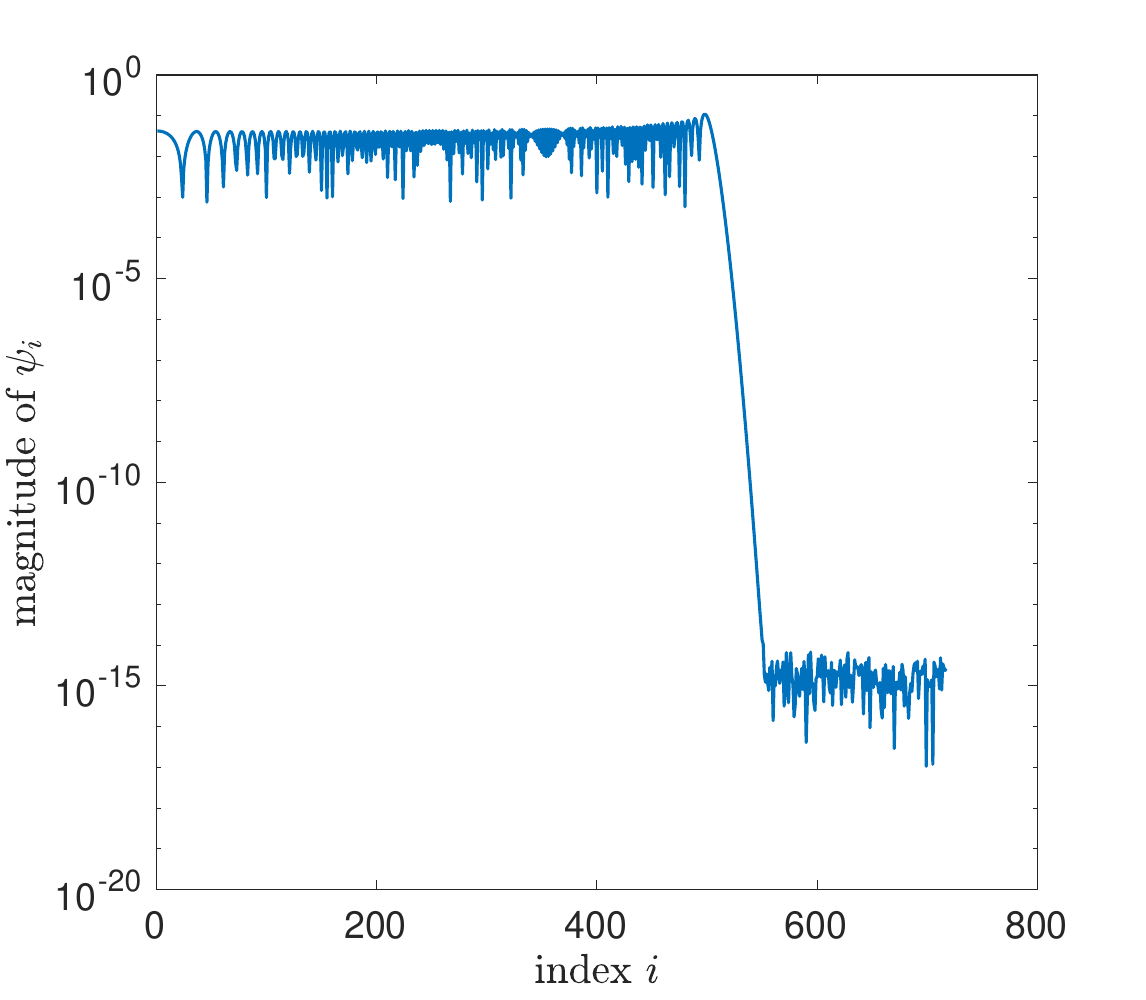}
  \caption{The magnitude of phase factors $\Psi$ for each index $i$.}
  \label{fig:phi_per_index_HS}
\end{subfigure}
\hfill
\begin{subfigure}{0.45\textwidth}
  \includegraphics[width=60mm]{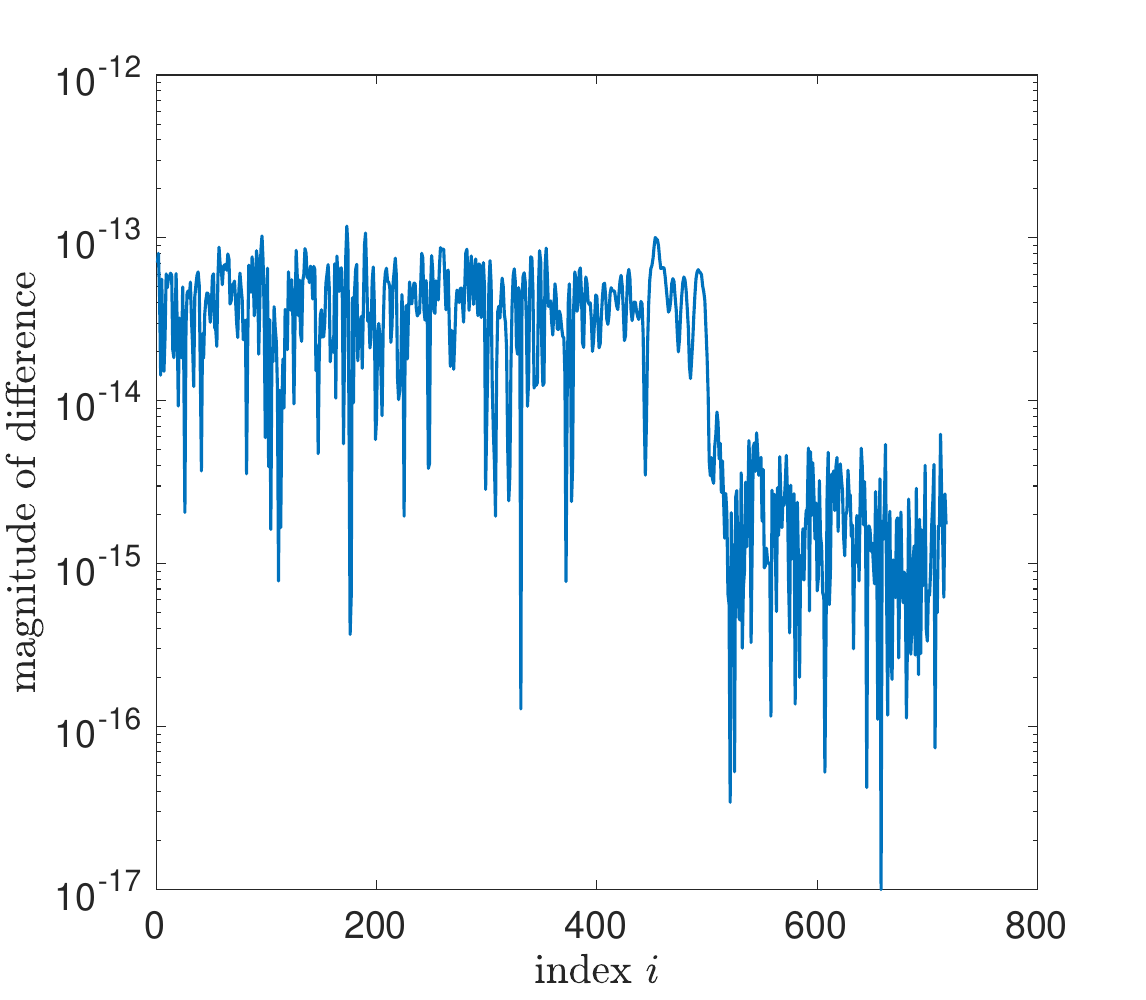}
   \caption{The difference in phase factors $\Psi$ obtained by \cref{alg:RHW_alg} and Newton's method.}
  \label{fig:error_per_index_HS}
\end{subfigure}
\caption{The performance of \cref{alg:RHW_alg} and Newton's method to find phase factors for the target function $0.999\cos(\tau x)$.}
\label{fig:numerics_HS}
\end{figure} 

\begin{figure}
\centering
  \includegraphics[width=60mm]{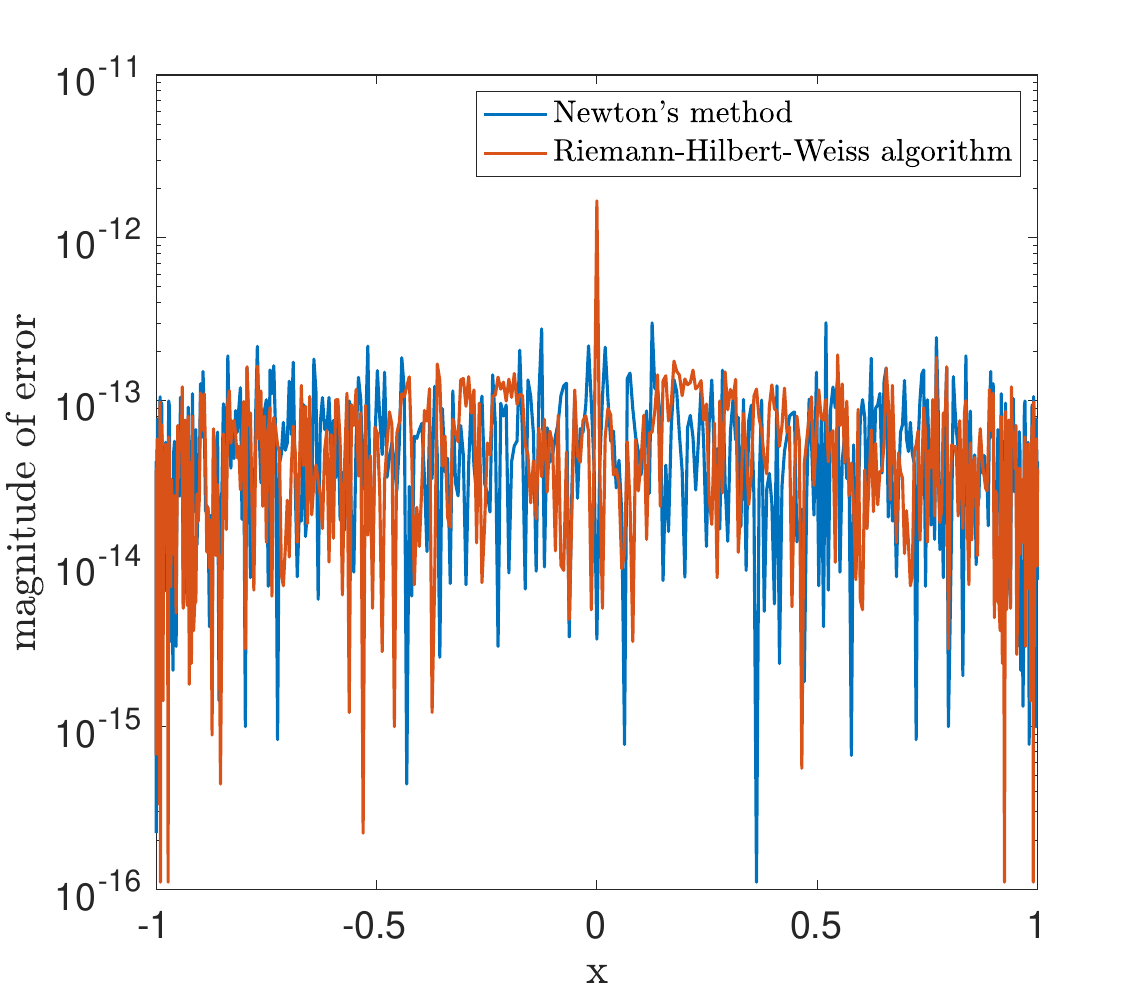}
   \caption{The difference $\abs{\Im[u_d(x,\Psi)] - f(x)}$ over interval $[0,1]$ for the phase factors obtained by \cref{alg:RHW_alg} and Newton's method for the target function $0.999\cos(\tau x)$.}
  \label{fig:error_per_x_HS}
\end{figure}

\section{Infinite Quantum signal processing using square summable sequences. }\label{sec:main_proof}

In this section, we present the proof of \cref{thm:main} which provides a representation of measurable functions by a square summable sequence. As mentioned previously in \cref{sec:existing_rhf} and \cref{sec:outline}, we prove \cref{thm:RH_a_bded_below}, extending the Riemann-Hilbert factorization result \cite[Theorem 11]{AlexisMnatsakanyanThiele2023} in nonlinear Fourier analysis to all $(a,b) \in \mathbf{B}$. Our key idea here is to argue that the operator $M$ arising in \cref{eq:naive_operator} is antisymmetric on the Hilbert space $H^2 (\D^*) \times H^2 (\D)$, meaning it possesses pure imaginary spectrum and hence $\Id + M$ is invertible. If $(a,b) \in \mathbf{B}$ satisfies $|a| > \eta'$ for some $\eta'>0$, then $\frac{b}{a}$ is bounded on $\T$ and the above reasoning is rigorous. However, if $(a,b) \in \mathbf{B}$, then $\frac{b}{a}$ may not be integrable and $M$ may not be well-defined on $H^2 (\D^*) \times H^2(\D)$. This leads to the technicalities in this section as we must appropriately interpret $M$ at the endpoint case $\eta' =0$. Indeed, the simple operator $M$ in \cref{eq:naive_operator} must be interpreted as a densely defined operator as in \cref{eq:def_extension_M}, which involves truncating $a$ from below and taking a weak limit the in the truncation parameter. And it's by using the theory of unbounded operators \cite[Chapter 13]{rudin} that we argue that $M$ is an antisymmetric unbounded operator to then get $\Id + M$ has bounded inverse. The arguments in the case $|a| > \eta'$ for some $\eta' >0$, which follows from $f \in \mathbf{S}_{\eta}$ for some $\eta > 0$, are considerably simpler as the weak limits coincide with the simple operator in \cref{eq:naive_operator}. For instance, \cref{lem:iM_self_adjt} is near immediate if we take $M$ as in \cref{eq:naive_operator}. Thus any reader uninterested in the endpoint case $\eta =0$ may read this section while ignoring all weak limits. 

This section is organized as follows. We first introduce a vector space 
\[
\ddense \subset H^2 (\D^*) \times H^2 (\D)
\]
and an appropriate unbounded analogue of $M$ on $\ddense$. We show the operator $M$ has pure imaginary spectrum, meaning 
\[
\Id  \pm M : \ddense \to H^2 (\D^*) \times H^2 (\D)
\] is a linear bijection with bounded inverse, which allows us to prove the Riemann-Hilbert factorization \cref{thm:RH_a_bded_below}.
Next, we prove the Lipschitz estimate in \cref{thm:Lip_bds} which bounds the infinity distance of sequences by the norm distance in $\frac{b}{a}$. Then we discuss that for every element $(a,b) \in \mathbf{B}$, after the change of variables specified in \cref{lem:change_of_vars}, $b$ corresponds to the target function $f$, while $a$ corresponds to the maximal solution proposed in \cite{WangDongLin2021}. Furthermore, because $a$ is the unique function for which one has equality in the Plancherel inequality \cref{eq:Plancherel_ineq}, we then get \cref{plancherel} for iQSP. Finally, combining all these results with \cref{thm:RH_a_bded_below}, we prove \cref{thm:main}. 

\subsection{Inversion of an unbounded operator using spectral theory}\label{subsection:unbded_op}

Fix $(a,b) \in \mathbf{B}$.
Define the spaces
\[
\Hil := H^2 (\D^*) \times  H^2 (\D) \,  
\]
\begin{equation}\label{eq:dense_subset}
\dense := a H^2 (\D^*) \times a^*  H^2 (\D) 
\end{equation}
with norm  
\begin{equation}\label{eq:defn_norm}
\left \| 
    f 
 \right \| ^2 := \int\limits_{\T} |f_1|^2 + |f_2|^2 
\end{equation}
induced by the standard inner product $\langle \cdot, \cdot \rangle$ on $L^2(\T)\times L^2(\T)$.
As $a^*$ is bounded and outer, we have
\[\dense \subset \Hil\subset L^2(\T)\times L^2(\T).\]
The space $\Hil$ is closed. The space $\dense$ is dense in $\Hil$ by Beurling's theorem \cite[Chapter 2, Corollary 7.3]{garnett}, which says that the set of all polynomials in $z$, when multiplied by an outer function on $\D$, is dense in $H^2 (\D)$. 

Let 
\[
\mathcal{P}_{\Hil} = \begin{pmatrix}
    P_{\D} & 0 \\ 0 & P_{\D^*}
\end{pmatrix}
\]
denote the $L^2 \times L^2$ orthogonal projection onto $\Hil$. 

To simplify the notation, in this subsection we slightly abuse the notation, drop the prime notation and identify $\eta'$ with $\eta$.

Within the Hilbert space $\Hil$, define $\ddense$ to be the subspace of elements $f \in \Hil$ for which 
\begin{equation}\label{eq:def_set_symm}
\mathcal{P}_{\Hil} \begin{pmatrix}
    0 & \frac{b^*}{a^* _{\eta}} \\ - \frac{b}{a _{\eta}} & 0
\end{pmatrix} f
\end{equation}
converges weakly in $\Hil$ as $\eta \to 0$, where $a_{\eta} ^*$ is defined to be the unique outer function on $\D$ whose boundary values on $\T$ satisfy
\begin{equation} \label{eq:outer_approx}
\log |a_{\eta}| := \mathbf{1}_{ \{ |a| > \eta \} } \log |a| \, .
\end{equation}

If $|a|$ is bounded below uniformly away from $0$, then $\ddense = \Hil$ and for all $\eta$ sufficiently small \cref{eq:def_set_symm} equals
\[
\mathcal{P}_{\Hil} \begin{pmatrix}
    0 & \frac{b^*}{a^*} \\ - \frac{b}{a} & 0
\end{pmatrix} f \, .
\]
For general outer $a$, we may not have $\ddense = \Hil$.
\begin{lemma}\label{lem:limit_Dn}
    If $(a,b) \in \mathbf{B}$, then $\dense \subset \ddense$, and for any $f \in \dense$ the weak limit of \cref{eq:def_set_symm}, which is also a strong limit, equals
\begin{equation}\label{eq:limit_Dn}
 \mathcal{P}_{\Hil} \begin{pmatrix}
    0 & \frac{b^*}{a^*} \\ -\frac{b}{a} & 0
\end{pmatrix} f  \, .
\end{equation} 
\end{lemma} 
\begin{proof}
For any $f \in \dense$ there exists an $h \in \Hil$ for which
\[
f = \begin{pmatrix}
    a & 0 \\ 0 & a^*
\end{pmatrix}h \, . 
\]
Noting that \cref{eq:limit_Dn} equals
\[ \mathcal{P}_{\Hil} \begin{pmatrix}
    0 & b^* \\ -b & 0
\end{pmatrix} h \ , 
\]
we compute 
\[
\left \| \mathcal{P}_{\Hil} \begin{pmatrix}
    0 & \frac{b^*}{a^* _{\eta}} \\ - \frac{b}{a _{\eta}} & 0
\end{pmatrix} f  - \mathcal{P}_{\Hil} \begin{pmatrix}
    0 & b^* \\ -b & 0
\end{pmatrix} h\right \| \leq  \left \| \begin{pmatrix}
    0 & \frac{b^*}{a^* _{\eta}} \\ - \frac{b}{a _{\eta}} & 0
\end{pmatrix} f  -  \begin{pmatrix}
    0 & b^* \\ -b & 0
\end{pmatrix} h\right \|
\]
\[
= \left \| \begin{pmatrix}
    0 & b^* \frac{a^*}{a^* _{\eta}} \\ - b \frac{a}{a _{\eta}} & 0
\end{pmatrix} h  -  \begin{pmatrix}
    0 & b^* \\ -b & 0
\end{pmatrix} h\right \| = \left \| \begin{pmatrix}
    0 & b^*  \\ - b  & 0
\end{pmatrix} \begin{pmatrix}
    \frac{a}{a _{\eta}} - 1 & 0  \\ 0 & \frac{a^*}{a^* _{\eta}} - 1
\end{pmatrix} h\right \| \, .
\]
Using the fact that $|b| \leq 1$ a.e.\ on $\T$, the square of this last term is at most
\[
 \int\limits_{\T} \left ( |h_1|^2 + |h_2|^2 \right ) \left | \frac{a}{a_{\eta}} - 1 \right |^2 \, . 
\]
Since $|h_1|^2 + |h_2|^2$ is integrable, if 
we show that for every integrable nonnegative function $f$ we have
\begin{equation}\label{eq:bds_a_epsilon}
 \lim\limits_{\eta \to 0} \int\limits_{\T} f \left | \frac{a}{a_{\eta}} - 1 \right |^{\ell} = 0
\end{equation}
for all integers $\ell \geq 1$,
%\begin{equation}\label{eq:bds_a_epsilon}
 %\lim\limits_{\eta \to 0} \frac{a}{a_{\eta}} = 1 \, , \quad \left | \frac{a}{a_{\eta}} \right | \leq 1 \, 
%\end{equation} 
then we will get strong convergence of \cref{eq:def_set_symm} to \cref{eq:limit_Dn}, which will complete the proof.

Assume to the contrary that \cref{eq:bds_a_epsilon} fails, meaning there exists $\zeta > 0$ and a sequence $\eta_{m} \to 0$ for which 
\[
 \lim\limits_{\eta_m \to 0} \int\limits_{\T} f \left | \frac{a}{a_{\eta_m}} - 1 \right |^{\ell} >\zeta \, .
\]
Since the Hilbert transform $\mathrm{H}$ maps $L^1 (\T)$ into $L^{1, \infty} (\T)$, there exists an absolute constant $C$ such that for all $\nu, \eta > 0$ we have 
 \[
| \{ |\mathrm{H} \log |a| \mathbf{1}_{ \{|a| \leq \eta\} }| > \nu   \} | \leq \frac{C }{\nu} \int\limits_{\T} \log |a| \mathbf{1}_{ \{|a| \leq \eta\}  } \, . 
 \]
Thus as $m \to \infty$, we have 
\[
\mathrm{H} \log |a| \mathbf{1}_{ \{|a| \leq \eta_m\} } \to 0
\]
in measure. By passing to a subsequence, we may assume without loss of generality that this convergence holds pointwise almost everywhere on $\T$. 
From \cref{eq:outer_approx}, we also get $\lim\limits_{\eta \to 0} |\frac{a}{a_{\eta}}| = 1$ almost everywhere on $\T$

Combining with \cref{eq:outer_approx}, we obtain the pointwise almost everywhere equality \[
\lim\limits_{m \to \infty} \frac{a}{a_{\eta_m}} = \lim\limits_{m \to \infty} e^{(\Id-\I\mathrm{H}) \log |a| \mathbf{1}_{ \{|a| \leq \eta_m\} }} = 1 \, ,
\]
 as well as the inequality $| \frac{a}{a_{\eta_m}}| \leq 1$ almost everywhere. By the dominated convergence theorem, we obtain
\[
 \lim\limits_{\eta_m \to 0} \int\limits_{\T} f \left | \frac{a}{a_{\eta_m}} - 1 \right |^{\ell} =0 \, ,
\]
contradicting our assumption.

 %The upper bound in \cref{eq:bds_a_epsilon} follows easily from \cref{eq:outer_approx}, from which we also get $\lim\limits_{\eta \to 0} |\frac{a}{a_{\eta}}| = 1$ almost everywhere on $\T$. So the limit in \cref{eq:bds_a_epsilon} follows by showing the argument $\mathrm{H} \log |a| \mathbf{1}_{ \{|a| \leq \eta\} }$ of $\frac{a}{a_{\eta}}$ converges to $0$ almost everywhere. Since the Hilbert transform $\mathrm{H}$ maps $L^1 (\T)$ into $L^{1, \infty} (\T)$, there exists an absolute constant $C$ such that for all $\nu > 0$ we have 
 %\[
%| \{ |\mathrm{H} \log |a| \mathbf{1}_{ \{|a| \leq %\eta\} }| > \nu   \} | \leq \frac{C }{\nu} \int\limits_{\T} \log |a| \mathbf{1}_{ \{|a| \leq \eta\}  } \, . 
 %\]
 %By the dominated convergence theorem, the right side goes to $0$ as $\eta \to 0$. Thus $\mathrm{H} \log |a| \mathbf{1}_{ \{|a| \leq \eta\} }$ converges $0$ in measure, and hence pointwise almost everywhere as well. 
 \end{proof}
 Thus $\ddense$ is dense in $\Hil$. Define the unbounded operator $M: \ddense \to \Hil$ by
\begin{equation}\label{eq:def_extension_M}
M f := \lim\limits_{\eta \to 0} \mathcal{P}_{\Hil}  \begin{pmatrix}
    0 & \frac{b^*}{a^* _{\eta}} \\ - \frac{b}{a _{\eta}} & 0
\end{pmatrix} f \, ,
\end{equation}
where the limit above is understood to be in the weak sense in $\Hil$. Note $M$ has domain $\mathcal{D} (M) := \ddense$, and when $|a|$ is bounded uniformly below on $\T$, then this definition of $M$ coincides with the operator in \cref{eq:naive_operator}. 

We refer to \cite[Chapter 13]{rudin} for the spectral theory of unbounded operators on a Hilbert space.

\begin{lemma}\label{lem:iM_self_adjt}
Let $n \in \Z$ and let $(a,b) \in \mathbf{B}$.
The unbounded operator $i M$ is self-adjoint and hence has real spectrum. 
\end{lemma}
\begin{proof}
Recall from \cite[Definition 13.1]{rudin} that given the densely defined operator $M: \mathcal{D} (M) \to \Hil$, the domain $\mathcal{D} (M^*)$ of the adjoint operator $M^*$ consists of all those $g \in \Hil$ for which
\[
f \mapsto \langle M f, g \rangle
\]
is a continuous linear map on $\mathcal{D} (M)$ with respect to the norm of $\Hil$. And the adjoint operator $M^* :\mathcal{D} (M^*) \to \Hil$ is the unique densely defined operator for which
\[
\langle M f, g \rangle =  \langle f, M^* g \rangle 
\]
for all $f \in \mathcal{D} (M)$ and $g \in \mathcal{D} (M^*)$.

To show $\I M$ is self-adjoint \cite[Definition 13.3]{rudin}, we must check that for all $f, g \in \mathcal{D} (M)$, we have
\begin{equation}\label{eq:symmetric_defn}
\langle \I M f, g \rangle = \langle  f, \I M g \rangle \, ,
\end{equation}
and that the domain of definition of the adjoint operator $M ^* $ is a subset of that of $M$, i.e.,
\begin{equation}\label{eq:domain_defn_adjoint}
  \mathcal{D} ( M ^*)\subset \mathcal{D} ( M)\, .
\end{equation}

To see \cref{eq:symmetric_defn}, let $f, g \in \mathcal{D} (M) = \ddense$ and using the definition of $M$ we write
\begin{equation}\label{mfg}
\langle M f, g \rangle =
 \lim\limits_{\eta \to 0} \langle  \mathcal{P}_{\Hil} \begin{pmatrix}
    0 & \frac{b^*}{a^* _{\eta}} \\ - \frac{b}{a _{\eta}} & 0
\end{pmatrix}  f ,
g \rangle \, .
\end{equation} As $g \in \Hil$, we can get rid of the self-adjoint projection $\mathcal{P}_{\Hil}$ to write  \begin{equation}\label{eq:anti_symm_proof_1}
\langle   \begin{pmatrix}
    0 & \frac{b^*}{a^* _{\eta}} \\ - \frac{b}{a _{\eta}} & 0
\end{pmatrix}  f ,
g \rangle
= \langle     f ,
 \begin{pmatrix}
    0 & -\frac{b ^*}{a _{\eta} ^*} \\  \frac{b }{a _{\eta}} & 0
\end{pmatrix}g \rangle =  \langle     f ,
 \mathcal{P}_{\Hil}\begin{pmatrix}
    0 & -\frac{b ^*}{a _{\eta} ^*} \\  \frac{b }{a _{\eta}} & 0
\end{pmatrix}  g \rangle \, , 
\end{equation}
where in the last step we add $\mathcal{P}_{\Hil}$ back in because $f \in \Hil$.
Since $g \in \ddense$, the weak limit of the term in the second argument on the right side exists and so \cref{mfg} and \cref{eq:anti_symm_proof_1} yield
\[
 \langle  M f, g \rangle= \langle  f, - Mg \rangle \, ,
\]
or equivalently \cref{eq:symmetric_defn} follows.

As for \cref{eq:domain_defn_adjoint}, let $g \in \mathcal{D} ( M ^*)$. We need to show $g \in \ddense $, which means that there exists some $u \in \Hil$ such that for every $h \in \Hil$ we have
\begin{equation}\label{eq:weak_conv}
\lim\limits_{\eta \to 0} \langle \mathcal{P}_{\Hil} \begin{pmatrix}
    0 & \frac{b^*}{a_{\eta} ^*} \\
    -\frac{b}{a_{\eta} } & 0
\end{pmatrix} g , h \rangle = \langle  u, h \rangle \, . 
\end{equation}
Let $\eta > 0$, let $h \in \Hil$, and define 
\[
f_{\eta}:= \begin{pmatrix}
 \frac{a}{a_{\eta}} & 0 \\ 0 & \frac{a^*}{a_{\eta} ^*}
\end{pmatrix} h \in \dense \, .
\]
We compute adding and removing as before the operator $\mathcal{P}_{\Hil}$ freely where appropriate 
\begin{equation}\label{eq:massage_ip}
 - \langle \begin{pmatrix}
    0 & \frac{b^*}{a_{\eta} ^*} \\
    -\frac{b}{a_{\eta} } & 0
\end{pmatrix} g , h  \rangle=  \langle g, \begin{pmatrix}
    0 & \frac{b^*}{a_{\eta} ^*} \\
    -\frac{b}{a_{\eta} } & 0
\end{pmatrix} h \rangle =    \langle g,  \mathcal{P}_{\Hil} \begin{pmatrix}
    0 & \frac{b^*}{a_{\eta} ^*} \\
    -\frac{b}{a_{\eta} } & 0
\end{pmatrix} h  \rangle \, ,
\end{equation}
and using \cref{lem:limit_Dn}, the last term equals
\begin{equation}\label{eq:formal_dual_ip}
  \langle g, M f_{\eta}   \rangle =   \langle M^* g , f_{\eta}   \rangle 
 =  \langle  M ^*g ,  \begin{pmatrix}
 \frac{a}{a_{\eta}} & 0 \\ 0 & \frac{a^*}{a_{\eta} ^*}
\end{pmatrix} h  \rangle \, ,
\end{equation}
where the penultimate equality follows from the fact that $f_{\eta} \in \dense \subset \ddense = \mathcal{D}(M)$ and $g \in \mathcal{D} (M^*)$.
Thus letting $\eta \to 0$, the limit of the left side of \cref{eq:massage_ip} equals the negative of the limit of the right side of \cref{eq:formal_dual_ip}, which equals $\left \langle  M ^*g , h \right \rangle$ by the dominated convergence theorem and \cref{eq:bds_a_epsilon}, where we use that $h = (h_1, h_2)$ and $M^* h = ((M^* h)_1, (M^* h)_2) $ satisfy $(h_j )^*  (M ^*g)_j  \in L^1$ for $j=1,2$, . Thus \cref{eq:weak_conv} follows with $u = M ^*g$ in $\Hil$. Therefore 
\[
\mathcal{D} ( M ^* ) \subset  \mathcal{D} ( M )\, .
\]
Since $\I M$ is self-adjoint, then by \cite[Theorem 13.30]{rudin} it has real spectrum. 
\end{proof}

By \cref{lem:iM_self_adjt}, $M$ has pure imaginary spectrum, and so for any nonzero real $\lambda$, the operator $\Id  +\lambda M$ is a linear bijection from $\ddense$ onto $\Hil$ with bounded inverse \cite[Definition 13.26]{rudin}.  We also have the operator norm estimate 
\begin{equation}\label{eq:inverse_bded}
\left \| (\Id +\lambda M)^{-1} \right \| \leq 1 \, ,
\end{equation}
which follows by computing for any $x \in \ddense$, 
\[
 \left \| (\Id +\lambda M)x \right \| ^2 
 = \langle x,x \rangle + \langle \lambda M x, \lambda Mx \rangle 
 \]
 \[+ \langle \lambda Mx,x \rangle + \langle x,\lambda Mx \rangle  
  = \|x\|^2  + \|\lambda M x\|^2  \geq \|x\|^2 \, . 
\]
Plugging in $x = (\Id +\lambda M)^{-1} w$ for arbitrary $w \in \Hil$ yields \cref{eq:inverse_bded}. 

\begin{lemma}\label{lem:defn_AB}
 Let $(a,b) \in \mathbf{B}$, and let $f, g \in \Hil$. Then 
 \begin{equation}\label{eq:defn_AB}
 g= \left ( \Id + M \right )^{-1} f
 \end{equation}
 if and only if for all $u \in \dense$, we have
 \begin{equation}\label{eq:property_AB}
  \left \langle \left (\Id - M \right ) u, g \right \rangle = \left \langle u, f \right \rangle \, .
 \end{equation}
\end{lemma}
\begin{proof}
We first show \cref{eq:defn_AB} implies \cref{eq:property_AB}.  Let $f, g \in \Hil$ be as in \cref{eq:defn_AB}. Fix $u \in \dense$. 
   As $1$ is not in the spectrum of $M$, and $\dense\subset \ddense$,  we can write for the
   right side of \cref{eq:property_AB},
\begin{equation}\label{abuvxy}
    \langle (\Id -M)^{-1} (\Id -M) u, f \rangle = \langle  (\Id -M) u,  g \rangle \, ,
\end{equation}
where we used the fact that $(\Id -M)^{-1}$ has adjoint $(\Id + M)^{-1}$ and \cref{eq:defn_AB}.
This implies \cref{eq:property_AB}.

We now show \cref{eq:property_AB} implies \cref{eq:defn_AB}.
Assume $g$ satisfies \cref{eq:property_AB}.
For all $u \in \ddense$, we have
\[
\begin{pmatrix}
  \frac{1}{a_\eta} & 0 \\
  0 & \frac{1}{a_\eta ^* }
\end{pmatrix} u \in \Hil
\]
and therefore \[
\begin{pmatrix}
  \frac{a}{a_\eta} & 0 \\
  0 & \frac{a^*}{a_\eta ^* }
\end{pmatrix} u \in \dense
\]and so \cref{eq:property_AB} and then \cref{eq:limit_Dn} yield 
\[
\langle \begin{pmatrix}
  \frac{a}{a_\eta} & 0 \\
  0 & \frac{a^*}{a_\eta ^* }
\end{pmatrix} u , f \rangle  =  \langle (\Id - M) \begin{pmatrix}
  \frac{a}{a_\eta} & 0 \\
  0 & \frac{a^*}{a_\eta ^* }
\end{pmatrix} u , g \rangle =
 \langle \begin{pmatrix}
\frac{a}{a_{\eta}} & -\frac{b^*}{a_{\eta}^* } \\ \frac{b}{a_{\eta} }  &  \frac{a ^*}{a_{\eta} ^*}
\end{pmatrix}  u , g \rangle \, .
\]
Because $u \in \ddense$, then by the dominated convergence theorem and \cref{eq:bds_a_epsilon}, taking $\eta \to 0$ yields
\[
\langle u, f \rangle = \langle \left ( \Id - M \right )u , g \rangle
\]

For each $v \in \Hil$, there exists
$u \in \ddense$ for which
\[
 \left ( \Id  -M  \right ) u= v \, ,
\]
and hence
\[
\langle \left ( \Id - M \right )^{-1} v , f \rangle = \langle  v ,  g \rangle \, , 
\]
meaning \cref{eq:defn_AB} must hold.
\end{proof}

\subsection{Riemann-Hilbert factorization: Proof \texorpdfstring{of \cref{thm:RH_a_bded_below}}{}}

Without loss of generality, let $k=0$, as the general theorem can be obtained by factorizing 
\[
(a, b z^{-k}) = (a_{-}, b_{-}) (a_+, b_+)
\]
where $(a_-, b_-)$ is the NLFT of a sequence supported on $(-\infty, -1]$ and $(a_+, b_+)$ is the NLFT of a sequence supported on $[0, \infty)$, and then it suffices to define 
\[
(a_{<k}, b_{<k}) := (a_{-}, z^k b_{-}) \, , \quad (a_{\geq k}, b_{\geq k}) := (a_{+}, z^k b_+) \, . 
\]For simplicity of notation, we write $(a_{-}, b_{-})$ and $(a_{+}, b_{+})$ in place of $(a_{<0}, b_{<0})$ and $(a_{\geq 0}, b_{\geq 0})$.

To see uniqueness of the factorization, assume we have two factorizations \cref{eq: factorization} of $(a,b)$, i.e., for $j=1,2$ there exists $ (a_{-,j},b_{-,j}) \in \mathbf{H}_{\leq -1}$ and $ (a_{+,j},b_{+,j}) \in \mathbf{H}_{\geq 0}$  such that 
\[
 (a_{-,j},b_{-,j})(a_{+,j},b_{+,j}) = (a,b) \, . 
\]
Then for each $j$ we have
\begin{equation}\label{eq:prods_cst_terms_factorization}
a_{-, j} (\infty) a_{+,j} (\infty) = a (\infty) \, ,
\end{equation}
and also 
\[
(a_{-,j}, b_{-,j}) = (a,b)(a_{+,j} ^*, -b_{+,j}) \, ,
\]
from which it follows
\begin{equation}\label{eq:identities_factorization}
 \frac{a_{-, j}}{a}= a_{+,j} ^* + \frac{b}{a} b_{+, j} ^* \, , \qquad  -\frac{b_{-,j}}{a}=  b_{+,j}  - \frac{b}{a} a_{+, j}      \, .
\end{equation} 

Because $Mf$ coincides with \cref{eq:limit_Dn} for elements $f \in \dense$, we may write for all $\begin{pmatrix}
    A' \\ B'
\end{pmatrix} \in \dense$,
\begin{align}
\label{eq: linearfunctional motivation}
\langle (\Id - M)\begin{pmatrix}
    A' \\ B'
\end{pmatrix}, \begin{pmatrix}
    a_{+,j} \\ b_{+, j}
\end{pmatrix} \rangle  =  \langle \begin{pmatrix}
    1 & -\frac{b^*}{a^*} \\ \frac{b}{a} & 1
\end{pmatrix}\begin{pmatrix}
    A' \\ B'
\end{pmatrix}, \begin{pmatrix}
    a_{+,j} \\ b_{+, j}
\end{pmatrix} \rangle \, .
\end{align}
Applying \cref{eq:identities_factorization}, the mean value theorem, and then \cref{eq:prods_cst_terms_factorization}, this last term equals
\begin{align}
 \int\limits_{\T} A'\frac{a_{-,j} }{a} - B' \frac{b_{-, j} ^*}{a ^*} 
    = A' (\infty) \frac{a_{-,j} (\infty)}{a (\infty)} = \frac{A' \left ( \infty \right )}{a_{+,j} \left ( \infty \right )} \, , \nonumber 
\end{align}
and so by the mean value theorem, 
\[
\langle (\Id - M) \begin{pmatrix}
    A' \\ B'
\end{pmatrix}, a_{+,j}(\infty) \begin{pmatrix}
    a_{+, j} \\ b_{+, j}
\end{pmatrix}  \rangle = A' (\infty) = \langle \begin{pmatrix}
    A' \\ B'
\end{pmatrix}, \begin{pmatrix}
    1 \\ 0
\end{pmatrix} \rangle \, .
\]
Hence, by \cref{lem:defn_AB},
\[
a_{+,j}(\infty) \begin{pmatrix}
    a_{+, j} \\ b_{+, j}
\end{pmatrix} = (\Id + M)^{-1} \begin{pmatrix}
    1 \\ 0
\end{pmatrix} 
\]
and so
\[
a_{+,1}(\infty) \begin{pmatrix}
    a_{+, 1} \\ b_{+, 1}
\end{pmatrix} = a_{+,2}(\infty) \begin{pmatrix}
    a_{+, 2} \\ b_{+, 2}
\end{pmatrix} \, .
\] Since $(a_{+,j},b_{+,j}) \in \mathbf{H}_{\geq 0}$, we know that $|a_{+,j}|^2+|b_{+,j}|^2=1$, hence we deduce $a_{+,1} (\infty) = a_{+,2} (\infty)$ and $(a_{+,1},b_{+,1}) = (a_ {+,2},b_ {+,2})$. Since the left factors $ (a_{-,j},b_{-,j}) $ are uniquely determined by the right factors $ (a_{+,j},b_{+,j}) $, this implies the factorization is unique.

As for existence, first define 
\begin{equation}\label{eq:defn_AB_2}
\begin{pmatrix}
    A \\ B
\end{pmatrix} := (\Id + M)^{-1} \begin{pmatrix}
    1 \\ 0
\end{pmatrix} \, .
\end{equation}
Then we may write the real-valued $L^1 (\T)$ function
\begin{equation}\label{eq:modulus_fn}
f := A A^* + B B^* = A \left [ 1 - \lim\limits_{\eta\to 0} P_{\D} (\frac{b}{a_{\eta}} B^* ) \right ] + B^* \lim\limits_{\eta \to 0} P_{\D} (\frac{b}{a_{\eta}} A) \, ,
\end{equation}
where we used that
\begin{equation}\label{eq:fixed_pt}
\begin{pmatrix}
    A \\ B
\end{pmatrix} = \begin{pmatrix}
    1 \\ 0
\end{pmatrix}-  M\begin{pmatrix}
    A \\ B
\end{pmatrix}  = \begin{pmatrix}
    1 - \lim\limits_{\eta \to 0} P_{\D^*} \frac{b^*}{a _{\eta} ^*} B\\ 
    \lim\limits_{\eta \to 0} P_{\D} \frac{b}{a_{\eta}} A
\end{pmatrix}\, .
\end{equation}
Adding and subtracting $A B^* b a_{\eta} ^{-1}$ in \cref{eq:modulus_fn} then yields $f$ is the weak limit of 
\[
 A \left [ 1 + (\Id -  P_{\D}) (\frac{b}{a_{\eta}} B^* ) \right ] - B^* (\Id -  P_{\D}) (\frac{b}{a_{\eta}} A) \, .
\]
For a fixed $\eta$, the expression above is in $H^1  (\D^*)$. Since $f$ is a weak limit of such functions, then it must be the case that $f \in H^1 (\D^*)$. Because $f$ is real-valued, then this also implies $f = f^* \in H^1 (\D)$, and so $f$ is constant. 

Thus by the mean value theorem, 
\[
f(\infty) = \langle f, 1 \rangle = \lim\limits_{\eta\to 0} \langle A \left [ 1 + (\Id -  P_{\D}) (\frac{b}{a_{\eta}} B^* ) \right ] - B^* (\Id -  P_{\D}) (\frac{b}{a_{\eta}} A), 1 \rangle
\]
\[
= \lim\limits_{\eta \to 0} \langle A, 1 \rangle  = A (\infty)  \, .
\]
This also implies $A(\infty) > 0$ since
\[
A (\infty) = f (\infty) = \int f = \int |A|^2 + |B|^2 \geq 0
\]
with equality if and only if $(A,B)= (0,0)$, which cannot occur since $(A,B)$ is the image of $(1,0)$ under an injective linear map.

So now define
\begin{equation}\label{eq:little_to_big_A}
(a_+, b_+) := \frac{1}{A(\infty)^{\frac 1 2}}(A, B) \, ,
\end{equation}
which is an element of $H^2 (\D^*) \times H^2 (\D)$ and satisfies
\begin{equation}\label{eq:a+_SU2}
|a_+|^2 + |b_+|^2 = 1 
\end{equation}
on $\T$. Thus $(a_+, b_+) \in \bar{\mathbf{H}_{\geq 0}}$. We also note from \cref{eq:little_to_big_A}, the mean value theorem, and \cref{eq:a+_SU2}, that
\begin{equation}\label{eq:bound_A_infty}
 A (\infty)^{\frac 1 2}= a_{+} (\infty) \leq 1  \, .
\end{equation}
Thus we may define
\[
(a_{-}, b_{-}) := (a, b) (a_{+}, b_{+})^{-1}  = (a, b) (a_{+} ^*, -b_{+}) \, .
\]
Since $(a_{-}, b_{-})$ is the product of matrices in $SU(2)$, then
\[
|a_-|^2 + |b_-|^2 = 1 \, . 
\]
on $\T$. To check $(a_{-}, b_{-}) \in \bar{\mathbf{H}_{\leq -1} }$, using \cref{eq:fixed_pt} and \cref{eq:little_to_big_A} we write
\[
a_{-} = a a_{+} ^* + b b_{+} ^* = \lim\limits_{\eta \to 0} a ( \frac{1}{a_+ (\infty)} - P_{\D} (\frac{b}{a_{\eta}}b_{+}^*) ) + b b_{+} ^* 
= \lim\limits_{\eta \to 0} a ( \frac{1}{a_+ (\infty)} +(\Id- P_{\D}) (\frac{b}{a_{\eta}}b_{+}^*) ) 
\]
since we have the $L^2$ strong limit
\[
\lim\limits_{\eta \to 0} b b_{+} ^* (1 - \frac{a}{a_{\eta}}) = 0 \, .
\]
Thus $a_-$ is a weak limit of elements in $H^2 (\D^*)$ and so is in $H^2 (\D^*)$.

Similarly, we have
\[
b_{-} = - b_+ a + b a_+ = \lim\limits_{\eta \to 0} -a (P_{\D} \frac{b}{a_{\eta}} a_+ ) + b a_+ = \lim\limits_{\eta \to 0} a (\Id - P_{\D})(\frac{b}{a_{\eta}} a_{+})
\]
because again we have the $L^2$ strong limit
\[
 \lim\limits_{\eta \to 0} b a_{+}(1- \frac{a}{a_{\eta}}) =0 \, .
\]
 Thus, $b_{-}$ is the weak limit of elements in $H^2 (\D^*)$, and so $b_{-} \in H^2 (\D^*)$. Therefore, $(a_{-}, b_{-} ) \in \bar{\mathbf{H}_{\leq -1}}$.

Finally we must check that $(a_{+}, b_{+}) \in \mathbf{H}_{\geq 0}$ and $(a_{-}, b_{-}) \in \mathbf{H}_{\leq -1}$, namely we must verify that $a_{+} ^*$ and $b_+$ share no common inner factor, and likewise for $a_{-} ^*$ and $b_{-} ^*$. To see the first claim, suppose $g$ is a common inner factor for $a_{+} ^* $ and $ b_{+}$. Then
\begin{equation}\label{eq:inner_factor}
a^* g^{-1} = a_{-} ^* (a_+ ^* g^{-1}) - (b_+ g^{-1}) b_{-} ^* 
\end{equation}
is an $H^2 (\D)$ function. Thus $g$ is an inner factor of the outer function $a$, and so must be constant. Similarly, if $g$ is a common inner factor for  $a_{-} ^*$ and $b_{-} ^*$, then \cref{eq:inner_factor} yields again that $a^* g^{-1} \in H^2 (\D^*)$ and so $g$ is an inner factor for the outer function $a$, meaning $g$ must be constant. This completes the proof of \cref{thm:RH_a_bded_below}.

\subsection{Lipschitz estimate}

In this section, we show the map from $\frac{b}{a} \in L^2 (\T)$ to an individual nonlinear Fourier coefficient is Lipschitz continuous whenever the Szeg\H o condition is uniformly bounded below.
More precisely, given $S > 0$, let $\mathbf{B}_{S}$ consist of the pairs $(a,b) \in \mathbf{B}$ for which  
\begin{equation}
\int\limits_{\T} \log  |a (z) | > -S \, . 
\end{equation}

\begin{theorem}[Lipschitz estimate]\label{thm:Lip_bds}
    Let $S> 0$, and suppose $(a,b), (a',b') \in \mathbf{B}_{S}$ are the NLFTs of the sequences $F, F' \in \ell^2 (\Z)$, respectively. Then we have the Lipschitz bound
    \[
    \left \| F - F '\right \|_{\infty} \leq 2^{ \frac 1 2}e^{2S} (1 +e^{-S}) \left \| \frac{b}{a} - \frac{b'}{a'} \right \|_{L^2 (\T)} \, .
    \]
\end{theorem} 
%%END

We first show Lipschitz continuity of the map $(a,b) \mapsto (A,B)$, where $(A,B)$ is as defined in \cref{eq:defn_AB_2}.

\begin{lemma}\label{lem:Lip_bds_ratio}
Let $(a,b)$ and $(a',b')$ be elements of $\mathbf{B}$. Then 
\begin{equation}\label{eq:Lip_bds_init}
\left \| \begin{pmatrix} A \\ B \end{pmatrix}-  \begin{pmatrix}
   A ' \\ B ' 
\end{pmatrix} \right \|  \leq 2^{\frac 1 2} \min \{A (\infty) ^{\frac 1 2}, A ' (\infty) ^{\frac 1 2}  \}\left \| \frac{b'}{a'} - \frac{b}{a} \right \|_{L^2 (\T)}   \, .
\end{equation}
\end{lemma}

Combined with the fact that $ A (\infty), A ' (\infty)$ are both bounded in absolute value by $1$ as in \cref{eq:bound_A_infty}, then \cref{lem:Lip_bds_ratio} 
implies
\begin{equation}\label{eq:natural_Lip_bd} 
\left \| \begin{pmatrix} A ' \\ B ' \end{pmatrix}-  \begin{pmatrix}
   A  \\ B  
\end{pmatrix} \right \| \leq 2^{\frac 1 2} \left \| \frac{b'}{a'} - \frac{b}{a} \right \|_{L^2 (\T)} \, .
\end{equation}

\begin{proof}
Let $(a,b), (a',b') \in \mathbf{B}$. We write
\[
M := M_{(a,b)} = \lim\limits_{\eta \to 0} \begin{pmatrix}
    0 & \frac{b^*}{a_{\eta} ^*} \\ -\frac{b}{a_{\eta}} & 0 
\end{pmatrix} \, , \qquad M' := M_{(a',b')} = \lim\limits_{\eta \to 0} \begin{pmatrix}
    0 & \frac{(b')^*}{(a_{\eta} ') ^*} \\ -\frac{b'}{a_{\eta}'} & 0 
\end{pmatrix} \, ,
\]
and 
\[
\dense := \begin{pmatrix}
    a & 0 \\ 0 & a^*
\end{pmatrix} \Hil \, , \qquad  \dense ' := \begin{pmatrix}
    a' & 0 \\ 0 & (a ')^*
\end{pmatrix} \Hil \, .
\]
Define $\dddense$ to be the image of $\dense \cap \dense'$ under $\Id - M$. We claim $\dddense$ is dense in $\Hil$ with respect to the weak topology. Indeed, let $f \in \Hil$. Then $f$ is the image of some $h \in \mathcal{E}$ under $\Id-M$, so we in fact have
\[
 \lim\limits_{\eta \to 0} \begin{pmatrix}
    \frac{a}{a _{\eta}} & 0 \\
   0 & \frac{a^*}{a _{\eta} ^*} 
\end{pmatrix} h -  \mathcal{P}_{\Hil} \begin{pmatrix}
    0 & \frac{b^*}{a ^*} \\
   - \frac{b}{a} & 0 
\end{pmatrix} \begin{pmatrix}
    \frac{a}{a _{\eta}} & 0 \\
   0 & \frac{a^*}{a _{\eta} ^*} 
\end{pmatrix} h= (\Id- M) h = f \, . 
\]
Thus $f$ is the weak limit of elements
\[
(\Id - M) \begin{pmatrix}
    \frac{a}{a _{\eta}} & 0 \\
   0 & \frac{a^*}{a _{\eta} ^*} 
\end{pmatrix} h \, , 
\]
which are the images of elements in $\dense$ under $\Id - M$. But fixing $\eta$, we may write this element of $\Hil$ as
\[
\mathcal{P}_{\Hil} \begin{pmatrix}
    \frac{a}{a _{\eta}} & \frac{b^*}{a _{\eta} ^*} \\
   -\frac{b}{a _{\eta} } & \frac{a^*}{a _{\eta} ^*} 
\end{pmatrix} h 
\]
which is the strong limit of 
\[\mathcal{P}_{\Hil} \begin{pmatrix}
    \frac{a}{a _{\eta}} & \frac{b^*}{a _{\eta} ^*} \\
   -\frac{b}{a _{\eta} } & \frac{a^*}{a _{\eta} ^*} 
\end{pmatrix} \begin{pmatrix}
    \frac{a' }{a _{\gamma}'} & 0 \\ 0 & \left ( \frac{a'}{a _{\gamma} '} \right )^* 
\end{pmatrix} h =  (\Id - M)\begin{pmatrix}
    \frac{a' }{a _{\gamma}'} \frac{a}{a _{\eta}} & 0 \\ 0 & \left ( \frac{a'}{a _{\gamma} '} \right )^* \left ( \frac{a}{a _{\eta} } \right )^* 
\end{pmatrix} h \, , 
\]
as $\gamma \to 0$.
Thus the collection
\[
\left \{ (\Id - M)\begin{pmatrix}
    \frac{a' }{a _{\gamma}'} \frac{a}{a _{\eta}} & 0 \\ 0 & \left ( \frac{a'}{a _{\gamma} '} \right )^* \left ( \frac{a}{a _{\eta} } \right )^* 
\end{pmatrix} h \right \}_{0 < \eta, \gamma < 1}
\]
has limit point $f$ in $\Hil$ under the weak topology. Since this collection is the image of a collection in $\dense  \cap \dense '$, this completes the proof of the claim.

With this in mind, we can write
\[
\left \| \begin{pmatrix}
    A \\ B 
\end{pmatrix} - \begin{pmatrix}
    A ' \\ B ' 
\end{pmatrix} \right \| = \sup \left | \langle  \begin{pmatrix}
    C \\ D 
\end{pmatrix}, \begin{pmatrix}
    A \\ B 
\end{pmatrix} - \begin{pmatrix}
    A ' \\ B ' 
\end{pmatrix} \rangle \right| \, ,
\]
where the supremum is taken over elements $\begin{pmatrix}
    C \\ D 
\end{pmatrix} \in \dddense$ with norm at most one. 
But for any such $\begin{pmatrix}
    C \\ D 
\end{pmatrix}$, we have
\[
\langle \begin{pmatrix}
    C \\ D 
\end{pmatrix}, \begin{pmatrix}
    A \\ B 
\end{pmatrix} - \begin{pmatrix}
    A ' \\ B ' 
\end{pmatrix} \rangle 
= \langle \begin{pmatrix}
    C \\ D 
\end{pmatrix}, \left ( \left ( \Id + M  \right )^{-1}   -  \left ( \Id + M'  \right )^{-1} \right)  \begin{pmatrix}
    1 \\ 0 
\end{pmatrix}\rangle \, , 
\]
\[
=\langle \begin{pmatrix}
    C \\ D 
\end{pmatrix}, \left ( \Id + M  \right )^{-1} \left ( M' - M   \right ) \left ( \Id + M'  \right )^{-1}  \begin{pmatrix}
    1 \\ 0 
\end{pmatrix}\rangle \, .
\]
By duality, this equals
\[
\langle \left ( M  - M'   \right ) \left ( \Id - M  \right )^{-1} \begin{pmatrix}
    C \\ D 
\end{pmatrix},   \left ( \Id + M ' \right )^{-1}  \begin{pmatrix}
    1 \\ 0 
\end{pmatrix}\rangle \, .
\]
Because $\begin{pmatrix}
    C \\ D
\end{pmatrix} \in \dddense$, then \REV{$(\Id - M)^{-1} \begin{pmatrix}
    C \\ D
\end{pmatrix} \in \dense \cap \dense '$} and so by \cref{eq:limit_Dn} we may write this last inner product as
\[
\langle \begin{pmatrix}
    0 & (\frac{b}{a} - \frac{b' }{a' } )^* \\
    -(\frac{b}{a} - \frac{b'}{a' } ) & 0
\end{pmatrix} \left ( \Id - M  \right )^{-1} \begin{pmatrix}
    C \\ D 
\end{pmatrix},   \left ( \Id + M ' \right )^{-1}  \begin{pmatrix}
    1 \\ 0 
\end{pmatrix}\rangle \, , 
\]
\[
= 
\langle \begin{pmatrix}
    0 & (\frac{b}{a} - \frac{b'}{a'} )^* \\
    -(\frac{b}{a} - \frac{b' }{a' } ) & 0
\end{pmatrix} \left ( \Id - M  \right )^{-1} \begin{pmatrix}
    C \\ D 
\end{pmatrix},   \begin{pmatrix}
    A ' \\ B ' 
\end{pmatrix}\rangle \, ,
\]
or rather
\[
=
\langle  \left ( \Id - M  \right )^{-1} \begin{pmatrix}
    C \\ D 
\end{pmatrix},  \begin{pmatrix}
    0 & -(\frac{b}{a} - \frac{b' }{a' } )^* \\
    (\frac{b}{a} - \frac{b' }{a' } ) & 0
\end{pmatrix} \begin{pmatrix}
    A ' \\ B ' 
\end{pmatrix}\rangle \, ,
\]

By Cauchy-Schwarz, the operator norm bound \cref{eq:inverse_bded} and the norm of $\begin{pmatrix}
    C \\ D 
\end{pmatrix}$ being bounded by $1$, this last inner product is bounded in absolute value by 
\[
\left \|  \begin{pmatrix}
    0 & \frac{b'}{a'} - \frac{b}{a}  \\ -\left (\frac{b'}{a'} - \frac{b}{a}\right )^* & 0 
\end{pmatrix} \begin{pmatrix}
    A ' \\ B ' 
\end{pmatrix} \right \|  \, ,
\]
which by H\"older's inequality is at most
\[
 2^{\frac{1}{2}} \left \| \frac{b'}{a' } - \frac{b}{a} \right \|_{L^2 (\T)}  \max \{ \| A ' \|_{L^\infty (\T)}, \| B ' \|_{L^\infty (\T)} \} \, . 
\]
But by \cref{eq:little_to_big_A}, we have
\[
 A '  (\infty) ^{-\frac 1 2}\begin{pmatrix}
    A ' \\ B ' 
\end{pmatrix}
\]
is the NLFT of some sequence, and so has components all bounded above in absolute value by $1$.
Hence
\[
\left \| \begin{pmatrix} A ' \\ B ' \end{pmatrix}-  \begin{pmatrix}
   A  \\ B  
\end{pmatrix} \right \| \leq 2^{\frac 1 2} A ' (\infty) ^{\frac 1 2} \left \| \frac{b'}{a'} - \frac{b}{a} \right \|_{L^2 (\T)} \, .
\]
Finally, \cref{eq:Lip_bds_init} follows by symmetry.
\end{proof}
We can now prove \cref{thm:Lip_bds}. To bound the infinity norm of $F-F'$, we must bound
\[
\left | F_n - F_n ' \right | 
\]
uniformly in $n$. Without loss of generality, take $n$ to be $0$. By \cite[(6.13)]{AlexisMnatsakanyanThiele2023} and then \cref{eq:little_to_big_A}, we may write
\[
| F_0 - F_0 '| = \left | \frac{b_+ (0)}{a_+ ^* (0)} - \frac{b_+ '(0)}{(a_+ ') ^* (0)} \right | = \left | \frac{B(0)}{A ^* (0)} - \frac{B '(0)}{(A ') ^* (0)} \right | \, , 
\]
which, after putting everything on the same denominator, is at most
\[
\left |  \frac{ B(0) - B ' (0)}{A ^* (0) } \right | + \left | \frac{ B '(0) (A ^* (0) - A (0) )  }{A ^* (0) (A ') ^* (0)} \right | \, . 
\]
By the mean value theorem and then Cauchy-Schwarz, the above is at most \REV{
\begin{equation}\label{eq:common_denom_split}
A (\infty)^{- 1} \left \|B - B'  \right \|_{L^2 (\T)} + A (\infty)^{- 1} A ' (\infty)^{- \frac 1 2} \left | \frac{B ' (0)}{(A ' )  (\infty) ^{\frac 1 2}} \right | \left \|A - A'  \right \|_{L^2 (\T)} \, .
\end{equation}}
Using the fact that $ (A ' )  (\infty) ^{- \frac 1 2} (A', B' )$ has entries bounded by $1$ in absolute value, \cref{eq:common_denom_split} is at most 
\[
A (\infty)^{- 1} \left \|B - B'  \right \|_{L^2 (\T)} + A (\infty)^{- 1} A ' (\infty)^{- \frac 1 2} \left \|A - A'  \right \|_{L^2 (\T)} \, . 
\]
Because $A (\infty), A ' (\infty) \leq 1$ by \cref{eq:bound_A_infty}, then we can bound this last expression by
\[
A (\infty) ^{- 1}(1+ A '(\infty) ^{- \frac 1 2}) \left \| \begin{pmatrix}
    A \\ B
\end{pmatrix} - \begin{pmatrix}
    A ' \\ B '
\end{pmatrix} \right \| \, .
\]
By \cref{lem:Lip_bds_ratio}, this is bounded by
\[
2^{ \frac 1 2} A (\infty) ^{- \frac 1 2} (1+  A '(\infty) ^{- \frac 1 2}) \left \| \frac{b}{a} - \frac{b'}{a'} \right \|_{L^2 (\T)}
\]

By \cref{eq:bound_A_infty} and \cref{eq:prods_cst_terms_factorization}, we deduce
\begin{equation}\label{eq:lower_bd_A}
A (\infty) ^{\frac 1 2}   \geq a (\infty)   > e^{-S} \, ,
\end{equation}
and similarly for $A' (\infty)$,
which then yields the estimate in the Lemma.

\subsection{Plancherel equality and maximal solution}\label{sec:maximal}

We remark that the particular choice of $a(z)$ given by \cref{thm:construct_a} actually corresponds to the maximal solution, which is a special class of solutions with symmetric phase-factor  proposed in \cite{WangDongLin2021} and enjoys many desirable properties. To reveal the connection between the function $a(z)$ constructed by \cref{thm:construct_a} and the maximal solution, we first review the construction of the $SU(2)$ matrix corresponding to the maximal solution. Without loss of generality, let \REV{$f(x)$} be a real even target polynomial of degree $2d$. Then we factorize
\begin{equation}
\mathfrak{F}(z):=1-\abs{f\left(\frac{z+z^{-1}}{2}\right)}^2=\alpha \prod_{i=1}^{4d} (z-r_i)(z^{-1}-r_i).
\end{equation}
Suppose that $\norm{f}_\infty<1$, then for $\mathfrak{F}(z)$ there is no root on $\T$. To construct the maximal solution, we take $\mc{D}=\set{r_i}_{i=1}^{4d}$ to be the set of all roots of $\mathfrak{F}$ in the unit disk $\mathbb{D}$. Then
\begin{equation}
\wt{a}(z)=\sqrt{\alpha}\prod_{i=1}^{4d}(z-r_i),
\end{equation}
is a polynomial and is in $H^2(\mathbb{D})$, and it satisfies $\abs{\wt{a}(z)}^2=\mathfrak{F}(z)$. Following \cite[Theorem 4]{WangDongLin2021}, we construct
\begin{equation*}
    \begin{pmatrix}
        \frac{\wt{a}(z) + \wt{a}^*(z)}{2} + f\left(\frac{z+z^{-1}}{2}\right) & \frac{\wt{a}(z) - \wt{a}^*(z)}{2}\\
        \frac{\wt{a}(z) - \wt{a}^*(z)}{2} & \frac{\wt{a}(z) + \wt{a}^*(z)}{2} - f\left(\frac{z+z^{-1}}{2}\right)
    \end{pmatrix}
\end{equation*}
which is the unitary matrix associated with the maximal solution. It is exactly the unitary matrix $U_d(x, \Psi)$ associated with $(a(z), b(z))$ by exploiting $a(z^2)=z^{-4d} \wt{a}(z)$ and verifying \cref{eqn:QSP_NLFT_connection}. \REV{Note in particular, that the polynomial $a^*$ has no roots in $\overline{\D}$ and therefore must be an outer function. As outer functions are determined by their absolute value on $\T$ by our definition, $a^*$ must be the same outer function constructed in the Weiss algorithm. Hence, our algorithm constructs the maximal solution with theoretical guarantees.}

We will need the following identity for later.
\begin{lemma} If $(a,b)$ is the NLFT of a sequence $F \in \ell^2 (\Z)$, then 
\begin{equation}\label{eq:zero_a}
   a^* (0) = \prod\limits_{j} (1 + |F_j|^2)^{- \frac 1 2} > 0 \, . 
\end{equation}
\end{lemma}
\begin{proof}
    We first verify it in the case that $F$ has compact support. From \cref{eq:defn_NLFT_gen}, we have
\[
\prod\limits_{j} \left ( 1+ |F_j| \right )^{\frac 1 2}
\begin{pmatrix}
    a & b \\ - b^* & a
\end{pmatrix} =  \prod\limits_j \begin{pmatrix} 1 & F_j z^j \\ -\overline{F_j}z ^{-j} & 1 
\end{pmatrix}  =  \prod\limits_j \left ( \begin{pmatrix} 1 & 0 \\ 0 & 1 
\end{pmatrix} + \begin{pmatrix} 0 & F_j z^j \\ -\overline{F_j} z^{-j} & 0 
\end{pmatrix} \right ) \, ,
\]
where the product of matrices is understood as lower indexed matrices are to the left of higher-index matrices. Then doing a binomial expansion, this last term equals
\[
\sum\limits_{n \geq 0} \sum\limits_{k_1 < \ldots < k_n} \prod\limits_{j=1}^n \begin{pmatrix} 0 & F_{k_j} z^{k_j} \\ -\overline{F_{k_j}} z^{-{k_j}} & 0 
\end{pmatrix} \, . 
\]
In the sum above, the terms corresponding to $n$ even are diagonal, while the terms with $n$ odd are anti-diagonal. Thus
\[
 \prod\limits_{j} \left ( 1+ |F_j| \right )^{\frac 1 2} a(z)  = \sum\limits_{n \text{ even }} \sum\limits_{k_1 < \ldots < k_n} F_{k_1} z^{k_1} (-\overline{F_{k_1}} z^{-k_2}) \ldots F_{k_{n-1}} z^{k_{n-1}} (-\overline{F_{k_n}} z^{-k_n}) \, ,
\]
which is a Laurent polynomial with constant term $1$,  corresponding to $n=0$. Thus $a$ has constant term $\prod\limits_{j} (1 + |F_j|^2)^{- \frac 1 2}$, i.e., \cref{eq:zero_a} holds when $F$ has compact support. A limiting argument then allows us to extend \cref{eq:zero_a} to any $F \in \ell^2 (\Z)$.
\end{proof}

We have the nonlinear Plancherel inequality below.
\begin{lemma}
    If $(a,b)$ is the NLFT of some $F \in\ell^2 (\Z)$, then
\begin{equation}\label{eq:Plancherel_ineq}
\sum\limits_{n} (1 + |F_n|^2) \geq - \int\limits_{\T} \log (1- |b(z)|^2)  \, , 
\end{equation}
where equality holds if and only if $a^*$ is outer.
\end{lemma}
\begin{proof}
If $(a,b)$ is the NLFT of a sequence $F \in \ell^2 (\Z)$, then by the inner-outer factorization theorem for $H^{\infty} (\D)$ functions \cite[Corollary 5.6, Chapter II]{garnett}, we may write $a^* = I O$ where $I$ and $O$ are inner and outer functions on $\T$, respectively. By \cref{eq:det_cond_1} we may write
\[
\int\limits_{\T} \log (1- |b|^2) =  2\int\limits_{\T} \log |a^*| = 2 \int\limits_{\T} \log |O| = 2 \log |O (0)|
\]
where in the last equality, we used the mean value theorem for the harmonic function $\log |O(\cdot)|$. Then this last term equals
\[
2 \log |a^* (0)| - 2 \log |I(0)|= -\sum\limits_{n} \log (1 + |F_n|^2) - 2 \log |I(0)| \, ,
\]
where the last equality follows from \cref{eq:zero_a}. 
Thus as in \cite[(3.1) and the Remark on p.16]{tsai}, we have the nonlinear Plancherel identity
\begin{equation}\label{eq:nonlinear_Planch_finite}
 \sum\limits_{n} \log (1 + |F_n|^2) = - \int\limits_{\T} \log (1- |b(z)|^2) - 2 \log |I(0)| \, .
\end{equation}

 In particular, for every NLFT $(a,b)$, because $|I(0)|\leq \|I\|_{\infty} = 1$ by the maximum principle, we have the nonlinear Plancherel inequality \cref{eq:Plancherel_ineq}, where equality holds if and only if $|I (0)| = 1$. By the maximum principle, this can only occur if $I$ equals a unimodular constant $\lambda$. Because $a^* = \lambda O$, the proof of the lemma will be completed once we show $\lambda =1$, since then $a^*$ will equal the outer function $O$. But $a^* (0) > 0$ by \cref{eq:zero_a}. And $O (0)  >0 $ because
 \[
 \log O (0) = (\log |O| + \I H \log |O|) (0) 
 \]
 is the constant Fourier coefficient of $\log |O| + \I H \log |O|$, which is the constant Fourier coefficient of the real-valued function $\log |O|$, and so is real-valued. Thus the unimodular constant $\lambda = \frac{a^* (0)}{O (0)} > 0$, meaning $\lambda =1$. This completes the proof. 
\end{proof}

Given $b$, the function $a$  constructed in \cref{thm:construct_a} satisfies $a^*$ is outer, as can be seen in \cref{eq:a_G}. Thus given $b$, $a$ is the unique function for which we have equality in \cref{eq:Plancherel_ineq}, i.e.,
\[
\sum\limits_{n} (1 + |F_n|^2) = -\int\limits_{\T} \log (1- |b(z)|^2) \, .
\]
This not only yields \cref{plancherel} for iQSP after the appropriate change of variables, but also justifies the name ``maximal solution'' from the perspective of nonlinear Fourier analysis, since $a$ is the unique function for which $(a,b)$ is the NLFT of a sequence $F \in \ell^2 (\Z)$ for which 
\[
-\sum\limits_{n} (1 + |F_n|^2) 
\]
achieves its maximum value of
\[
\int\limits_{\T} \log (1- |b(z)|^2) \, .
\]

\subsection{Proof of \texorpdfstring{\cref{thm:main}}{main Theorem}}
We first show the existence of $\Psi$.
Given $x \in [0,1]$, let $\theta \in [0, \frac{\pi }{2}]$ be the unique number for which 
\[
x = \cos \theta\, , 
\]
and then set
\[
z \equiv e^{2 \I \theta} \, .
\]
Define
\[
b (z)\equiv \I f(x) \, 
\]
for $z \in \T \cap \C^{+}$ and extend $b$ evenly across to the lower-half plane
as in \cite[Section 4]{AlexisMnatsakanyanThiele2023}. 
Then $b$ is bounded in absolute value away from $1$.
By \cref{thm:construct_a}, there exists an outer function $a^*$ on $\D$ such that the pair $(a,b) \in \mathbf{B}$. Applying \cref{thm:RH_a_bded_below}, we find a factorization
 \begin{equation}\label{eq: factorization2}
        (a,b) = (a_-,b_-) (a_+, b_+)
    \end{equation}
with  $(a_{-}, b_{-}) \in \mathbf{H}_{0} ^*$ and $(a_{+}, b_{+}) \in \mathbf{H}$. By \cite[Theorems 2 and 9]{AlexisMnatsakanyanThiele2023}, $(a_{-}, b_{-})$ and $(a_{+}, b_{+})$ are the nonlinear Fourier transforms (NLFTs) of sequences supported on the negative and nonnegative integers, respectively. Let $\{F_k\}$ be the sum of these sequences. Observe that $(a,b)$ is the NLFT of $\{F_k\}$, see \cite[Section 7]{AlexisMnatsakanyanThiele2023}. 
Define $\{\psi_k\} \in \mathbf{P}$ by 
\begin{equation}\label{eq:defn_F_psi}
F_k \equiv \I \tan \psi_{|k|} \, .
\end{equation}
Then \cite[Sections 4 and 8]{AlexisMnatsakanyanThiele2023} shows $\{\psi_k\}$ is the sequence of phase factors associated to the signal $f$, which shows the existence part of \cref{thm:main}.

Uniqueness follows similarly from the uniqueness of \cref{thm:RH_a_bded_below} and the argument in
\cite[paragraphs between (8.1) and (8.2)]{AlexisMnatsakanyanThiele2023}.

A careful reading of this section up till now, along with the existence and uniqueness proofs above yields the following result, which is a more comprehensive version of \cref{lem:f_to_phase_factor}.

\begin{lemma}
Let $k \in \mathbb{N}$. Given any $f \in \mathbf{S}$, we can recover the phase factor $\psi_k$ via the maps
    \begin{equation}
f \mapsto \frac{b}{a} \mapsto (A_k ,B_k) \mapsto F_k \mapsto \psi_k \, , 
\end{equation}
where 
\[
F_k := \frac{(B_k z^{-k}) (0)}{A_k ^*(0)} \, ,
\]
and
\[
\psi_k := \arctan\left(-\I F_k\right) \, ,
\]
and where $(A_k, B_k)$ is the unique element of $H^2 (\D^*) \times z^{k} H^2 (\D)$ satisfying
\begin{equation}
(\Id +M_k) \begin{pmatrix}
    A_k \\ B_k
\end{pmatrix} = \begin{pmatrix}
    1 \\ 0
\end{pmatrix} \, ,
\end{equation}
in which
\[
M_k = \lim\limits_{\eta \to 0} \begin{pmatrix}
    0 & P_{\D^*}  \frac{b ^* }{a_{\eta} ^*} \\ - z^k  P_{\D} z^{-k} \frac{b}{a_{\eta}} & 0
\end{pmatrix} \, .
\]

% Alternatively, $\begin{pmatrix}
%     A_k \\ B_k
% \end{pmatrix}$ is the unique element of $H^2 (\D^*) \times z^k H^2 (\D)$ satisfying 
% \begin{equation}\label{eq:Riesz_formula}
% \int\limits_{\T} U (A_k ^*  + \frac{b}{a} B_k ^*) + V (A_k ^* - \frac{b^*}{a^*} B_k ^*) = \langle (\Id - M_k ) \begin{pmatrix}
%     U \\ V
% \end{pmatrix}, \begin{pmatrix}
%     A_k \\ B_k
% \end{pmatrix} \rangle =  
% U (\infty)
% \end{equation}
% for all 
% \[
% \begin{pmatrix}
%     U \\ V
% \end{pmatrix} \in a H^2 (\D^*) \times z^k H^2 (\D) \, .
% \]
\end{lemma}

As for the Lipschitz bounds, we proceed similar to \cite[Section 8]{AlexisMnatsakanyanThiele2023}, but with some minor changes. We begin by describing how to recover any phase factor $\psi_k$ from a signal $f$. To compute the Lipschitz constant of
\[
f \mapsto \mathbf{\Psi} \, ,
\]
it suffices to compute the Lipschitz constant of the map
\[
f \mapsto \psi_k 
\]
for arbitrary $k$. Assume without loss of generality that $k=0$, and write this last map as the composition of maps
\[
f \mapsto \frac{b}{a} \mapsto F_0 \mapsto \psi_0 \, . 
\]
We compute the Lipschitz constant for each of these maps, noting that if $f \in \mathbf{S}_{\eta}$ for $0<\eta < \frac{1}{2}$, then by \cref{eq:lower_bd_A} the resulting pair $(a,b)$ satisfies
\[
|a (z) | > \delta \, 
\]
for all $z \in \T$, where
\begin{equation}\label{eq:delta}
\delta := \sqrt{\frac{3}{2} } \eta^{ \frac 1 2}  \, .
\end{equation}
In particular, we have $(a,b) \in \mathbf{B}_{S}$ for $e^{-S} := \delta$.

As $\arctan (x)$ has slope between $-1$ and $1$, by \cref{eq:defn_F_psi} we have
\[
|\psi_0 - \tilde{\psi}_0 | \leq |F_0 -\tilde{F}_0| \, ,
\]
so the last of the maps has Lipschitz constant at most $1$. And by \cref{thm:Lip_bds}, the middle map sending $\frac{b}{a}$ to $F_0$ has Lipschitz constant at most
\[
2^{ \frac 1 2} \delta^{-2} (1+ \delta)  \, .
\]

As for the Lipschitz constant of the map sending $f$ to $\frac{b}{a}$,
we write 
\[
\frac{b'}{a'} - \frac{b}{a}  = \frac{b' - b}{a'} + \frac{b(a-a')}{a'a} \, ,
\]
so that we can estimate
\[
\left \| \frac{b'}{a'} - \frac{b}{a} \right \|_{L^2 (\T)} \leq \delta^{-1} \left \| b' - b \right \|_{L^2 (\T)} + \delta^{-2} \left \| a- a' \right \|_{L^2 (\T)}  
\]

By \cite[(8.5)]{AlexisMnatsakanyanThiele2023}, we have
\[
\left \| a - a' \right \|_{L^2 (\T)} \leq \left \| \sqrt{1 - |b|^2} - \sqrt{1- |b'|^2} \right \|_{ L^2 (\T)} + \frac{1}{4} \left \| \log \left | 1 - |b|^2 \right | - \log \left | 1- |b'|^2 \right | \right \|_{ L^2 (\T)} \, . 
\]
% \ma{estimates below should be checked by another person other than myself :)}\JS{I will check it again later.} 
Recall $|b| = |f|$ takes values in $[0, 1-\eta]$. By the mean value theorem, the Lipschitz constants of the functions $\sqrt{1-x^2}$ and $\log (1-x^2)$ on the interval $[0, 1-\eta]$ are at most $ \delta^{-1} $ and $2\delta^{-2}$, respectively. Thus we obtain
\[
 \left \| a - a' \right \|_{L^2 (\T)} \leq (\delta^{-1} + \frac{ \delta^{-2}}{2}  ) \left \| |b| - |b'| \right \|_{ L^2 (\T)} \leq  (\delta^{-1} + \frac{\delta^{-2}}{2}  ) \left \| b - b' \right \|_{ L^2 (\T)} \, ,
\]
and therefore
\[
\left \| \frac{b'}{a'} - \frac{b}{a} \right \|_{L^2 (\T)} \leq (\delta^{-1} + \delta^{-3} +  \frac{ \delta^{-4}}{2}   )   \left \| b- b' \right \|_{L^2 (\T)} 
\]
Putting everything together, and in particular applying \cref{thm:RH_a_bded_below} with parameter $\delta$ given by \cref{eq:delta} and plugging in \cref{eq:delta}
yields that when $f \neq f'$, we have
\[
\frac{|\psi_0 - \psi_0 ' |}{ \left \| f - f' \right \|_{\mathbf{S}}} \leq 1 \cdot 2^{\frac 1 2} \delta^{-2} (1+\delta) \cdot (\delta^{-1} + \delta^{-3} +  \frac{ \delta^{-4}}{2}   )   \leq \delta^{-6} 2^{\frac 1 2} (1+\delta) \cdot (\delta^{3} + \delta+  \frac{1}{2} ) \, . 
\]
By \cref{eq:delta}, and using the fact that $0< \eta < \frac{1}{2}$ also implies $0<\delta <\frac{\sqrt{3}}{2}$, we have the above is at most
\[\eta ^{-3} \left ( \frac{2}{3} \right )^{3} 2^{\frac 1 2} (1+ \delta)(\delta^{3} + \delta+  \frac{1}{2} ) \leq  \eta ^{-3} \cdot \frac{8}{27} 2^{\frac 1 2} \cdot \left (1 + \frac{\sqrt{3}}{2} \right ) \cdot \left ( \left (\frac{\sqrt{3}}{2} \right  )^3 + (\frac{\sqrt{3}}{2}) + \frac{1}{2} \right ) \leq 1.6 \eta^{-3} \, .
\]

\section{Complexity analysis of Riemann-Hilbert-Weiss algorithm}\label{sec:complexity_analysis}
In this section we prove \cref{thm:main_alg}. The proof is given in three parts: the error, complexity and bit analyses. The formal statement of the error analysis is presented in \cref{thm:sufficient_condition_N}, whose proof is displayed in Sections \ref{fft error}-\ref{proof of thm 8}. In \cref{computational cost proof}, we establish the computational cost bound and in \cref{bit requirement proof}, we prove the bit requirement part.

Henceforth, we fix $d\in \N$, $0<\eta <\frac{1}{2}$,  and $b$ to be a Laurent polynomial of degree $d$ satisfying $b(z)=b(z^{-1})$ and $\norm{b}_{\infty} \leq 1-\eta$, as in the assumption of the theorem. Also fix $0\leq k \leq d$, $0<\epsilon<1$ and $N$ an even integer as in \cref{thm:sufficient_condition_N}.

Let $a^*$ be the unique outer function with $|a^* (z)|= \sqrt{1-|b(z)|^2}$ for $z\in \T$ and recall 
(\cite{tsai}) that because of 
\begin{equation}
   aa^*=1-bb^* 
\end{equation}
we have that $a^*$ is a polynomial of degree $2d$ whose zeros with multiplicities are those of $1-bb^*$ that are outside the closed unit disc.

Let $\{z_\ell \}_{0\le \ell <N}$ be the $N$th roots of unity
in natural counterclockwise order. For a continuous function $u$ on the unit circle,
let $\mathcal{F}(u)$ be the Fourier transform
of $u$ as function on $\T$ and $\mathcal{F}^N(u)$ the Fourier transform of $u$ on the finite group $\{z_\ell \}_{0\le \ell <N}$,
that is,
\begin{equation}
    \mathcal{F}(u)(j) := \int_{\mathbb{T}} u(z) z^{-j}, \quad \mathcal{F}^N (u)(j): = \frac{1}{N} \sum_{\ell=0}^{N-1} z_{\ell}^{-j} u(z_{\ell}).
\end{equation}
% Our argument works for any sufficiently large integer $N$, but if $N$ is a power of two, the coefficients 
% $\mathcal{F}^N(u)$ can be most efficiently computed by the Fast Fourier transform. As the estimates in our argument improve with increasing $N$, it is most efficient to choose $N$ the next larger power of $2$ from the threshold that guarantees the desired accuracy.
Our argument works for any sufficiently large integer $N$, and the coefficients 
$\mathcal{F}^N(u)$ can be efficiently computed by the Fast Fourier transform. 
We shall assume $0<\eta< \frac 12$ and  $N>2d\eta^{-1}$
so that
\begin{equation}
\label{Nlarge}
    (1-\eta)^{-\frac{N}{2d}}>2\, .
\end{equation}

\subsection{Errors analysis for the output \texorpdfstring{$\hat{c}_j$}{} of the Weiss algorithm}\label{fft error}     

Recall that $|b|\le 1-\eta$ on $\T$. As $b$ is a Laurent polynomial of degree $2d$, this implies that $|b|<1$ 
in an annulus about $\T$ of width comparable with $\eta/d$. 
This leads to the following Proposition.

Let $r>1$ be such that
\begin{equation}
    r^{2d}=(1-\eta)^{-1}.
\end{equation}
and note that $r^N>2$ by \cref{Nlarge}. Define
\begin{equation}
    A_r:=\{z: r^{-1}<|z|<r\}\ .
\end{equation} 
\begin{prop}\label{log1bbound}
    The function $\log\sqrt{1-|b|^2}$ has an analytic extension to  $A_r$ which is pointwise bounded by 
$\frac 12 \abs{\log{\eta}}$.\end{prop}
\begin{proof}
By the Schwartz reflection principle across $\T$, it suffices to present the analytic continuation on the inner annulus
\begin{equation}
    \tilde{A}_r=\{z: r^{-1}< |z|< 1\}
\end{equation}
and establish the claimed estimate there.
We have
\[\log\sqrt{1-|b|^2}=\frac 12 \log(1-bb^*)\, ,\]
where the right-hand-side has an analytic extension to $\tilde{A}_r$ provided $\log(1-bb^*)$ has no zero in $\tilde{A}_r$.
As the polynomials $z^db(z)$ and $z^db^*(z)$ are in $H^2(\D)$, we conclude by the maximum principle 
\begin{equation}\label{eq:bd_b_annulus}
\abs{b(z)}\le (1-\eta)r^{d}, \ \abs{b^*(z)}\le (1-\eta)r^d
\end{equation}
for each $z\in \tilde{A}_r$. Moreover, 
\[\abs{b(z) b^*(z)}\le (1-\eta)^2 r^{2d}= 1-\eta\, .\]
It follows that $1-bb^*$ has no zero in $A_r$ and 
\[\abs{\frac 12 \log(1-b(z)b^*(z))}\le \frac 12 |\log \eta|\, ,\]
where the last estimate follows by applying the triangle inequality to the Taylor expansion of the analytic function $z \mapsto \log (1-z)$.
This completes the proof of the proposition.
\end{proof}

A function with bounded analytic extension to the annulus $A_r$ is well approximated by a discrete Fourier series (see e.g. 
\cite{TrefethenWeideman14}).
Define for $j\in \Z$
$$
r_j := \mathcal{F} (\log \sqrt{1-|b|^2}) (j) \, ,\ 
\hat{r}_j = \mathcal{F}^{N} (\log \sqrt{1-|b|^2})(j) \, ,
$$
\begin{prop}\label{rjprop}
    We have, for every $j\in \Z$,
   \begin{equation}\label{rjbound}
     \abs{r_j}\le \frac 12 r^{-|j|} \abs{\log\eta }\, .
   \end{equation} 
   We have, for every $0\le j\le \frac N2$,
     \begin{equation}\label{rjdifference}
       \abs{r_j-\hat{r}_j}\le  2r^{j-N}\abs{\log \eta} \, .
   \end{equation} 
\end{prop}
\begin{proof}
    We write the Fourier coefficient as a contour integral
    with a curve $\gamma$ in the annulus $A_r$ that is homotopic to the standard contour around $\T$ as follows:
    \[r_j=
    \frac 12\int_\T z^{-j}  \log(1-b(z)b^*(z))=
    \frac 12\int_\gamma z^{-j}  \log(1-b(z)b^*(z)) \frac {dz}{2\pi i z}\, .\]
     Passing to to a contour $\gamma$ that describes a circle about the origin of radius near $r$
    on the one hand and near $r^{-1}$ on the other hand, we obtain  \cref{rjbound} with \cref{log1bbound}.

Substituting the 
function $\log\sqrt{1-|b|^2}$ by its Fourier expansion 
in the definition of $\widehat{r}_j$, we obtain
\[\hat{r}_j=
\frac 1N \sum_{\ell=0}^{N-1}
z_\ell^{-j} \sum_{k\in \Z} r_k z_\ell^k\ .
\]
Interchanging the order of summation and using that
\begin{equation}
    \sum_{\ell=0}^{N-1}z_\ell^{-j+k} 
\end{equation}
is equal to $N$ if $k$ is of the form $j+Nk'$ with $k'\in \Z$ 
and equal to $0$ otherwise, we obtain
\begin{equation}\label{eq:FFT_to_std_coeffs_1}
    \hat{r}_j=\sum_{k'\in \mathbb{Z}} r_{j+Nk'}\, .
\end{equation}
Splitting the sum into positive $k$ and negative $k$
and using geometric decay of the summands, we obtain
    \begin{equation}
  \abs{\widehat{r}_j-r_j}\le \sum_{k\in \mathbb{Z},k\neq 0} |r_{j+Nk}|
  \leq  \frac{r^{-j+N} + r^{j-N}}{2} \sum_{k=0}^\infty  r^{-Nk} |\log \eta | \, .%\le (|r_{j-N}| + |r_{j+N}|)\sum_{k=0}^\infty  r^{-Nk}\, .
\end{equation}
 The geometric sum is less than $2$
by assumption \cref{Nlarge} on $N$, and we dominate $r^{-j+N}$ by $r^{j-N}$ using the fact that $0 \leq j \leq \frac{N}{2}$. This shows \cref{rjdifference}
and completes the proof of the proposition.
\end{proof}

Define 
\begin{equation}
    G(z) = 2\sum_{\ell = 1}^{\infty} r_{-\ell} z^{-\ell} + r_0 \, ,\ 
    \hat{G}(z) := 2\sum_{\ell = 1}^{\frac{N}{2} }\hat{r}_{-\ell} z^{-\ell} + \hat{r}_0 \, .
\end{equation}
\begin{prop} \label{propgdifference}
    We have for every $z\in \T$,
    \begin{equation}\label{gdifference}
        \abs{ G(z) - \hat{G}(z)}\le 
5 \frac {r^{-N/2}}{1-r^{-1}} \abs{\log(\eta)}
    \, .
    \end{equation}
\end{prop}
\begin{proof}
    For any $z\in \T$, we have
$$\abs{ G(z) - \hat{G}(z)} = \abs{2\sum_{\ell = 1}^{ \frac{N}{2} }\hat{r}_{-\ell} z^{-\ell} + \hat{r}_0 -  2\sum_{\ell = 1}^{\infty} r_{-\ell} z^{-\ell} - r_0 }
\le 2\sum_{j=0}^{\frac N 2} \abs{ r_j-\hat{r}_j} +\sum_{j=\frac N2 +1}^\infty |r_j|\, .
$$
With \cref{rjprop}, we bound the last display by
\begin{equation}
    4\sum_{j=0}^{\frac N 2} r^{j-N}\abs{\log(\eta)} 
  +\frac 12\sum_{j=\frac N2 +1}^\infty  r^{-j} \abs{\log(\eta)}
  \le \left( 4 +\frac 12\right) 
  \frac {r^{-N/2}}{1-r^{-1}} \abs{\log(\eta)}
    \, .
\end{equation}
This implies \cref{gdifference} and completes the proof of the proposition.

\end{proof}

We continue with a similar analysis of the discrete
Fourier coefficients of the function $ba^{-1}$. We use analytic
continuation to the annulus $A_r$ with the same radius $r$ as before.
\begin{prop}\label{annulusbound}

The function $ba^{-1}$ extends analytically to the annulus $A_r$
and satisfies for all $z\in A_r$ the bound
\begin{equation}
    |b(z)a^{-1}(z)|\le r^d \eta^{-1} \, .
\end{equation}
\end{prop}
\begin{proof}
    As the function mapping $z$ to $z^{-d} b(z)a^{-1}(z)$
    is in $H^2(\mathbb{D}^*)$, we apply the maximum principle to estimate for all $\abs{z}\ge 1$
$$\abs{z^{-d}{b(z)}{a^{-1}(z)}} \leq 
 \|ba^{-1}\|_{L^\infty(\T)}\le (1-\eta) \eta^{-\frac 12}\, ,$$
 where we used that on $\T$ we have that $b$ is bounded above by $1-\eta$  and, using $0<\eta < 1 $,
 \begin{equation}\label{lowerabound}
      |a|^2=1-|b|^2\ge 2\eta-\eta^2\ge \eta\, .
 \end{equation}

 Hence, if $1\le \abs{z}\le r$,
 \begin{equation}\label{outerannulus}
     \abs{{b(z)}{a^{-1}(z)}} \leq 
 r^d(1-\eta)\eta^{-\frac 12}\, .
 \end{equation}
Since $a^* \in H^2(\D)$, we obtain by the maximum principle,
for each $r^{-1}\le \abs{z}\le 1$, 
\begin{equation}\label{eq:bd_a_annulus}
\abs{a^*(z)}\le 1\, .
\end{equation}
Using $aa^*+bb^*=1$, \cref{eq:bd_b_annulus} and finally \cref{eq:bd_a_annulus}, we obtain
$$\abs{b(z)a^{-1}(z)}=\abs{b(z)}\abs{a(z)^*+ b^*(z)b(z)a^{-1}(z)} \le r^d(1-\eta)(1+r^d(1-\eta)\abs{b(z)a^{-1}(z)})$$
Moving the term with $ba^{-1}$ from the right-hand-side to the left-hand side and using the definition of $r$, we obtain
$$(1-(1-\eta)) \abs{b(z)a^{-1}(z)} \le r^d(1-\eta)\, ,$$
or rather,
\begin{equation}\label{innerannulus}
    \abs{b(z)a^{-1}(z)} \le r^d(1-\eta)\eta^{-1}\, .
\end{equation}
The proposition now follows from \cref{outerannulus} and
\cref{innerannulus}  and $0<\eta< 1$. 
\end{proof}   

We use the last proposition to estimate the coefficients
\begin{equation}\label{eqn:coefficients_def}
    c_j := \mathcal{F} \left(\frac{b}{a} \right)(j),\ c_j':= \mathcal{F}^N \left(\frac{b}{a} \right)(j)\, .
\end{equation}
\begin{prop}\label{cjdifference}
We have for every $j\in \Z$,
\begin{equation}
\label{cjbound}
    \abs{c_j}\le r^{d-|j|}\eta^{-1} \, .
\end{equation}
We have for every $0\le j\le d$
\begin{equation}
    |c_j-c_j'|\le {4r^{d+j-N}}\eta^{-1}
\end{equation}
\end{prop}

\begin{proof}
We have for any contour $\gamma$ in $A_r$  homotopic to the standard contour around $\T$,
\begin{equation}\label{cjcontour}
c_j=\int_\T z^{-j}b(z)a(z)^{-1}=\int_\gamma 
z^{-j}b(z)a(z)^{-1} \frac{dz}{2\pi i z}  \, .  
\end{equation}
With \cref{annulusbound} and $\gamma$ tracing
circles close to radius $r$ and radius $r^{-1}$ about the origin, 
we obtain the bound
\[\abs{c_j}\le r^{d-|j|}\eta^{-1} \, .\]
We have for each $j$ with the Fourier inversion formula,
similarly as in \cref{eq:FFT_to_std_coeffs_1}
\[c_j'=\frac 1N \sum_{\ell=0}^{N-1}
z_\ell^{-j} b(z_l) a(z_l)^{-1}=
\frac 1N \sum_{\ell=0}^{N-1}
z_\ell^{-j} \sum_{k\in \Z} c_k z_\ell^k
=\sum_{k'\in Z} c_{j+Nk'}\, .\]
Hence
\begin{equation}\label{cjcjprime}
  \abs{c_j'-c_j}\le \sum_{k\in Z,k\neq 0} |c_{j+Nk}|  
\le (r^{d-|j+N|}+r^{d-|j-N|})\sum_{k=0}^\infty  r^{-{Nk}} \eta^{-1}
\le  4r^{d+j-N}\eta^{-1}\, ,
\end{equation}
where the geometric sum is bounded by $2$ and we have used
\cref{cjbound}.
This completes the proof of the proposition.

\end{proof}
Now recall $a=e^G$ and define  $\hat{a} := e^{\hat{G}}$ to be the approximation to $a$ given by the Weiss \cref{alg:weiss_alg}.
Define
\begin{equation}\label{eqn:coefficients_def_2}
    \hat{c}_j:= \mathcal{F}^N \left(\frac{b}{\hat{a}} \right)(j)\, .
\end{equation}

\begin{prop}\label{cchatprop}
We have for every $0\le j\le d$
\begin{equation}\label{eq:c-cNhat}
        |c_j' -\hat{c}_j| \leq \eta^{-\frac 12} \|1 -e^{G-\hat{G}}\|_{L^\infty(\T)}
        \, .
    \end{equation}
\end{prop}
\begin{proof}
We estimate
$$\abs{c_j' -\hat{c}_j}= \abs{\frac{1}{N} \sum_{\ell=0}^{N-1} z_\ell^{-j} b(z_\ell)(a^{-1}(z_\ell) -\hat{a}^{-1}(z_\ell))
}
\leq \|{b }({a^{-1} }- {\hat{a}^{-1} })\|_{L^\infty(\T)}
 \leq  \eta^{-\frac 12} \|1 -a \hat{a} ^{-1}\|_{L^\infty(\T)}\, ,
$$
where we used $|b|\le 1$ and $|a|\ge \eta^{\frac 12}$ as in \cref{lowerabound}.
Expressing $a$ and $\hat{a}$ by $G$ and $\hat{G}$ proves the proposition.
\end{proof}
\begin{prop}\label{thm:N_estimate}
    Let $0<\epsilon'<\frac 12$ and assume
    \begin{equation}
        N\ge \frac {8d}{\eta}\log\left(\frac {48d} {\eta^2 \epsilon'}\right)\, .
    \end{equation}
    Then we have for every $0\le j\le d$
\begin{equation}
    \abs{c_j-\hat{c}_j}\le \epsilon'\, .
\end{equation}
    
\end{prop}
\begin{proof}
Let $0\le j\le d$.
   We have with definition of $r$ 
   \begin{equation}\label{rto2d}
       r^{\frac {2d}\eta}=(1-\eta)^{-\frac 1 \eta}\ge (1+\eta)^{\frac 1{\eta}} \ge e^{\frac 12}\, ,
   \end{equation}
    where the last inequality is thanks to $\eta<\frac 12$.
    Hence
    $$r^{\frac N2} \ge \frac {48d} {\eta^2 \epsilon'}\, ,$$
 and with \cref{cjdifference} and $N>4d$
 $$\abs{c_j-c_j'}\le r^{d+j-N}\eta^{-1}\le r^{-\frac N2}\eta^{-1}\le \frac{\epsilon'}{48}\, .$$
    
    Moreover, using \cref{rto2d}
$$r-1\ge \log r \ge \frac \eta{4d}\, .$$
Hence for every $z\in \T$, with \cref{propgdifference},
and $r<\sqrt{2}$, %\ma{Is is clear that $r < 6/5$? Have $r^{2 d} = (1-\eta)^{-1}$. Taking $\eta = 1/2$ and $d=1$, we get $r^2 = 2$, or $r = \sqrt{2} > 6/5$. Replacing $5r <6$ by $r < \sqrt{2}$.},
$$|G(z)-\hat{G}(z)|\le 
5 \frac {r^{-N/2}}{1-r^{-1}} \abs{\log(\eta)}
 \le \frac {5 \sqrt{2}}{r-1} \frac {\eta^2 \epsilon'}{48d } |\log\eta| 
 \leq \frac{5 \sqrt{2}}{12} \eta \epsilon' |\log \eta| \leq  \frac {\eta^{\frac 12} \epsilon'} 2\, ,$$
 where for the last inequality we use the fact that $\eta^{\frac 1 2} |\log \eta| \leq \frac{2}{e}$.

   We estimate for $0\le j\le d$, 
   \[\abs{c_j-\hat{c}_j}\le
   \abs{c_j-{c}_j'}+
   \abs{c_j'-\hat{c}_j}\le \frac{\epsilon'}{48}+\eta^{-\frac 12} \abs{1-e^{\frac{\eta^{\frac 12}\epsilon'}2}}\le \epsilon' \, ,\]
   where in the before last inequality we used \cref{cchatprop} along with the inequality $|1-e^z| \leq |1-e^{|z|}|$, which follows from the triangle inequality applied to the Taylor expansion of the function $1-e^z$.
This completes the proof of the proposition
\end{proof}

\subsection{Error analysis for solving the linear system}\label{sec:property_Xi}

%In this subsection, we present some properties of the linear system described in \cref{eqn:numerical_linear_system}., which later play an important role in deriving an estimate of the potential inaccuracy in the solution after approximation. 

Let $\Xi_k$ and $\hat{\Xi}_k$ be the Hankel matrices with $(c_k , c_{k+1}, \cdots, c_d)^\top$ and $(\hat{c}_k , \hat{c}_{k+1}, \cdots, \hat{c}_d)^\top$ as their first column respectively and with zeros below the secondary diagonal. Also denote
$$\Id + M_k:=\begin{pmatrix}
    I & - \Xi_k \\
    -  \Xi_k & I 
\end{pmatrix}\, ,\quad \Id + \hat{M}_k:=\begin{pmatrix}
    I & - \hat{\Xi}_k \\
    -  \hat{\Xi}_k & I 
\end{pmatrix} \, .$$

We denote the $\ell_2$ norm of a vector $x = (x_0, x_1, \ldots, x_{m-1})\in\CC^m$ as
\begin{equation}
\norm{x}: = \sqrt{{\sum_{k=0}^{m-1}|x_k|^2}}.
\end{equation}
For a matrix $A\in\CC^{m\times m}$, we denote the operator norm induced by the vector $\ell_2$ norms as
\begin{equation}
\norm{A} : =\sup_{x \neq 0}  \frac{\norm{A x} }{\norm{x}}.
\end{equation}

\begin{prop}\label{lm:cond_Hk}
$\Xi_k$ has pure imaginary coefficients and
\begin{equation}\label{upper bound on Xi norm}
    \| \Xi_k \| \leq \norm{\frac{b}{a}}_\infty \, .
\end{equation}
Also the matrices $\Id + M_k$ and $\Id + \hat{M}_k$ are invertible,
\begin{equation}\label{inverting id plus M k}
    \norm{\left(\Id + M_k\right)^{-1}}\leq 1, \quad  \norm{\Id + M_k}\leq \sqrt{1+ \norm{\Xi_k}^2} \, ,
\end{equation}
and
\begin{equation}\label{perturbed inversion}
    \norm{\left(\Id + \hat{M}_k\right)^{-1}} \leq 1 \, .
\end{equation}
\end{prop}
\begin{proof}
To show that $\Xi_k$ has pure imaginary coefficients is equivalent to showing that the Fourier coefficients of $\frac{b}{a}$ are pure imaginary, which follows from \cref{lem:imaginary_Fourier_coeffs}. %To see this, note that $a$ has real Fourier coefficients by the multilinear expansion \cite[(5.9) and (5.10)]{AlexisMnatsakanyanThiele2023}. Then, we have $\min_{z\in \T} \abs{a(z)}\geq \sqrt{1-(1-\eta)^2}\geq \sqrt{\eta}$, so the uniformly convergent series 
%\[
%a^{-1} = \frac{1}{1-(1-a)} = \sum\limits_{\ell=0}^{\infty} (1-a)^{\ell}
%\]
%also has real Fourier coefficients. We conclude, the Fourier coefficients of $\frac{b}{a}$ are all pure imaginary.

Since we truncate the space $\mathcal{H}_k$ to $\mathcal{H}^{d-k}_k$, $\Xi_k$ is exactly the matrix representation of the operator 
$$
( x_{0}, ,\dots, x_{-(d-k)} ) \mapsto \mathcal{P}_{[k, d]}\mathcal{F}^{-1} \left( \frac{b(z)}{a(z)} \sum_{j=-(d-k)}^{0} x_jz^j \right) \, ,
$$
where $\mathcal{F}^{-1}$
is the inverse Fourier transform taking a function on $\T$ to the (infinite) vector of its Fourier coefficients, indexed by frequency, and $\mathcal{P}_{[k, d]}$ is the projection operator onto the subspace of vectors whose nonzero entries only appear between indices $k$ and $d$. By Parseval's identity we have for any $x=(x_0, x_1, \cdots, x_d)$ 
$$
\left\| \mathcal{P}_{[k, d]} \mathcal{F}^{-1} \left( \frac{b(z)}{a(z)} \sum_{j=-(d-k)}^{0} x_jz^j \right) \Big|_{[k,d]} \right\| \leq \left\| \mathcal{F}^{-1} \left( \frac{b(z)}{a(z)} \sum_{j=-(d-k)}^{0} x_jz^j \right) \right\| 
$$
$$
= \left\| \frac{b(z)}{a(z)} \sum_{j=-(d-k)}^{0} x_jz^j  \right\|_{2} \leq \norm{ \frac{b}{a}}_\infty \norm{\sum_{j=-(d-k)}^{0} x_jz^j}_{2} = \norm{ \frac{b}{a}}_\infty \cdot \norm{x } \, .
$$
This proves \cref{upper bound on Xi norm}.
%Notice that $\Xi_1$ can be derived by removing the first column and the last row of $\Xi_0$, leading to $\norm{\Xi_1}\leq \norm{\frac{b}{a}}_{\infty}$. Extending this analysis, we examine the relationship between $\Xi_k$ and $\Xi_0$ and subsequently obtain that $\norm{\Xi_k}\leq \norm{\frac{b}{a}}_{\infty}$ for any $0\leq k\leq d$. 

Since $\Xi_k^* = -\Xi_k$, the eigenvalues of $\Xi_k$ are all pure imaginary. We denote them $\I \lambda_0, \I \lambda_1, \cdots, \I \lambda_{d-k}$. Here, $\lambda_j \in \RR$ for $0\leq j\leq d-k$. Then, as
$$\begin{pmatrix}
    I & -  \Xi_k \\
    - \Xi_k & I 
\end{pmatrix}\begin{pmatrix}
    I & -\Xi_k \\
    - \Xi_k & I 
\end{pmatrix}^{*} = \begin{pmatrix}
    I - (\Xi_k)^2 & 0 \\
    0 & I -(\Xi_k)^2
\end{pmatrix} \, ,$$
we see that the singular values of $\Id + M_k$ are
$$\sqrt{1+\lambda_0^2}, \sqrt{1+\lambda_1^2}, \cdots, \sqrt{1+\lambda_{d-k}^2} \, .$$

Therefore, $\Id + M_k$ is non-singular, 
$$\norm{(\Id + M_k)^{-1}}\leq 1 \, ,$$
and
$$\norm{\Id + M_k}\leq \sqrt{1+ \norm{\Xi_k}^2} \, . $$

To show that $\Id + \hat{M}_k$ is invertible and $\norm{\left(\Id + \hat{M}_k\right)^{-1}}\leq 1$, we only need to show $\hat{\Xi}_k^*= -\hat{\Xi}_k$, so that we may apply the same argument as that for $\Id + M_k$. Recall that $\hat{\Xi}_k$ is the Hankel matrix with $(\hat{c}_k , \hat{c}_{k+1}, \cdots, \hat{c}_d)^\top$ as it first column. To show that $\hat{\Xi}_k$ is anti-Hermitian, we examine the Weiss algorithm, which outputs $(\hat{c}_k , \hat{c}_{k+1}, \cdots, \hat{c}_d)^\top$. We have the following relation 
\begin{equation*}
    \hat{c}_j = \frac{1}{N} \sum_{\ell=0}^{N-1} z_{\ell}^{-j} \hat{u} (z_{\ell}), \quad \forall k\leq j \leq d,
\end{equation*}
with 
\begin{equation*}
    \hat{u} (z) := b(z)e^{-\hat{G}(z)}= b(z) e^{-\hat{r}_0 - 2\sum_{\ell = 1}^{\frac{N}{2}} \hat{r}_{-\ell} z^{-\ell}}.
\end{equation*}
Note that $\hat{r}_\ell = \mathcal{F}^{N} \left(\log(1-\abs{b(z)}^2)\right)(\ell) $. So $\hat{r}_\ell\in \RR$ for any $\ell$. Together with that the Fourier coefficients of $b$ are all pure imaginary, we know that the Fourier coefficients of $\hat{u}$ are all pure imaginary. Due to the relation \cite[Equations (2.2) and (2.8)]{TrefethenWeideman14}, we know that for any $ k\leq j \leq d$,
\begin{equation}
    \hat{c}_j= \mathcal{F}^N(\hat{u})(j) = \sum_{\ell = -\infty}^{\infty} \mathcal{F} (\hat{u})(j+\ell N), 
\end{equation} 
 we get that $\hat{c}_j$ are all pure imaginary. Thus $\hat{\Xi}_k$ is anti-Hermitian, which completes the proof.

% We turn to \eqref{perturbed inversion}. We write
% $$
% \Id + \hat{M}_k = (\Id + M_k)\left( \Id + (\Id+M_k)^{-1}(\hat{M}_k-M_k) \right) \, .
% $$
% By \eqref{inverting id plus M k} and Proposition \ref{thm:sufficient_condition_N}, we have
% $$
% \| (\Id+M_k)^{-1}(\hat{M}_k-M_k) \| \leq (d+1)\epsilon' < 1 \, ,
% $$
% so that $\Id + \hat{M}_k$ is invertible. Furthermore, we have
% $$
% (\Id + \hat{M}_k)^{-1} =  \left( \Id + (\Id+M_k)^{-1}(\hat{M}_k-M_k) \right)^{-1} (\Id + M_k)^{-1}
% $$
% $$
% = \left( \sum_{j=0}^{\infty} ((\Id+M_k)^{-1}(\hat{M}_k-M_k))^{j} \right) (\Id + M_k)^{-1} \, .
% $$
% And the estimate
% $$
% \| (\Id + \hat{M}_k)^{-1}\| \leq \sum_{j=0}^{\infty} \|((\Id+M_k)^{-1}(\hat{M}_k-M_k))^{j} \| \leq \sum_{j=0}^{\infty} \|\hat{M}_k-M_k \|^{j}
% $$
% $$
% \leq \sum_{j=0}^\infty (d+1)^j\epsilon'^j \leq \frac{1}{1-(d+1)\epsilon'}\leq 2 \, ,
% $$
% follows.
\end{proof}

%\begin{remark}[A variant of \cref{alg:RHW_alg}]
%    Instead of solving for phase factor independently, we may peel $(A_0, B_0)$ using Layer stripping technique which is first come up with in \cite{sylvesterwinebrenner} and later applied to our QSP problem \cite[Theorem 8]{AlexisMnatsakanyanThiele2023}. It only requires to solve linear system once. 
%\end{remark}
Remark that the inverse of matrix $\Id + M_k$ can be explicitly written as 
\begin{equation*}
    \begin{pmatrix}
        \left(I - (\Xi_k)^2\right)^{-1} & \left(I - (\Xi_k)^2\right)^{-1}\Xi_k\\
        \left(I - (\Xi_k)^2\right)^{-1}\Xi_k & \left(I - (\Xi_k)^2\right)^{-1}
    \end{pmatrix}
\end{equation*}
Then the solution to \cref{eqn:numerical_linear_system} is given by
\begin{equation}
    \begin{pmatrix}
    \va_k\\
    \vb_k
\end{pmatrix} = \begin{pmatrix}
    \left(I - (\Xi_k)^2\right)^{-1}\ve_0\\
    \left(I - (\Xi_k)^2\right)^{-1} \Xi_k\ve_0
\end{pmatrix}.
\end{equation}
So instead of solving a large linear system, we can just solve a smaller linear system, 
\begin{equation}\label{eqn:alter_linear_system}
    \left(I -(\Xi_k)^2\right) (\va_k, \vb_k) = (\ve_0, \Xi_k\ve_0),
\end{equation}
where  $(a_k, b_k)$ is a matrix of size $n\times 2$.

\subsection{Proof of \texorpdfstring{\cref{thm:sufficient_condition_N}}{Theorem 8}}\label{proof of thm 8}
We finish the proof of \cref{thm:sufficient_condition_N} applying \cref{thm:N_estimate}  and \cref{lm:cond_Hk}. Let $0<\epsilon \leq 1$. Given $\epsilon$, define $\epsilon':= \frac{\epsilon\eta^2}{12 d}$, and choose $N$ as in \cref{thm:sufficient_condition_N}, or equivalently as in \cref{thm:N_estimate}. 

Denote
$$
a_{k,0} :=\begin{pmatrix}
        \ve_0^\top,
        \bvec{0}^\top
    \end{pmatrix}(\Id+M_k)^{-1} \begin{pmatrix}
        \ve_0\\
        \bvec{0}
    \end{pmatrix},
$$
$$
\hat{a}_{k,0} := \begin{pmatrix}
        \ve_0^\top,
        \bvec{0}^\top
    \end{pmatrix}(\Id+\hat{M}_k)^{-1} \begin{pmatrix}
        \ve_0\\
        \bvec{0}
    \end{pmatrix} \, .
$$
Then, we have
$$
|a_{k,0} - \hat{a}_{k,0}| = \left| \begin{pmatrix}
        \ve_0^\top,
        \bvec{0}^\top
    \end{pmatrix} \left[ (\Id+M_k)^{-1} - (\Id+\hat{M}_k)^{-1} \right] \begin{pmatrix}
        \ve_0\\
        \bvec{0}
    \end{pmatrix} \right|
$$
$$
=\left| \begin{pmatrix}
        \ve_0^\top,
        \bvec{0}^\top
    \end{pmatrix} (\Id+M_k)^{-1} (\hat{M}_k-M_k)(\Id+\hat{M}_k)^{-1}
\begin{pmatrix}
        \ve_0\\
        \bvec{0}
    \end{pmatrix}\right|
$$
$$
\leq \left\| (\Id+M_k)^{-1} (\hat{M}_k-M_k)(\Id+\hat{M}_k)^{-1}
\right\| \, .
$$
By \cref{lm:cond_Hk}, we know that
$$
\left\| \left(\Id+M_k\right)^{-1} \right\| \leq 1 \text{ and } \norm{\left(\Id + \hat{M}_k\right)^{-1}}\leq 1 \, .
$$
On the other hand, by \cref{thm:N_estimate}, we have 
\begin{equation*}
    \abs{c_j -\hat{c}_j} \leq \epsilon ', \quad \forall 0\leq j \leq d.
\end{equation*}
Thus 
\begin{equation}\label{eqn:difference_in_Xi}
\begin{split}
    \|\hat{M}_k-M_k\| &= 
 \left\|\begin{pmatrix}
    0 & -(\hat{\Xi}_k - \Xi_k) \\
   -(\hat{\Xi}_k - \Xi_k) & 0 
\end{pmatrix} \right\|\leq \left( 2 \sum_{\substack{0\leq l,j\\
l+j\leq d-k}} |c_{j+l+k}-\hat{c}_{j+l+k}|^2 \right)^{1/2}\\
& \leq \left(4(d-k+1)(d-k+2)(\epsilon ')^2 \right)^{1/2} \leq 2 (d+2) \epsilon' \leq 6 d \epsilon' \, .
\end{split}
\end{equation}
The first equality is due to that the lower right triangular part of $\Xi_k$ and $\hat{\Xi}_k$ are all zero.

Plugging this estimates above, we get
\begin{equation}\label{eq:estimate on difference of a k 0}
    |a_{k,0} - \hat{a}_{k,0}| \leq 6 d \epsilon' \, .
\end{equation}

Similarly, let
$$
b_{k,0} :=\begin{pmatrix}
        \bvec{0}^\top,
        \ve_0^\top
    \end{pmatrix}(\Id+M_k)^{-1} \begin{pmatrix}
        \ve_0\\
        \bvec{0}
    \end{pmatrix} , \text{ and }
\hat{b}_{k,0} := \begin{pmatrix}
        \bvec{0}^\top,
        \ve_0^\top
    \end{pmatrix} (\Id+\hat{M}_k)^{-1} \begin{pmatrix}
        \ve_0\\
        \bvec{0}
    \end{pmatrix} \, .
$$
Analogous to \cref{eq:estimate on difference of a k 0}, we can show
\begin{equation}\label{eq:estimate on difference of b k 0}
    |b_{k,0} - \hat{b}_{k,0}| < 6 d \epsilon' \, .
\end{equation}
Furthermore, by \cref{lm:cond_Hk}, we have
\begin{equation}\label{upper bound on b single coefficient}
     \abs{ b_{k,0}}\leq \norm{(\Id+M_k)^{-1}}\leq 1 \, .
\end{equation}
Letting $\lambda_{\text{min}} (T)$ denote the smallest eigenvalue of a matrix $T$, then also by \cref{eqn:alter_linear_system} and again by \cref{lm:cond_Hk}, we have 
 then also by \cref{eqn:alter_linear_system} and again by \cref{lm:cond_Hk}, we have 
\begin{equation}\label{eq: lower bound on a k 0}
    \begin{split}
        \abs{ a_{k,0}} &= \abs{\ve_0^\top \left(I -(\Xi_k)^2\right)^{-1}\ve_0} = \ve_0^\top \left(I -(\Xi_k)^2\right)^{-1}\ve_0\\
        & \geq \lambda_{\min} \left(\left(I -(\Xi_k)^2\right)^{-1}\right) = \frac{1}{1+\norm{\Xi_k}^2}
\geq \frac{1}{1+\norm{\frac{b}{a}}_{\infty}^2}\\
& \geq \frac{1}{1+ \frac{(1-\eta)^2}{1-(1-\eta)^2}}= 2\eta -\eta^2\geq \eta,
    \end{split}
\end{equation}
where the second equality is because $I -(\Xi_k)^2$ is a real symmetric positive definite matrix.
%\begin{equation}\label{eq: lower bound on a k 0}
%    \begin{split}
%        \abs{ a_{k,0}} &= \abs{\ve_0^\top \left(I -(\Xi_k)^2\right)^{-1}\ve_0} = \ve_0^\top \left(I -(\Xi_k)^2\right)^{-1}\ve_0\\
%        & \geq \lambda_{\min} \left(\left(I -(\Xi_k)^2\right)^{-1}\right) = \frac{1}{1+\norm{\Xi_k}^2}
%\geq \frac{1}{1+\norm{\frac{b}{a}}_{\infty}^2}\\
%& \geq \frac{1}{1+ \frac{(1-\eta)^2}{1-(1-\eta)^2}}= 2\eta -\eta^2\geq \eta,
%    \end{split}
%\end{equation}
 
Applying the estimates \cref{eq: lower bound on a k 0}, \cref{eq:estimate on difference of a k 0}, \cref{eq:estimate on difference of b k 0} and \cref{upper bound on b single coefficient}, we finish the proof. We have
\begin{equation}\label{eqn:error_in_computing_psi}
    \begin{split}
        |\psi_k -\hat{\psi}_k| &= \abs{ \arctan\left(-\I\frac{ b_{k,0}}{a_{k,0}}\right) - \arctan\left(-\I\frac{\hat{b}_{k,0}}{\hat{a}_{k,0}} \right)}\\
        &\leq \abs{ \frac{ b_{k,0}}{a_{k,0}} - \frac{\hat{b}_{k,0}}{\hat{a}_{k,0}} }= \abs{\frac{b_{k,0}-\hat{b}_{k,0}}{\hat{a}_{k,0} } - \frac{b_{k,0} (a_{k,0}-\hat{a}_{k,0})}{a_{k,0}\hat{a}_{k,0}}}\\
        &\leq \abs{\frac{b_{k,0}-\hat{b}_{k,0}}{\hat{a}_{k,0} }} +\abs{ \frac{b_{k,0} (a_{k,0}-\hat{a}_{k,0})}{a_{k,0}\hat{a}_{k,0}}}\\
        &\leq \frac{6 d \epsilon' }{\eta- 6 d \epsilon'} + \frac{6 d \epsilon' }{\eta\left(\eta - 6 d \epsilon'\right)}=\frac{6 d \epsilon' }{\eta\left(\eta -\frac{\epsilon \eta^2}{2}\right)}(1+\eta)\\
        & \leq  \frac{6 d \epsilon' }{\eta\left(\eta - \frac{\eta}{4}\right)}\left(1 +\frac{1}{2}\right)= 6 d\epsilon' \frac{2}{\eta^2} = \epsilon
    \end{split}
\end{equation}
by the choice of $\epsilon'= \frac{\epsilon\eta^2}{12 d}$.

\subsection{Computational cost}\label{computational cost proof}

\REV{The FFT for a sequence of length $N$ requires $\mathcal{O}(N \log(N))$ operations. From \cref{thm:sufficient_condition_N}, we set $N=\mathcal{O}(\frac{d}{\eta} \log(\frac{d}{\eta\epsilon}))$, and the number of operations  of FFT is $\mathcal{O}\left(\frac{d}{\eta} \log(\frac{d}{\eta\epsilon})\left(\log(\frac{d}{\eta})+\operatorname{loglog}(\frac{d}{\eta\epsilon})\right)\right)=\mathcal{O}\left(\frac{d}{\eta} \log^2(\frac{d}{\eta\epsilon})\right)$.} The cost for solving a $d\times d$ linear system is $\mathcal{O}(d^3)$. The overall  computational cost of using \cref{alg:RHW_alg} to determine a single phase factor is, then \REV{$\mathcal{O}(d^3 + \frac{d}{\eta} \log^2(\frac{d}{\eta\epsilon}))$}. To compute all the phase factors associated to a polynomial $f$ of degree $2d$, the FFT requires still the same number of operations, but we now solve $(d+1)$ many linear systems to compute phase factors, which then takes $\mathcal{O}(d^4)$ many operations. Thus we are left with a total operational cost of \REV{$\mathcal{O}(d^4 + \frac{d}{\eta} \log^2(\frac{d}{\eta\epsilon}))$} to compute all the phase factors of $f$. This completes the analysis of the computational cost as outlined in \cref{thm:main_alg}.

%FFT requires $\mathcal{O}(N \log(N))$ operations, that is, $\wt{\mathcal{O}}(\frac{d}{\eta} \log(\frac{d}{\eta\epsilon}))$ operations. Here, the notation $\wt{\mathcal{O}}(\cdot)$ omits the double logarithmic terms.  The cost for solving a $d\times d$ linear system is $\mathcal{O}(d^3)$. The overall  computational cost of using \cref{alg:RHW_alg} to determine a single phase factor is, then, $\wt{\mathcal{O}}(d^3 + \frac{d}{\eta} \log(\frac{d}{\eta\epsilon}))$. To compute all the phase factors associated to a polynomial $f$ of degree $2d$, the FFT requires still the same number of operations, but we now solve $(d+1)$ many linear systems to compute phase factors, which then takes $\mathcal{O}(d^4)$ many operations. Thus we are left with a total operational cost of $\wt{\mathcal{O}}(d^4 + \frac{d}{\eta} \log(\frac{d}{\eta\epsilon}))$ to compute all the phase factors of $f$. This completes the analysis of the computational cost as outlined in \cref{thm:main_alg}.

\subsection{Bit requirement}\label{bit requirement proof}
In this subsection, we discuss the bit requirement of \cref{alg:RHW_alg}. Since our algorithm mainly relies on FFT and on solving well conditioned linear systems, the analysis of the bit-requirement is standard, and this is included mainly for completeness.

We assume radix-2, precision-$p$ arithmetic, with rounding unit $u = 2^{-p}$. According to \cite{Ramos}, we have the following estimate.

\begin{OldTheorem}[\cite{Ramos}]\label{thm:Ramos}
    With the assumptions above, let $m\in \N$ and let $\bvec{Z}$ be a vector of exact FFT coefficients in $C^{2^m}$. Also let $\hat{\bvec{Z}}$ be the vector of numerically computed FFT coefficients with rounding unit $u=2^{-p}$. Then,
    \begin{equation}
        \frac{\norm{\hat{\bvec{Z}}-\bvec{Z}}}{\norm{\bvec{Z}}}\leq ((4+\sqrt{2})m-4)u + \mathcal{O}(u^2).
    \end{equation}
\end{OldTheorem}
In the Weiss algorithm, we need to apply two FFTs to obtain the Fourier coefficients of $R$ and $b/a$. Due to the Parseval theorem for Discrete Fourier Transform, the 2-norm of the vector of exact FFT coefficients for $R$ is bounded by
\begin{equation*}
    \|R\|_\infty=-\log\left(\sqrt{1-\norm{b}_{\infty}^2}\right)\leq \frac{1}{2}\log\left(\frac{1}{1-(1-\eta)^2}\right)\leq \frac{1}{2} \log\left(\frac{1}{\eta}\right),
\end{equation*}
while that of $b/a$ is bounded by 
\begin{equation*}
   \norm{b/a}_{\infty}\leq \frac{1-\eta}{\sqrt{1-(1-\eta)^2}} \leq \frac{1}{\sqrt{\eta}}.
\end{equation*}
\cref{thm:Ramos} indicates the total absolute error of the Weiss algorithm is 
\begin{equation*}
    \Or \left(\log(N)\left(\norm{\frac{b}{a}}_{\infty}-\log\left(\sqrt{1-\norm{b}_{\infty}^2}\right) \right)u \right) = \Or\left(\frac{\log(N)}{\sqrt{\eta}}u \right).
\end{equation*}

Now we examine how the rounding error accumulates when solving the linear system  \cref{eqn:alter_linear_system}
%$$(\Id + M_k)\vx =\bvec{y} \, .$$
%The discussion in \cref{sec:property_Xi} implies that instead we need to solve
\begin{equation}
    \left(I -(\Xi_k)^2\right) (\va_k, \vb_k) = (\ve_0, \Xi_k\ve_0)\, ,
\end{equation}
for $\bvec{a}_k$ and $\bvec{b}_k$. Recall, that $I -(\Xi_k)^2$ is a real, symmetric, positive-definite matrix, and by \cref{upper bound on Xi norm}, we have the matrix estimate 
$$\norm{I -(\Xi_k)^2} \leq 1+ \norm{\frac{b}{a}}_{\infty}^2 \, .$$ 
We use it with the help of Cholesky factorization and a normwise error analysis of Wilkinson \cite{WilkinsonICM}.
\begin{OldTheorem}[\cite{Higham2002}, \cite{WilkinsonICM}]\label{thm:Wilkinson}
    Let $A \in \RR^{m\times m}$ be a symmetric, positive-definite matrix. Suppose that the Cholesky factorization produces a solution $\hat{x}$ to $Ax=b$. If $\max\left((3m+1)u , \frac{m(m+1)u}{1-(m+1)u}\right)<\frac{1}{2}$, then
    \begin{equation*}
        (A+E)\hat{x} = b, \quad \norm{E}\leq 4m(3m+1)u\norm{A}.
    \end{equation*}
\end{OldTheorem}
So bounding $\norm{E}$ for the linear system \cref{eqn:alter_linear_system}, the backward error of solving this single linear system \cref{eqn:alter_linear_system} is  
\begin{equation*}
     \Or\left(d^2(1+\norm{b/a}^2_{\infty})u\right) = \Or\left(\frac{d^2u}{\eta}\right) \, .
\end{equation*}

\REV{The following identity holds if both $A$ and $A+E$ are nonsingular:
\begin{equation*}
    (A+E)^{-1} - A^{-1} = (A+E)^{-1} E A^{-1}.
\end{equation*}
We assume that $u$ is sufficiently small and make use of the bound $\norm{(1 - (\Xi_k)^2)^{-1}} < 1$ for all $k$.}
 Then the forward error of solving this single linear system \cref{eqn:alter_linear_system} is still
\begin{equation*}
     \Or\left(\frac{d^2u}{\eta}\right) \, .
\end{equation*}

% Combining the analysis from  \cref{eqn:difference_in_Xi}, we know that the additional forward error of solving linear system \eqref{eqn:alter_linear_system} due to round-off is 
% \begin{equation*}
%     \Or\left(\frac{d^2u}{\eta^3}\right).
% \end{equation*}
% \ma{Why do we get $\eta^3$ in the denominator above? I don't see how this arises.}\ma{From discussion, use same logic as in equation below (139), and make sure additional forward error is sufficiently small.}
In the last step of the algorithm, we compute 
$$\psi_k = \arctan\left(-\I \frac{b_{k,0}}{a_{k,0}}\right) \, ,$$
where $a_{k,0}$ and $b_{k,0}$ are the first entries of $\va_k$ and $\vb_k$. The previous discussion implies that the error in computing $a_{k,0}$ and $b_{k,0}$ is $\Or\left(\frac{d^2u}{\eta}\right)$. 
Similar to the analysis in \cref{eqn:error_in_computing_psi}, the error of computing $\psi_k$ is
$$\Or \left(\frac{d^2u}{\eta^3}\right) \, .$$ 
Combining all the errors above, we know that the number of the bits required is  
$$\mathcal{O}\left( \operatorname{loglog}(N)+\log(d)+\log \left(\frac{1}{\eta}\right)+ \log \left(\frac{1}{\epsilon} \right) \right) \, .$$
Plugging in the estimate for $N$ from \cref{thm:sufficient_condition_N}, we find that the bit required by \cref{alg:RHW_alg} is \REV{$\mathcal{O}(\log(\frac{d}{\epsilon \eta}))$}, which completes the proof of the bit requirement part of \cref{thm:main_alg}.

\bibliographystyle{abbrvurl}
\bibliography{references}

\end{document}